\documentclass{article}
\usepackage{fullpage}

\usepackage{epsf}
\usepackage{fancyhdr}
\usepackage{graphics}
\usepackage{graphicx}
\usepackage{psfrag}
\usepackage{color}

\usepackage{amsthm}
\usepackage{amsfonts}
\usepackage{amsmath}
\usepackage{amssymb}
\usepackage{bbm}

\usepackage[utf8]{inputenc}
\usepackage[colorlinks=true, allcolors=blue]{hyperref}
\usepackage[square,sort,comma,numbers]{natbib}
\usepackage{appendix}

\usepackage{mathtools}
\usepackage{graphicx}
\usepackage{caption}
\usepackage{subcaption}
\newsavebox{\bigimage}
\usepackage{multirow}
\usepackage{float}

\usepackage{pifont}

\newcommand{\pvalue}{$p$-value\xspace}
\newcommand{\pvalues}{$p$-values\xspace}
\newcommand{\premt}{preordered martingale test\xspace}
\newcommand{\mst}{martingale Stouffer test\xspace}
\newcommand{\postmt}{adaptively ordered martingale test\xspace}
\newcommand{\imt}{interactively ordered martingale test\xspace}
\newcommand{\tset}{testing set\xspace}
\newcommand{\one}{\mathbbm{1}\xspace}

\usepackage[flushleft]{threeparttable}
\usepackage{footnote}
\makesavenoteenv{tabular}
\makesavenoteenv{table}

\usepackage[ruled,vlined]{algorithm2e}
\newtheorem{definition}{Definition}
\newtheorem{claim}{Claim}
\newtheorem{theorem}{Theorem}
\newtheorem{lemma}{Lemma}
\newtheorem{setting}{Setting}

\newtheorem{remark}{Remark}

\newcommand{\I}{\mathcal{I}}

\DeclareMathOperator*{\argmin}{argmin}
\DeclareMathOperator*{\argmax}{argmax}

\newlength{\widebarargwidth}
\newlength{\widebarargheight}
\newlength{\widebarargdepth}

\makeatletter
\long\def\@makecaption#1#2{
        \vskip 0.8ex
        \setbox\@tempboxa\hbox{\small {\bf #1:} #2}
        \parindent 1.5em  
        \dimen0=\hsize
        \advance\dimen0 by -3em
        \ifdim \wd\@tempboxa >\dimen0
                \hbox to \hsize{
                        \parindent 0em
                        \hfil 
                        \parbox{\dimen0}{\def\baselinestretch{0.96}\small
                                {\bf #1.} #2
                                } 
                        \hfil}
        \else \hbox to \hsize{\hfil \box\@tempboxa \hfil}
        \fi
        }
\makeatother

\begin{document}

\begin{center}

{\bf{\LARGE{Interactive Martingale Tests for the Global Null}}}

\vspace*{.2in}

{\large{
\begin{tabular}{cccc}
Boyan Duan &  Aaditya Ramdas & Sivaraman Balakrishnan & Larry Wasserman \\
\end{tabular}
\texttt{\{boyand,aramdas,siva,larry\}@stat.cmu.edu}
}}

\vspace*{.2in}

\begin{tabular}{c}
Department of Statistics and Data Science,\\
Carnegie Mellon University,  Pittsburgh, PA  15213
\end{tabular}

\vspace*{.2in}

\today

\vspace*{.2in}

\begin{abstract}
Global null testing is a classical problem going back about a century to Fisher’s and Stouffer’s combination tests. In this work, we present simple martingale analogs of these classical tests, which are applicable in two distinct settings: (a) the online setting in which there is a possibly infinite sequence of \pvalues, and (b) the batch setting, where one uses prior knowledge to preorder the hypotheses. Through theory and simulations, we demonstrate that our martingale variants have higher power than their classical counterparts even when the preordering is only weakly informative. Finally, using a recent idea of “masking” \pvalues, we develop a novel interactive test for the global null that can take advantage of covariates and repeated user guidance to create a data-adaptive ordering that achieves higher detection power against structured alternatives. 
\end{abstract}
\end{center}

\section{Introduction}

This paper proposes new martingale-based methods for testing the global null corresponding to hypotheses $\{H_i\}_{i\in \I}$ using a corresponding set of \pvalues $\{p_i\}_{i\in \I}$ and possibly other covariates $\{x_i\}_{i\in\I}$, where the index set $\I$ can be finite or countably infinite. Global null testing corresponds to testing if all individual hypotheses are truly nulls~(denoted as $H_i = 0$), against its complement:
$$
\mathcal{H_G}_0: H_i = 0 \text{ for all } i\in\I, \quad \mathcal{H_G}_1: H_i = 1 \text{ for at least one } i \in \I.
$$
As we review later in the introduction, this is a well-studied classical problem.
We consider two settings, the batch setting and the online setting, and our proposed framework applies to both settings:  
\begin{itemize}
    \item Batch setting: we have access to a fixed batch of $n$ hypotheses, thus ${\I = \{1,\ldots,n\}}$.
    \item Online setting: an unknown and potentially infinite number of hypotheses arrive sequentially in a stream, thus $\I = \{1,2,\ldots,k,\ldots\}$.
\end{itemize}
Most common global null tests involve a one-step operation, comparing a single statistic with a critical value derived from its null distribution. Observing that many classical tests effectively use a martingale-type test statistic, we propose novel martingale analogs of these tests that are inherently sequential (multi-step) in nature, and thus naturally apply in the online setting, or in the batch setting if an ordering can be created using prior knowledge and/or the data. Intriguingly, the ordering may also be created \emph{interactively}: this means that an analyst may adaptively create the ordering in a data-dependent manner if they adhere to a particular protocol of \emph{masking} and \emph{unmasking} (the definition is introduced later in equation~\eqref{default_decompse}). In order to understand why our interactive martingale tests have desirable properties (both controlling type-I errors and having higher power in structured settings), it is necessary to present them last, after having derived the vanilla non-interactive martingale global null tests, which are also novel in their own right. Specifically, for the purposes of progressively developing intuition, our 
treatment follows three steps of increasing complexity:
\begin{itemize}
    \item (Preordered setting, Section~\ref{sec:premt}) In the batch setting, the analyst employs \emph{prior} knowledge (data-independent) to preorder the hypotheses. In the online setting, an ordering of hypotheses is provided by nature. 
    \item (Data-adaptive ordering, Section~\ref{sec:postmt}) In the batch setting, the hypotheses are unordered, but an adaptive data-dependent ordering is created based on ``masked'' $p$-values. In the online setting, nature orders hypotheses, but the analyst discards some hypotheses from the ordering based on their masked $p$-values. Even though the data-adaptive and preordered settings proceed sequentially and handle the $p$-values one at a time, the analyst plays no role \emph{during} this sequential process, as all the rules for how to order the hypotheses are prespecified before the data is observed.
    \item (Interactive ordering, Section~\ref{sec:imt}). The utility of masking to enable interaction with a human is most compelling in the batch setting, where in addition to the unordered hypotheses, we suppose that the analyst also has additional side information in the form of covariates, and perhaps prior knowledge in the form of structural constraints on the non-null set. Using these, and any working models of their choice, the analyst interactively creates an ordering by initially observing only masked $p$-values, and progressively unmasking them one at a time. The analyst can update their prior knowledge and/or structural constraints and/or working model in the middle of the process (when only some hypotheses have been ordered and their $p$-values unmasked), thus intervening to change the rest of the ordering. It is important to note that \emph{even though an analyst is allowed to make subjective decisions at each step of the interaction, an algorithm can be deployed in place of the analyst}.
\end{itemize}

Since all our tests proceed sequentially in nature, accumulating evidence from one hypothesis at a time, the type-I error guarantee we achieve is that
\[
\mathbb{P}_0( \exists i \in \I: \text{the test stops and rejects } \mathcal{H_G}_0 \text{ after step } i) \leq \alpha,
\]
where $\mathbb{P}_0$ is the probability under the global null $\mathcal{H_G}_0$. They are judged based on their power,
\[
\mathbb{P}_1( \exists i \in \I: \text{the test stops and rejects } \mathcal{H_G}_0 \text{ after step } i),
\]
where $\mathbb{P}_1$ is the probability under some alternative in $\mathcal{H_G}_1$.
We remark that even though we formulate our tests in terms of a target type-I error level $\alpha$, there is an equivalent formulation in terms of creating a sequential ``always-valid'' $p$-value for the global null that is valid at any arbitrary stopping time. Section~\ref{sec:anytime} explicitly connects these two interpretations.

\subsection{Assumptions}

Instead of assuming that the marginal distribution of null $p$-values is exactly uniform, we relax it by allowing conservative $p$-values defined in two different ways. We either assume that (a) if the global null is true, all $p$-values are stochastically larger than uniform:
\begin{equation}
\label{cond:stoch_dom}
\text{If } \mathcal{H_G}_0 \text{ is true}, ~ \mathbb{P}(p_i \leq t) \leq t \text{ for all } t\in[0,1], i \in \I.
\end{equation}
or assume that (b) if the global null is true, all $p$-values are \emph{mirror-conservative}:
\begin{equation}
\label{cond:mirror_consv}
\text{If } \mathcal{H_G}_0 \text{ is true}, ~ f_i(a) \leq f_i(1 - a) \text{ for all } 0 \leq a \leq 0.5, i \in \I,
\end{equation}
where $f_i$ is the probability mass function of $p_i$ for discrete $p$-values or the density function otherwise.
Neither of the aforementioned conditions implies the other, though the former is more commonly made. Examples of mirror-conservative \pvalues include permutation $p$-values and one-sided tests of univariate parameters with monotone likelihood ratio~\citep{lei2018adapt}.
In the majority of the paper, it may be easier for the reader to pretend that the null $p$-values are exactly uniform for simplicity. Later in the paper, we explicitly demonstrate the distinct advantages of our tests for conservative $p$-values. We also assume that if the global null is true, the null \pvalues are independent of each other:
\[
\text{If } \mathcal{H_G}_0 \text{ is true}, ~ \{p_i\}_{i \in \I} \text{ are jointly independent.}
\]
This is also a common assumption; Fisher's test~\citep{fisher1992statistical} and Tukey's Higher Criticism~\citep{donoho2015special} are two other examples. Even though we are cognizant that independence is a strong assumption that only holds in some limited situations in practice (like meta-analysis), we wish to explore how much it can be exploited to design novel tests, for instance enabling the use of martingale techniques and ``masking'', as described soon. 

We remark that all aforementioned assumptions on the null $p$-values only need to hold under the global null. If the global null is not true, we do not require the null $p$-values (or the non-nulls) to have any particular marginal distribution or to satisfy any independence assumptions.

\subsection{Related work}

Our paper builds on and connects three distinct lines of work: classical work on global null testing, modern ideas on permitting interaction using $p$-value masking, and recent ideas on uniform martingale concentration inequalities. We discuss these separately below.

\paragraph{Global null testing.}
Most previous tests for the global null have been designed to work in the batch setting, and it continues to be an active area of research~\citep{ruger1978maximale, ruschendorf1982random, kost2002combining, owen2009karl, vovk2012combining}. Our work is most directly connected to tests which accumulate information as a sum, such as Fisher's and Stouffer's tests~\citep{stouffer1949american}.

There are many other global null tests like the Bonferroni method, Simes' test~\citep{simes1986improved}, and Higher Criticism, and our techniques do not apply to these. Importantly,
\emph{we do not claim that our interactive martingale tests are more powerful than prior work in any universal sense, but instead, our goal is to expand the creative design space of new procedures that can involve a human in the loop and explore their potential benefits.}

\paragraph{Permitting interaction by masking the $p$-values.}
The motivation behind masking $p$-values is to permit interaction with an analyst, who may freely employ models, prior knowledge and intuition, without any risk of violating type-I error control. The main idea is to decompose each individual \pvalue $p_i$ into two parts, 
\begin{equation}
\label{default_decompse}
h(p_i) = 2\cdot1\{p_i < 0.5\} - 1
  \quad\mathrm{and}\quad 
g(p_i) = \min \{p_i, 1 - p_i\}.
\end{equation}
Here, $g(p_i)$ is called the \emph{masked $p$-value}, while $h(p_i)$ is called the \emph{missing bit} since it is either plus or minus one. The critical observation is that $h(p_i)$ and $g(p_i)$ are independent if $H_i$ is null ($p_i$ is uniformly distributed). Masking was introduced recently by Lei and Fithian~\cite{lei2018adapt} in the context of false discovery rate (FDR) control, and further generalized and extended in Lei, Ramdas and Fithian~\cite{lei2017star} for FDR control under structural constraints, and then followed by work on FWER control~\citep{duan2020familywise}. The underlying property of masking can be traced to the ``knockoff'' method by Barber and Cand{\`e}s~\citep{barber2015controlling, arias2017distribution}. In this paper, we show that masking is also useful for global null testing in structured settings, and permitting interaction with an insightful analyst can improve power (but it is impossible for any analyst to violate type-I error control).

\paragraph{Uniform martingale concentration inequalities.} 
All new test statistics in this paper are designed to be martingales under the global null. The type-I error control guarantees for our tests thus stem from utilizing \emph{uniform} martingale concentration inequalities. These ``boundary crossing'' inequalities are high probability statements about the behavior of the entire trajectory of the martingale. In fact, several of our martingales have increments which are either fair coin flips ($\pm 1$) or standard Gaussians, which are some of the most well studied objects in sequential analysis, especially through their natural connections to Brownian motion~\citep{siegmund1986boundary}. In this paper, we care about nonasymptotic guarantees on the type-I error, and hence we use some recent line-crossing inequalities~\citep{howard2020time} and new curve-crossing inequalities~\citep{howard2020time1} that are nonasymptotic generalizations of the law of the iterated logarithm, which goes back to the work by Robbins~\citep{robbins1970statistical} (see Appendix~\ref{apd:stou_bounds} for a detailed comparison). For a martingale $M_k$, these boundaries are denoted $u_\alpha(k)$ and satisfy
\[
\mathbb{P}(\exists k \in \mathbb{N}: M_k > u_\alpha(k)) \leq \alpha.
\]
In the next section, we provide the exact expressions for the $u_\alpha(k)$ that we use, which are chosen because they have similar qualitative behavior but tighter constants than earlier work, references to which may be found within the aforementioned papers.

\subsection{Outline}

To progressively build intuition, the \premt is described in Section~\ref{sec:premt} followed by the \postmt in Section~\ref{sec:postmt}. In Section~\ref{sec:imt}, the general \imt is presented. For all these methods, the type-I error guarantees are presented immediately after the algorithms. 
However, power guarantees for all algorithms in the Gaussian sequence model are derived in Section~\ref{sec:power}. We then perform extensive simulations in Section~\ref{sec:ipm}. In Section~\ref{sec:robust}, we examine the robustness of our test to conservative nulls. Section~\ref{sec:anytime} explicitly describes how to interpret our tests as tracking an anytime-valid sequential $p$-value. Finally in Section~\ref{sec:other_decompose}, we discuss alternative ways of masking $p$-values. We end with a brief summary in Section~\ref{sec:dis}, and defer all proofs and additional experiments to the Appendix.

\section{The \premt} \label{sec:premt}

The \premt is not a single test, but instead, a general framework to extend the application of many classical methods that use the sum or product of transformed $p$-values, such as Stouffer's method~\citep{stouffer1949american} and Fisher's method~\citep{fisher1992statistical}, from the batch setting to the online setting. In this section, the ordering of hypotheses is fixed in advance by nature, or by the analyst using prior knowledge to place potential/suspected non-nulls early in the ordering.

\paragraph{The general framework.}
Our test takes the following general form:
\begin{equation}
\label{test:premt_batch}
   \text{Reject the null if } \sum_{i  = 1}^k f(p_i) \geq u_\alpha(k), \text{ for some } k \in \I,
\end{equation}
where $f(\cdot)$ is some transformation of the \pvalue, and $\{u_\alpha(k)\}_{k \in \mathbb{N}}$ is a boundary sequence depending on the choice of $f$. 
The boundary is determined by first establishing that the sequence $\{\sum_{i  = 1}^k f(p_i)\}_{k \in \mathbb{N}}$ is a martingale under the global null (after appropriate centering if needed). We then characterize the tail behavior of the martingale increments $f(p_i)$ for a uniform $p$-value. Finally, to control the type-I error, we employ recent results  ~\citep{howard2020time,howard2020time1} which provide boundaries under parametric and nonparametric conditions on the increments, such that with high probability the entire trajectory of the martingale is contained within the boundary. 

The preordered martingale test improves on its original batch version in two aspects. First, the applicability of the original test is extended from the batch setting to the online setting. Second, in the case of sparse non-nulls, the martingale version greatly improves the detection power if the non-nulls appear early on. As an example of converting a classic test to its martingale version, we develop the martingale Stouffer test below. Two more examples can be found in Appendix~\ref{apd:fisher} for a martingale Fisher test using $f(p_i)=-2\log p_i$,  and Appendix~\ref{apd:chi-square} for a martingale chi-square test using $f(p_i)=\left[\Phi^{-1}(1 - p_i)\right]^2$.

\paragraph{An example: martingale Stouffer test (MST).} 
The batch test by Stouffer~\cite{stouffer1949american} calculates $S_n = \sum_{i=1}^n \Phi^{-1}(1 - p_i)$, where $\Phi(\cdot)$ denotes the standard Gaussian CDF. Since the distribution of $S_n$ under the global null is $\mathcal{N}(0,n)$, the batch test rejects when $S_n > \sqrt{n} \Phi^{-1}(1-\alpha)$. To design the martingale test, simply observe that $\{S_k\}_{k \in \I}$ is a martingale whose increments $f(p_i)  = \Phi^{-1}(1 - p_i)$ are standard Gaussians under the global null. There are several types of uniform boundaries $u_\alpha(k)$ for a Gaussian increment martingale, and here we give two examples: linear and curved. The first boundary (transformed from equation~(2.29) in Howard et al. \cite{howard2020time}), which can be derived from the Gaussian sequential probability ratio test~\citep{wald1945sequential}, grows linearly with time. Specifically, the test rejects the global null if
\begin{equation}
\label{test:mst_lin}
    \exists k \in \mathbb{N}: \sum_{i=1}^k \Phi^{-1}(1-p_i) \geq \sqrt{\frac{-\log\alpha}{2m}}k + \sqrt{\frac{-m\log\alpha}{2} },
\end{equation}
where $m \in \mathbb{R}_+$ is a tuning parameter that determines the time at which the bound is tightest: a larger $m$ results in a lower slope but a larger offset, making the bound loose early on. We suggest a default value of $m=n/4$ if the number of hypotheses $n$ is finite, but it should be chosen based on the time by which we expect to have encountered most non-nulls (if any). In contrast, the \mst with a curved boundary (equation~(2) in Howard et al. \cite{howard2020time1}) rejects the global null if
\begin{equation}
\label{test:mst_curve}
    \exists k \in \mathbb{N}: \sum_{i=1}^k \Phi^{-1}(1-p_i) \geq 1.7\sqrt{k\left(\log\log(2 k) + 0.72 \log \frac{5.2}{\alpha}\right)}.
\end{equation}
These bounds differ in the quota of error budget distributed to every step $k = 1,2,\ldots$, which can influence the detection power of the martingale test as it is more likely to exceed a tighter bound. Curved bounds have a slower growth rate $O(\sqrt{k\log\log k})$ than the linear bounds, indicating a tighter bound for large enough $k$, but they are usually looser for small $k$. Comparisons of the test with several linear and curved boundaries are given in Appendix~\ref{apd:stou_bounds}. Generally, the linear bound is recommended for the batch setting, and the curved bound for the online setting. 

The \mst with either boundary controls the type-I error, if under the global null the sum $\{\sum_{i=1}^k \Phi^{-1}(1-p_i)\}_{k \in \mathbb{N}}$ is stochastically upper bounded by a martingale with standard Gaussian increments, which holds under our assumption that the null \pvalues are stochastically larger than uniform, as stated below. 
\begin{theorem}
\label{thm:mst_error}
If the $p$-values are independent and stochastically larger than uniform under the global null, then
the \mst with linear boundary~\eqref{test:mst_lin} or curved boundary~\eqref{test:mst_curve} controls the type-I error at level $\alpha$.
\end{theorem}

The next natural question is what we can prove about the detection power of the aforementioned tests. While this is treated more formally later in the paper, for now it suffices to say that the power of the \mst relies on a good preordering that places non-nulls up front. If such prior knowledge is not available (and say the preordering is completely random, or even adversarial), then the preordered martingale tests can have poor power. This motivates the development of methods based on data-adaptive orderings, as treated next.

\section{Adaptive and interactive methods} 
\label{sec:adp_interactive}
To develop intuition progressively, we first introduce a martingale test whose ordering depends on the $p$-values in Section~\ref{sec:postmt}, and extend it in Section~\ref{sec:imt} to an interactive test, whose ordering can additionally depend on side information (covariates) and human interaction.

\subsection{The \postmt (AMT)} \label{sec:postmt}
If we naively use the $p$-values to both determine the ordering as well as form the test statistic, the resulting ``double-dipped'' sequence of test statistics does not form a martingale under the global null. In order to allow using the $p$-value for determining the ordering, we use a recent idea called masking, as briefly mentioned in the introduction.
Each \pvalue $p_i$ is decomposed as
\begin{equation*}
h(p_i) = 2\cdot1\{p_i < 0.5\} - 1,
   \quad\quad 
g(p_i) = \min \{p_i, 1 - p_i\},
\end{equation*}
where $h(p_i)$ is called the missing bit, and $g(p_i)$ is called the masked $p$-value. The masked $p$-values are used to create the ordering (by placing smaller ones up front) while the test statistic just sums the missing bits $h(p_i)$ in that order. Since $h(p_i)$ and $g(p_i)$ are independent under the global null, sorting by the $g(p_i)$ values results in a uniformly random ordering, and the sum of $h(p_i)$ is just a random walk of independent coin flips. Formally, define the set $M_k$ as the first $k$ hypotheses ascendingly ordered by $g(p_i)$. Our test rejects $\mathcal{H_G}_0$ if
\begin{equation*}
\label{test:postmt_batch}
    \exists k \in \{1,\ldots,n\}: \sum_{i \in M_k} h(p_i) \geq u_\alpha(k),
\end{equation*}
where the upper bound $u_\alpha(k)$ is the same as for the \mst in equations~\eqref{test:mst_lin} and~\eqref{test:mst_curve}, since the sequence of sums $\sum_{i \in M_k} h(p_i)$ is also a martingale with 1-subGaussian increments under the global null. The adaptively ordered martingale test in the batch setting is summarized below.

\begin{algorithm}[H]
\label{alg:post_mt_batch}
\SetAlgoLined
\textbf{Input:}
\pvalues $(p_i)_{i = 1}^n$,
target type-I error rate $\alpha$\;
\textbf{Procedure:} Initialize $M_0 = \emptyset$\;
\For{$k=1,\ldots,n$ } {
$M_k = M_{k-1} \cup \argmin_{i \notin M_{k-1}} g(p_i) $\; 

\uIf{$\sum_{i \in M_k} h(p_i) > u_\alpha(k)$}{
   reject the global null and stop\;
   }
}
\caption{The \postmt (batch setting)}
\end{algorithm}

The \postmt in the online setting proceeds slightly differently: it screens the hypotheses by $g(p)$ so that only promising non-nulls enter the set $M_k$.  Specifically, given a threshold parameter $c$ (such as $0.05$), the set $M_k$ expands at time $t$ only if $g(p_t) < c$, as summarized below.

\begin{algorithm}[H]
\label{alg:post_mt_online}
\SetAlgoLined
\textbf{Input:}
target type-I error rate $\alpha$, threshold parameter $c$\;
\textbf{Procedure:} Initialize $M_0 = \emptyset$, size $k=0$\;
\For{$t=1,\ldots,$ } {
$p_t$ is revealed by nature\;
\uIf{$g(p_t) < c$}{$k \leftarrow k+1$, $M_k = M_{k-1} \cup \{t\}$\;} 
\uIf{$\sum_{i \in M_k} h(p_i) > u_\alpha(k)$}{
   reject the global null and stop\;
   }
}
\caption{The \postmt (online setting)}
\end{algorithm}

The \postmt controls type-I error if under the global null, all \pvalues are \textit{mirror-conservative} \eqref{cond:mirror_consv}, as formally stated below.

\begin{theorem}
\label{thm:adp_error}
If the $p$-values are independent and mirror-conservative under the global null, then the \postmt controls the type-I error at level $\alpha$.
\end{theorem}

In the batch setting, the adaptive ordering (as realized by the nested sequence $\{M_k\}$) is fully determined at the start of the procedure by sorting the masked $p$-values. In the next section, we demonstrate that in the presence of independent covariates $x_i$ for each hypothesis and side information such as structural constraints on potential rejected sets, it is actually beneficial to \emph{interactively} determine the ordering one step at a time with a human-in-the-loop, who may be guided by the masked $p$-values as well as intuition and working models.

\subsection{The \imt (IMT)} \label{sec:imt}

The \imt also applies to both batch and online settings. We first describe the method in the batch setting with side information and structural constraints, where the power of interactivity is more compelling.

To begin, first suppose that in addition to the $p$-values, the scientist also has some side information about each hypothesis available to them in the form of covariates $x_i$. For example, if the hypotheses are arranged in a rectangular grid, then $x_i$ could be the coordinates on the grid for hypothesis $i$ (examples in Section~\ref{sec:sim_cluster}). We then suppose that the scientist also has some prior knowledge or intuition about what structural constraints the non-nulls would have, if the global null is false. For example, perhaps the scientist thinks that the non-nulls (if any) would be clustered on the grid, themselves forming a rectangular shape (of some size, at some location). Our main assumption about the covariates is:
\[
    \text{Under the global null, $x_i \perp p_j$ for all $i, j \in \I$. }
\]
This is a common assumption for tests that incorporate covariate information, such as Independent Hypothesis Weighting \citep{ignatiadis2016data}, AdaPT \citep{lei2018adapt}, and STAR \citep{lei2017star}. In fact, because the aforementioned methods aim at error control of more stringent metrics such as FDR and FWER, their assumptions are stronger in the sense that the independence between $x_i$ and $p_i$ is required for the hypotheses that are truly null even when the global null is not true (i.e., there exist non-nulls). Our interactively ordered martingale test satisfies the following two properties: (a)~if the global null is true, the type-I error is controlled, regardless of what the scientist thinks or acts, (b)~if the global null is false, and the prior knowledge and/or structural constraints are accurate (or somewhat so), then the power of the test is high.
The interactive test proceeds as follows:
\begin{itemize}
    \item At the beginning, all covariates and masked $p$-values $(x_i, g(p_i))_{i \in \I}$ are revealed to the scientist, while only the missing bits $(h(p_i))_{i \in \I}$ remain hidden. We initialize $M_0=\emptyset$.
    \item The scientist repeats the following at each time step $k \geq 1$: they choose a promising hypothesis $i^\star_k$ from $[n] \backslash M_{k-1}$, and update ${M_k = M_{k-1} \cup \{i^\star_k\}}$.
    \item On doing so, they learn $h(p_{i^\star_k})$, and thus keep track of $S_k := \sum_{i \in M_{k}} h(p_i)$. If $S_k > u_\alpha(k)$ for any $k$, they stop and reject the global null. 
\end{itemize}
Type-I error control is essentially guaranteed because regardless of how the scientist acts at each step, if the global null is true, all the $g(p_i)$ values and the revealed $h(p_i)$ values do not provide any information about the still hidden missing bits, and thus $S_k$ is a martingale.

When the global null is false, we expect the power to be high because of the following reasons. First, the scientist may use any working model of their choice (or none at all) to guide their choice at each step. For example, they can attempt to estimate the likelihood of being non-null for each hypothesis $i$ at each step $k$, denoted as $\pi_i^{(k)}$ (posterior probability of being non-null). In fact, as they learn the missing bits at each step, they can change their model or update their prior knowledge based on the observed $p$-values thus far. The information available to the scientist at the end of step $k$ is denoted by the filtration
\[
\mathcal{F}_k := \sigma\left((x_i, g(p_i))_{i = 1}^n, (p_i)_{i \in M_k} \right),
\]
and thus the choice $i^\star_k$ is predictable, meaning it is measurable with respect to~$\mathcal{F}_{k-1}$. The general interactive framework is summarized below as Algorithm~\ref{alg:imt}. 

\begin{algorithm}[H]
\label{alg:imt}
\SetAlgoLined
\textbf{Information available to the scientist:} side covariate information and/or structural constraints, and masked \pvalues $\mathcal{F}_0 := \sigma((x_i, g(p_i))_{i = 1}^n)$,
target error $\alpha$\;
\textbf{Procedure:} Initialize $M_0 = \emptyset$\;
\For{$k=1,\ldots,n$ } {
Using $\mathcal{F}_{k-1}$, pick any $i^\star_k \in [n] \backslash M_{k-1}$. Update
$M_k = M_{k-1} \cup \{i^\star_k\}$\;
Reveal $h(p_{i^\star_k})$ and update $\mathcal{F}_k := \sigma\left((x_i, g(p_i))_{i = 1}^n, (p_i)_{i \in M_k}\right)$\;
\uIf{$\sum_{i \in M_{k}} h(p_i) > u_\alpha(k)$}{
   reject the global null and exit\;
   }
}
\caption{The \imt (batch setting)}
\end{algorithm}

The \imt in the online setting screens the hypotheses based on information in $\mathcal{F}_{t-1}$ such that $p_t$ enters the set $M_k$ only when it is a promising non-null, as described in Algorithm~\ref{alg:imt_online}.

\begin{algorithm}[H]
\label{alg:imt_online}
\SetAlgoLined
\textbf{Procedure:} Input target error $\alpha$. Initialize $M_0 = \emptyset$, size $k = 0$\;
\For{$t=1,\ldots,$ } {
\textbf{Information available to the scientist:} side covariate information and/or structural constraints, and (masked) \pvalues $\mathcal{F}_{t-1} := \sigma((x_i, g(p_i))_{i = 1}^{t}, (p_i)_{i = 1}^{t - 1})$\;
Using $\mathcal{F}_{t-1}$, decide whether hypothesis $t$ should be included in $M_{k-1}$\;
\uIf{include hypothesis $t$}{$k \leftarrow k+1$, $M_k = M_{k-1} \cup \{t\}$\;
} 
\uIf{$\sum_{i \in M_k} h(p_i) > u_\alpha(k)$}{
   reject the global null and stop\;
   }
}
\caption{The \imt (online setting)}
\end{algorithm}

The aforementioned algorithms (or frameworks) comes with the following error guarantee, regardless of the choices made by the scientist.

\begin{theorem}
\label{thm:imt_error}
If under $\mathcal{H_G}_0$, the \pvalues are \textit{mirror-conservative} and are independent of each other and of the covariates $x_i$, then the \imt controls the type-I error at level~$\alpha$.
\end{theorem}

Note that there is no requirement whatsoever on the null or non-null $p$-values (i.e., $p$-values from the hypotheses that are truly non-null) when the global null is false. As before, note that under the global null, the missing bits are random fair coin flips, and the masked $p$-values are uniform on $[0,0.5]$ and completely uninformative about the missing bit. However, under the alternative, the true signals have very small masked $p$-values (say 0.01, 0.003, etc.) and along with covariate information, one may be able to infer that the missing bit is more likely to be +1 and thus include it in the ordering. Continuing the grid example from the start of this section, by revealing all but one bit per $p$-value at the start of the procedure, the scientist can possibly notice if \emph{small} masked $p$-values are randomly scattered or clustered on the grid. 

\begin{remark}
For any particular setup, like our example of a grid with a cluster of signals, it may be possible to design a better global null test that is perfectly suited for that setting. Hence, we do not claim that our interactive method is the right test to use in all problem setups. Its main advantage is its generality: instead of having to design a new test for each situation (trying to figure out how to optimally combine prior knowledge, structural constraints and covariates from scratch), our general framework provides a simple and flexible alternative.
\end{remark}
 
The correctness of the test (proof in Appendix~\ref{apd:imt_proof}) hinges on one bit from each $p$-value being hidden from the scientist. Once this protocol has been run once, and all $p$-values have been unmasked, the procedure obviously cannot be run a second time from scratch. In other words, our interactive setup does not prevent these and related forms of $p$-hacking. This is similar to the traditional offline setup, where it is not allowed to pick the global null test after observing the $p$-values and guessing which test will have the highest power to reject, and if scientists do this anyway and report only the final finding, we would have no way to know whether such inappropriate double-dipping has occurred.

It is worth remarking on the main disadvantage of such a test, relative to (say) the martingale Stouffer test introduced earlier. The interactive test statistic is a sum of coin flips (missing bits) -- no matter how strong the signal might be, the interactive test statistic can only increase by one at most. On the other hand, the martingale Stouffer test adds up Gaussians, and if there is a strong signal (very small $p$-value), it can stop very early. If a relatively good prior ordering is known to the scientist, the martingale Stouffer test should be preferred. However, if the prior knowledge is not in the form of an ordering, but some intuition about how the covariates and $p$-values may be related or what type of structure the non-nulls may have (if any), then the interactive test can be much more powerful.

The above framework leaves the specific strategy of expanding $M_k$ unspecified, allowing much flexibility. Now, we give one example of how $i^\star_k$ can be chosen based on the available information $\mathcal{F}_k$. One straightforward choice for $i^\star_k$ is the hypothesis not in $M_k$ with the highest posterior probability of being non-null, computed with the aid of a working model, like the Bayesian two groups model, where each \pvalue $p_i$ is drawn from a mixture of a null distribution $F_0$ with probability $1 - \pi_i$ and an alternative distribution $F_1$ with probability $\pi_i$:
\begin{equation}
\label{eq:p_model}
    p_i \sim (1-\pi_i) F_0  +  \pi_i F_1.
\end{equation}
For example, we can choose $F_0$ as a uniform and $F_1$ as a beta distribution. We may further posit a working model that treats $\pi_i$ as a smooth function of $x_i$. The masked \pvalues $g(p_i)$ and the revealed missing bits in $\mathcal{F}_{k-1}$ can be used to infer the other missing bits using the EM algorithm (see Appendix~\ref{apd:em}). The missing bits that are inferred to be more likely +1 should be chosen, potentially in accordance with other structural constraints. Importantly, the type-I error is controlled regardless of the correctness of the working model or any heuristics to expand $M_k$.

\section{Power guarantees of non-interactive procedures} \label{sec:power}

This section is devoted to an analysis of the power of the \mst and the \postmt. It's hard to analyze the power for the \imt due to its flexible framework offered to the user: it can have high power if the user specifies a good interactive algorithm, and vice versa. Nevertheless, to demonstrate the advantages of the \imt, we present numerical results under structured non-nulls in the next section.

Our analysis includes power guarantees in the batch and online settings in a simple Gaussian setup. Specifically, we consider a simple multiple testing problem where each hypothesis is a one sided hypothesis on the mean value of a Gaussian. In this setting, the $i$-th null hypothesis is that a Gaussian has zero mean, and the alternative is that the Gaussian has a positive mean $\mu_i > 0$. 

\begin{setting}
\label{set:simple}
We observe $Z_1,\ldots, Z_n$ where $Z_i \sim N(\mu_i, 1)$ and wish to distinguish the following hypotheses: 
\begin{align*}
\mathcal{H_G}_0&: \mu_i = 0 \text{ for all } i\in\I, \quad \text{versus } \\
\mathcal{H_G}_1&: \mu_i > 0 \text{ for some } i \in \I.
\end{align*}
\end{setting}

\noindent In the remainder of this section, we let $r_i := I(\mu_i > 0)$ indicate the non-null hypotheses. Although we compare the power of various tests in this relatively simple setting, we emphasize that our tests are more broadly applicable to general settings where the $p$-values are mirror-conservative under the null.

\begin{figure}[!ht]
    \centering
    \begin{subfigure}[t]{0.45\textwidth}
        \centering
        \includegraphics[width=1\linewidth]{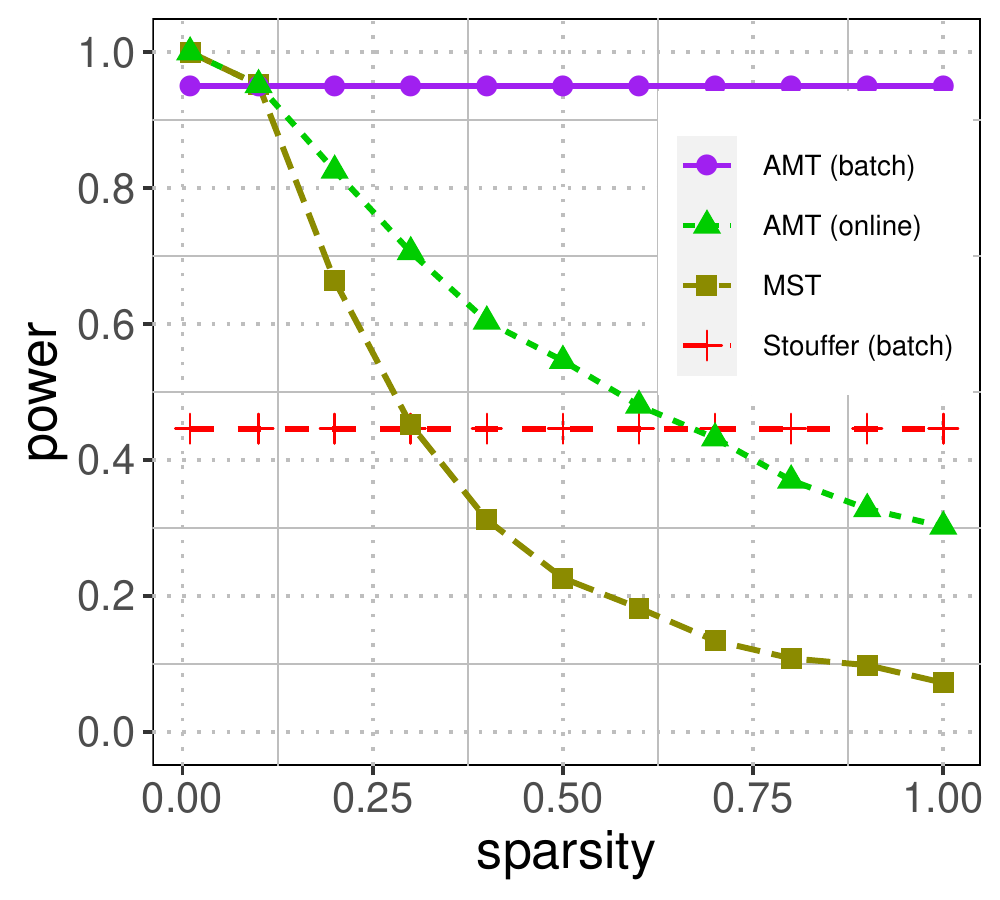}
        \caption{Power comparison in the batch setting when varying the prior ordering.
        Larger sparsity indicates a worse prior ordering. The AMT procedures (batch and online) adaptively alter the ordering and are more robust to bad quality of orderings than MST. Still, when the prior ordering is great, AMT has lower power than MST because the increments of AMT's test statistic are bounded by~$+1$. This phenomenon is mathematically predicted by Theorem~\ref{thm:power_batch} and Theorem~\ref{thm:power_batch_imt}.}
        \label{fig:reorder_batch}
    \end{subfigure}
    \hfill
    \begin{subfigure}[t]{0.45\textwidth}
        \centering
        \includegraphics[width=1\linewidth]{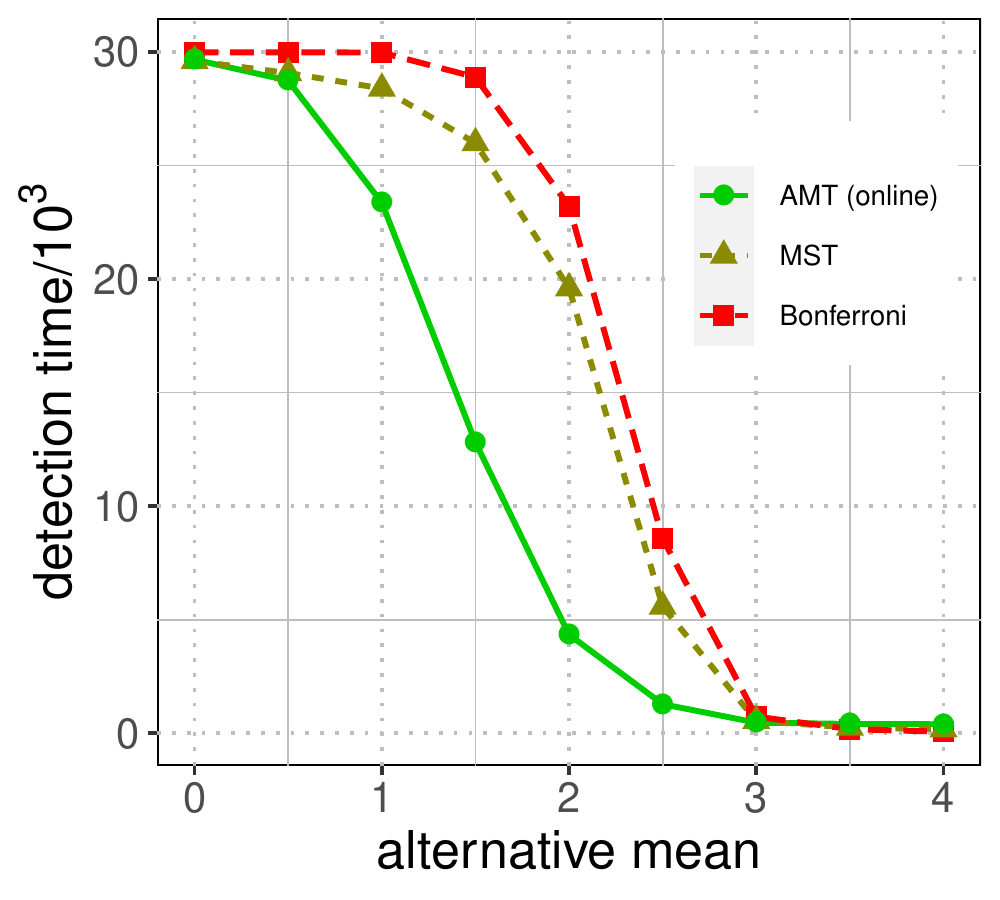}
        \caption{Number of hypotheses needed to reject the global null (detection time) in the online setting when varying the alternative mean $\mu$. Each hypothesis has the same probability of being non-null as $5\%$. The Bonferroni method cannot reject the global null unless $\mu$ is greater or equal than 3. AMT is the first to reject the global null when $\mu$ is small because it filters the hypotheses by masked \pvalues. This phenomenon is mathematically predicted by Theorem~\ref{thm:power_online} and
        Theorem~\ref{thm:power_online_imt}.}
        \label{fig:reorder_online}
    \end{subfigure}
    
    \caption{Illustrative simulations that compare the batch and online \mst (MST) and the \postmt (AMT) under Setting~\ref{set:simple}. All plots in this paper present the averaged power (in the batch setting) and averaged rejection time (in the online setting) over 500 repetitions, and the type-I error is $\alpha = 0.05$.}
    \label{fig:illustrative_power}
\end{figure}

With this setup in place, we now summarize the main results of this section.
\begin{itemize} 
    \item In Section~\ref{sec:power_batch}, we focus on the batch setting. In Theorem~\ref{thm:power_batch}, we compare the power of the \mst with its batch counterpart, showing that when a good a-priori ordering is used the \mst can have much higher power. Our next result, Theorem~\ref{thm:power_batch_imt}, studies the \postmt in the batch setting. The \postmt expands the \tset $M_k$ based on masked \pvalues, and tests the global null using the missing bits $h(p_i)$. We show that in cases when the signal strength is high, re-ordering by the masked \pvalues can significantly improve power of the resulting test by ensuring that promising hypotheses are considered early on with high-probability.
    
    \item In Section~\ref{sec:power_online} we turn our attention to the online setting. In Theorem~\ref{thm:power_online}, we study the power of a simple online Bonferroni test, and compare this in Theorem~\ref{thm:power_online_imt} with the power of the adaptively ordered martingale test. For the adaptively ordered martingale test, we study the role of the threshold parameter $c$ in the power of the test, characterizing some of the tradeoffs involved in the choice of this parameter.
\end{itemize}
Figure~\ref{fig:illustrative_power} visualizes the above power comparisons by two simple simulations in batch and online settings\footnote{\href{https://github.com/duanby/interactive-martingale}{https://github.com/duanby/interactive-martingale} has R code to reproduce all plots.}. Details of the batch experiment appear next. 

We simulate $10^4$ hypotheses with 50 non-nulls ($\mu_i = 3$). The position of the non-nulls is encoded by a \textit{sparsity} parameter: the non-nulls are uniformly distributed in the first $\textit{sparsity}\cdot n$ hypotheses. Thus, larger sparsity indicates a poorer prior ordering (the non-nulls are more scattered), and it is expected to result in lower power for order-dependent methods. Indeed, we observe that: (1) two batch procedures (the \postmt (AMT) in the batch version and Stouffer's test) get the $p$-values as a set, ignoring the prior ordering, and hence their power is a flat line; (2) the online AMT and the MST procedure uses $p$-values in the ordering provided to it, and their power degrades as the quality of the ordering degrades; (3) the online AMT is less sensitive to bad prior ordering than the MST because it discards possible nulls based on the masked $p$-values; but it could still let in many nulls if the discarding threshold is not tight and most nulls are in front, leading to lower power under a worse prior ordering; (4)~overall, the AMT procedures (batch and online) are more robust to bad prior ordering than the MST because they adaptively alter the ordering.

Keep in mind that the simulations above and the power analysis below assume no prior knowledge, but the \imt has higher power when taking advantage of the non-null structure, as shown in Section~\ref{sec:ipm}.

\subsection{Power guarantees in the batch setting} \label{sec:power_batch}

We begin by studying the power of the batch, martingale and interactive martingale tests in the batch setting.

\paragraph{The batch Stouffer test and the \mst:}
The batch Stouffer test simply aggregates the observed $Z_1,\ldots,Z_n$ and compares this with an appropriate threshold. In contrast, the martingale Stouffer test \emph{sequentially} compares partial aggregations with an appropriate threshold.

To state our result compactly, for a specified value $\gamma$, we define:
\begin{align}
\label{eqn:ckgamma}
C_{k}^\gamma &:= 1.7\sqrt{\log\log(2k) + 0.72\log \frac{5.2}{\gamma}},
\end{align}
which corresponds to the curved boundary in~\eqref{test:mst_curve} divided by $\sqrt{k}$. This quantity grows very slowly with $k$ (at the rate of $\sqrt{\log \log (k)}$) and for all practical purposes can be treated as a ``constant''. We have the following result:

\begin{theorem}
\label{thm:power_batch}
\begin{enumerate}
    \item[(a)] {\bf Batch Stouffer Test (necessary+sufficient):} A necessary and sufficient condition for the batch Stouffer test with type-I error~$\alpha$ to have at least~$1 - \beta$ power is that
    \begin{align} \label{eq:batch_stouffer_power}
    \sum_{i=1}^n r_i \mu_i \geq (Z_\alpha + Z_\beta) n^{1/2},
    \end{align}
    where $Z_\alpha = \Phi^{-1}(1 - \alpha)$ is the $(1-\alpha)$-quantile of a standard Gaussian.
    
    \item[(b)] {\bf Martingale Stouffer Test (sufficient):} A sufficient condition for MST to have power at least $1-\beta$ is
    \begin{align} \label{eq:batch_mst_power}
    \exists~k \in \{1,\ldots,n\}, \quad \sum_{i=1}^k r_i \mu_i  \geq \left(C_{k}^\alpha + C_{k}^\beta\right)k^{1/2}.
    \end{align}

    \item[(c)] {\bf Martingale Stouffer Test (necessary):} If ${\alpha<1-\beta}$, the power of MST is less than $1-\beta$ whenever
    \begin{align*}
    \forall~k \in \{1,\ldots,n\}, \quad \sum_{i=1}^k r_i \mu_i \leq (C_{k}^\alpha - C_{k}^{1-\beta}) k^{1/2}.
    \end{align*}
\end{enumerate}
\end{theorem}

\noindent We defer the proof of this result to Appendix~\ref{apd:power_batch_nonadaptive}. Several remarks are in order.
\vspace{0.05cm}

\begin{itemize}
    \item It is also possible to study the power of the Bonferroni test in the batch setting. A necessary condition for the power of the Bonferroni method to be at least $1-\beta$ is:
    \begin{align*}
    \exists~k \in \{1,\ldots,n\}, \quad r_k \mu_k \geq Z_{\alpha/n} + Z_\beta.
    \end{align*} 
    Comparing with the batch Stouffer test, we see that the Bonferroni method has high power when there is at least one large effect, but can have lower power in settings where there are many small non-null effects.
 
    \item Comparing condition~\eqref{eq:batch_stouffer_power} for the batch Stouffer test with its martingale counterpart (condition~\eqref{eq:batch_mst_power}), we observe that the batch test rejects when the average of \emph{all} the effects is sufficiently large, while the martingale test rejects as long as \emph{any} cumulative sum is sufficiently large. In cases where a good a-priori ordering is available, the martingale test can have much higher power.
\end{itemize}

\paragraph{The \postmt: } 
To ease our calculations, we assume that all the non-nulls have the same mean value, i.e. $\mu_i = \mu$ if $r_i = 1$. 
We denote the number of non-nulls by $N_1$ and the nulls by $N_0$. Let $Z(\nu)$ be a Gaussian random variable with unit variance and mean $\nu$, then the non-nulls are $\{Z_j(\mu)\}$ for $j = 1, \ldots, N_1$ and we let $Z_{(j)}(\mu)$ be the $j$-th non-null after ordering by its absolute value so that 
\begin{align}
    \label{def:permutation}
    |Z_{(1)}(\mu)| \geq |Z_{(2)}(\mu)| \geq \ldots \geq |Z_{(N_1)}(\mu)|.
\end{align}
Suppose that $X \sim \text{Bin}(n,p)$. We let $t_\alpha(n,p)$ denote the $\alpha$-upper quantile of the Binomial distribution $\text{Bin}(n,p)$, i.e. $\mathbb{P}(X \geq t_\alpha(n,p)) = \alpha.$
Recall the definition of $C_k^\gamma$ in equation~\eqref{eqn:ckgamma}.
We define, for $j \in \{1,\ldots,N_1\}$,
\begin{align*}
q_j := \mathbb{P}(|Z(0)| > |Z_{(j)}(\mu)|),
\end{align*}
to be a measure of signal strength. Roughly, the values $q_j$ will be close to 0, if the signal strength $\mu$ is large.

\begin{theorem}
\label{thm:power_batch_imt}
 The \postmt with level $\alpha$ has at least $1-\beta$ power if
\begin{align}
\label{eqn:uninterpretable}
    \exists~j \in \{1,\ldots,N_1\}:~&\sum_{s=1}^j \left(2\mathbb{P}(Z_{(s)}(\mu) > 0) - 1\right) \nonumber{}\\ & \geq\left(C_{n}^{\alpha} + C_{n}^{\beta/2}\right)(j + t_{\beta/(2N_1)}(N_0, q_j))^{1/2}.
\end{align}
\end{theorem}

\noindent We prove this result in Appendix~\ref{apd:power_imt}. 
Condition~\eqref{eqn:uninterpretable} gives a reasonably tight sufficient condition for the re-ordering based test to have high power (Figure~\ref{fig:heat_mu}). As expected, when the number of nulls increases (right columns) or the number of non-nulls decreases (bottom rows), the sufficient condition for the signal strength~$\mu$ to guarantee high power grows.

\begin{figure}[t]
    \centering
        \includegraphics[width=0.5\linewidth]{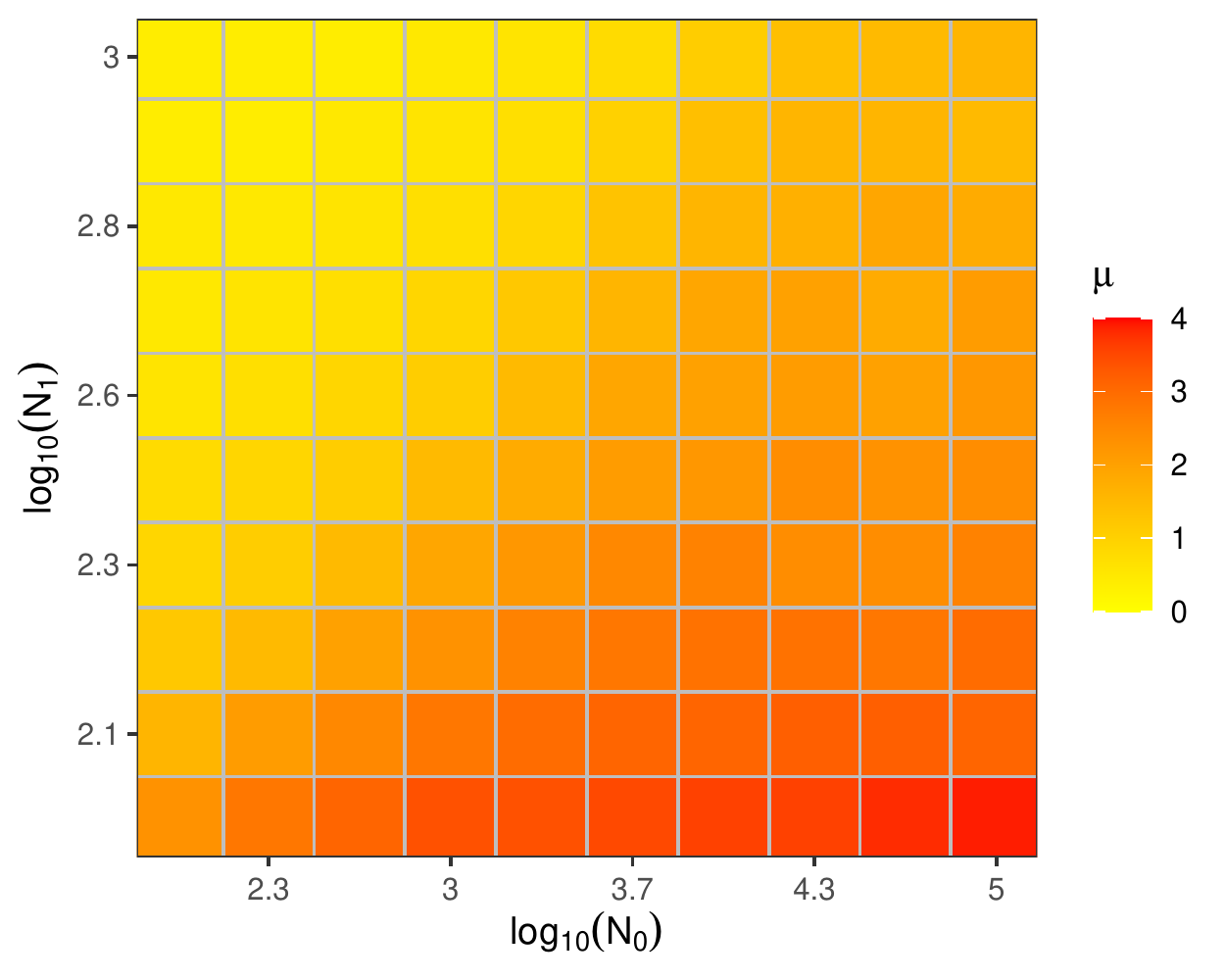}
    \caption{Sufficient signal strength $\mu$ for AMT to guarantee both type-I and type-II error control at 0.05 (derived from~\eqref{eqn:uninterpretable}), when varying the numbers of nulls~${N_0 \in [10^2, 10^{5}]}$ and non-nulls~${N_1\in [10^2, 10^3]}$. The required signal strength grows when the number of nulls increases or the number of non-nulls decreases.}
    \label{fig:heat_mu}
\end{figure}

The condition itself can be difficult to interpret as it depends on the distribution of Gaussian order statistics, as well as on the quantiles of a Binomial distribution. To build some intuition, we consider some simple cases.
\begin{itemize}
    \item In the extreme case, when the signal strength $\mu$ is quite large, the re-ordering will ensure that the non-nulls are placed early on with high-probability. In this case, the left-hand side in condition~\eqref{eqn:uninterpretable} grows linearly with $j$. On the other hand, if the signal strength is large then the probabilities $q_j$ will be small and we can ignore the term $t_{\beta/(2N_1)}(N_0, q_j)$, so that the right-hand side grows at the rate of roughly $\sqrt{j}$ (ignoring $\log \log$ factors), ensuring that the condition will be satisfied even for a moderate number of non-nulls.

    \item We provide other conditions that suffice to ensure high power in Appendix~\ref{apd:cond_imt_mu} by lower and upper bounding the left and right hand sides (respectively). We present one sufficient condition here. Suppose there are sufficient number of non-nulls such that $N_1 \geq 6 \left(C_n^{\alpha} + C_n^{\beta/2}\right)^2$, and that the number of nulls is sufficiently large, i.e. that $N_0 > 0.1 N_1^2$. A sufficient condition for the \postmt to have $1 - \beta$ power is
    \begin{align} \label{eq:cond_imt_mu}
    \mu  \geq \sqrt{2\log\left( \frac{N_0}{N_1^2}\right) + 4\log\left(C_{n}^{\alpha} + C_{n}^{\beta/2}\right) + 3.45}.
    \end{align}
    For comparison, the batch Stouffer test requires
    \begin{align} \label{eq:batch_mu_cond}
    \mu \geq (Z_\alpha + Z_{\beta}) \sqrt{\frac{N_0}{N_1^2} + \frac{1}{N_1}}.
    \end{align}
    Both conditions are stricter if the ratio $\frac{N_0}{N_1^2}$ is large, i.e. in the setting where there are many nulls and few non-nulls. However, the \postmt requires a signal strength that only grows logarithmically with this ratio.
\end{itemize}
\noindent In Appendix~\ref{apd:cond_imt_mu}, we relate condition~\eqref{eq:cond_imt_mu}
to the detection threshold derived in the work of Donoho and Jin~\cite{donoho2015special} for the same setting of detecting sparse Gaussian mixtures.

To summarize our findings in the batch setting: the \mst and the \postmt each require weaker conditions for the same power than the batch Stouffer test. The \mst relies on a good pre-defined ordering, whereas the \postmt relies on sufficiently large signal strength to ensure that re-ordering is helpful. We now turn our attention to the online setting.

\subsection{Power guarantees in the online setting} \label{sec:power_online}
When testing the global null, the natural test to compare to is the online Bonferroni method, which chooses a sequence of significance levels $\{\alpha_k\}_{k=1}^\infty$ such that $\sum_{k=1}^\infty \alpha_k = \alpha$, and rejects the global null if
\begin{align*}
    \exists~k \in \mathbb{N}: p_k \leq \alpha_k.
\end{align*}
The following sections compare the power guarantee of the online Bonferroni method with the \mst and \postmt. Specifically, we derive necessary conditions for the power of the online Bonferroni test, and compare it with sufficient conditions for the power of our proposed methods -- revealing situation where the online Bonferroni has lower power than our proposed methods.

\paragraph{The online Bonferroni method versus the \mst:} 

To better characterize the power of  online Bonferroni, we consider two cases: 
\begin{itemize}
    \item Dense non-nulls: the number of non-nulls is infinite,
    \begin{align} \label{eq:dense_nonnulls}
        \sum_{k  = 1}^\infty r_k = \infty.
    \end{align}
    \item Sparse non-nulls: the number of non-nulls is finite,
    \begin{align} \label{eq:sparse_nonnulls}
        \sum_{k  = 1}^\infty r_k \leq M < \infty \text{ for some large constant } M.
    \end{align}
\end{itemize}
The sparse case yields a stronger necessary condition when the sequence of significance levels satisfies a mild condition that $\{\alpha_k\}_{k=1}^\infty$ is nonincreasing.

Unlike previous methods, the online Bonferroni method does not aggregate \pvalues, so its power guarantee requires conditions on the individual means.

\begin{theorem}
\label{thm:power_online}

Suppose $\alpha \leq (1-\beta)/4$. In the case of dense non-nulls~\eqref{eq:dense_nonnulls}, a necessary condition for online Bonferroni to have at least $1-\beta$ power is
\begin{align} \label{eq:online_Bonferroni_power}
    \exists k \in \mathbb{N}: r_k\mu_k \geq 0.25 \left(\sqrt{2\log \left(\frac{k^2}{\alpha}\right)}\right)^{-1}.
\end{align}
A stronger necessary condition can be derived for sparse non-nulls~\eqref{eq:sparse_nonnulls}. If $\{\alpha_k\}_{k=1}^\infty$ is nonincreasing, then online Bonferroni can have at least $1 - \beta$ power only if
\begin{align}
    \exists k \in \mathbb{N}: 
    \begin{cases}
    r_k\mu_k \geq 0.4\sqrt{\alpha_{k^*}}, & \text{ if } k \leq k^*,\\
    r_k\mu_k \geq \sqrt{\log \left(\frac{k}{4\alpha}\right)} - \sqrt{2\log\left(\frac{M}{2(1 - \beta - 3\alpha)}\right)}, & \text{ if } k > k^*,
    \end{cases}
\end{align}
where $k^* = M^2/\alpha$, and $\alpha_{k^*}$ is the $k^*$-th significance level.

In contrast, a sufficient condition for the \mst to have at least $1-\beta$ power is
\begin{align}
\label{eqn:suff}
\exists~k \in \mathbb{N}:
\sum_{i=1}^k \mu_i r_i \geq (C_k^{\alpha} + C_k^{\beta}) k^{1/2}.
\end{align}
\end{theorem}

\noindent {\bf Remarks: }
\begin{itemize}
    \item Condition~\eqref{eqn:suff} is (up to constants) necessary, because if $\alpha < 1-\beta$, the power of the martingale Stouffer test is less than $1-\beta$ whenever
    \begin{align*} 
        \forall~k \in \mathbb{N}: \sum_{i=1}^k r_i \mu_i \leq (C_k^{\alpha} - C_k^{1 - \beta}) k^{1/2}.
    \end{align*}
    
    \item The necessary condition~\eqref{eq:online_Bonferroni_power} under dense non-nulls requires a lower bound on $r_k\mu_k$ that decreases at the rate of~$\left(\log k\right)^{-1/2}$. This lower bound is fairly tight: for an example of sequence $\{\alpha_k\}_{k=1}^\infty$ that decreases at the rate of $1/[k (\log k)^2]$, the power of the online Bonferroni test would be one if all hypotheses are non-null when $k > 1$ and the mean value decreases at a slower rate: $\mu_k = \left(\log k\right)^{-1/c}$ for any $c > 2$ (see Lemma~\ref{lm:Bonf_example} in Appendix~\ref{apd:power_online_Bonf}).

    \item The proof of Theorem~\ref{thm:power_online} is in Appendix~\ref{apd:power_online_Bonf}. 
    If asymptotically, the mean values are nonzero but fade as $k$ grows at a fast rate, the online Bonferroni method has little power, but the martingale Stouffer test can have good power. For example, suppose all the hypotheses are non-nulls and $\mu_k = k^{-1/3}/10$. Controlling the type-I error $\alpha$ at $0.15$, the online Bonferroni method has power less than $0.6$ (by condition~\eqref{eq:online_Bonferroni_power}) whereas the \mst has power that approaches $1$ (by condition~\eqref{eqn:suff}).
     
\end{itemize}

\paragraph{The \postmt:} 
For clarity, we consider the same mean value for the non-nulls, $\mu_i = \mu$ if $r_i = 1$. Let a $Z$ score for each hypothesis~$H_i$ be $Z_i = \Phi^{-1}(1 - p_i)$.
Our guarantee on the power for the \postmt depends critically on the choice of the threshold parameter~$c$ (we consider Algorithm~\ref{alg:post_mt_online} with the filtering $\Phi^{-1}(1 - g(p_t)) > c$, which is equivalent to $g(p_t) < c'$ for $c' = 1 - \Phi(c)$). To concisely state our results, define the following quantities:

\begin{align*} \small
A(\mu; c) &= \frac{5}{3}\frac{\sqrt{\Phi(-c)}}{\Phi(\mu-c) - \Phi(-\mu-c)}, \\
B(\mu; c) &= \frac{10(\Phi(\mu-c) + \Phi(-\mu-c) - 2\Phi(-c))}{9(\Phi(\mu-c) - \Phi(-\mu-c))^2} \vee \frac{25}{ (\Phi(\mu-c) + \Phi(-\mu-c))^2}, \\
T(\beta; c) &= \frac{0.79\log (15.57/\beta)\Phi^2(-c) + 0.4}{\Phi^4(-c)}.
\end{align*}
For a reasonable choice of the threshold parameter, i.e., setting $c = \mu$ for instance, we note that the quantity $B(\mu; \mu)$ is upper bounded by a universal constant (when $\mu > 0$). On the other hand, the quantity $A(\mu; \mu)$ decays exponentially for large signal strength, i.e., when $\mu > 0.25$ we have:
\begin{align} \label{eq:online_amt_para}
A(\mu; \mu) \leq e^{-\mu^2/4}.
\end{align}
With these quantities in place, we now state our main result on the power of the \postmt.

\begin{theorem}
\label{thm:power_online_imt}
A sufficient condition for the \postmt with type-I error $\alpha$ and threshold parameter $c$ to have $1-\beta$ power is that: 
\begin{align*}
    \exists~k \geq T(\beta; c): \sum_{i=1}^k r_i\geq~& A(\mu; c)\left(C_k^{\alpha} + C_k^{\beta/3}\right) k^{1/2}{}\\
    &+ B(\mu; c) \left(C_k^{\alpha} + C_k^{\beta/3}\right)^2k^{-1/2}.
\end{align*}
\end{theorem}

    It is interesting to compare the above result with the necessary condition for the \mst: the power of MST is less than $1 - \beta$ if
    \begin{align} \label{eq:online_mst_ness}
        \forall~k \in \mathbb{N}: \quad  \sum_{i=1}^k r_i \leq \mu^{-1}\left(C_k^{\alpha} - C_k^{1 - \beta}\right) k^{1/2}.
    \end{align}
    Both right-hand sides grow at the rate of $k^{1/2}$ (ignoring $\log \log$ factors), but the $\mu$-dependent term $\exp(-\mu^2/4)$ for AMT (derived in bound~\eqref{eq:online_amt_para} for $A(\mu; \mu)$) is much smaller than the corresponding $1/\mu$ term in condition~\eqref{eq:online_mst_ness} for MST. As a consequence, the \postmt will have higher power when the non-nulls have sufficiently large mean values but are sparse.

To summarize the basic insights we derive in this section, we find that both in the batch setting and the online setting, the \mst and the \postmt require weaker conditions than their classical counterparts to guarantee the same power when the non-nulls are sparse. The \mst relies on good prior knowledge to order the hypotheses, while the \postmt uses masked \pvalues to generate a good ordering.
The theoretical analyses in this section discuss the case with no prior knowledge, and the simulations in the next section delve deeper into the setting where the non-nulls are structured.

\section{Numerical simulations} \label{sec:ipm}

While the \mst can only use prior knowledge in the form of non-null probabilities for each hypothesis, the \imt combines (a) side covariate information (which could include prior non-null probabilities in working model~\eqref{eq:p_model} as a component) with (b) structural constraints on the unknown non-null set, and (c) masked \pvalues, to infer whether a hypothesis is non-null and thus include it earlier in the ordering. Here, we demonstrate that prior structural constraints can help the \imt attain a higher power than the \mst and some classical methods. 

We first consider the batch setting and use two non-null structures as simple examples: a blocked structure within a grid and a hierarchical structure within a tree; and we discuss similar structures in the online setting. For each of these, we customize a heuristic strategy to expand $M_k$ in the \imt (recalling that type-I error is controlled regardless of the heuristic used, and only power is affected).

\subsection{Clustered non-nulls in a grid of hypotheses} \label{sec:sim_cluster}

Consider the setting where the hypotheses are arranged in a rectangular grid, and if the null is false, then the non-nulls form a single coherent cluster.
This is a common structure which, as a hypothetical example, is a reasonable belief when trying to detect if there is a tumor in a brain image. Here, the covariates~$x_i$ are simply the two-dimensional location of the hypothesis $H_i$ on the grid. The blocked non-null structure is utilized in specifying the posterior probability of being non-null using model~\eqref{eq:p_model} by constraining the prior non-null probabilities~$\pi_i$ to be a smooth function of $x_i$. Details can be found in Appendix~\ref{apd:em}.

The block structure is also imposed in the strategy of interactively expanding~$M_k$ such that $M_k$ forms a single connected component. The \imt expands $M_k$ by only including possible non-nulls that are on the boundary of $M_k$ (see Figure~\ref{fig:include} for example). 

\begin{figure}[h]
    \centering
    \begin{subfigure}[t]{0.24\textwidth}
        \centering
        \includegraphics[width=1\linewidth]{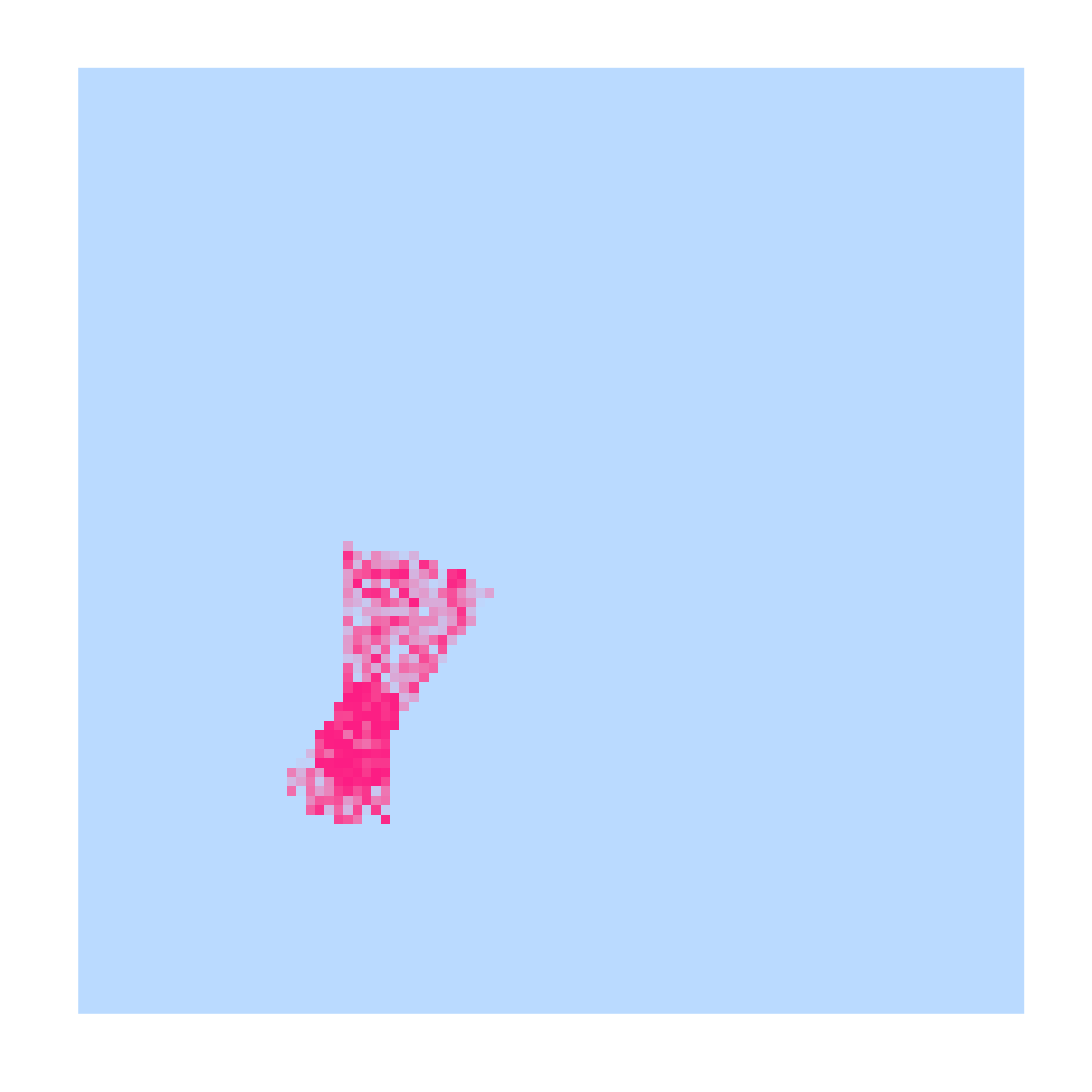}
    \end{subfigure}
    \hfill
    \begin{subfigure}[t]{0.24\textwidth}
        \centering
        \includegraphics[width=1\linewidth]{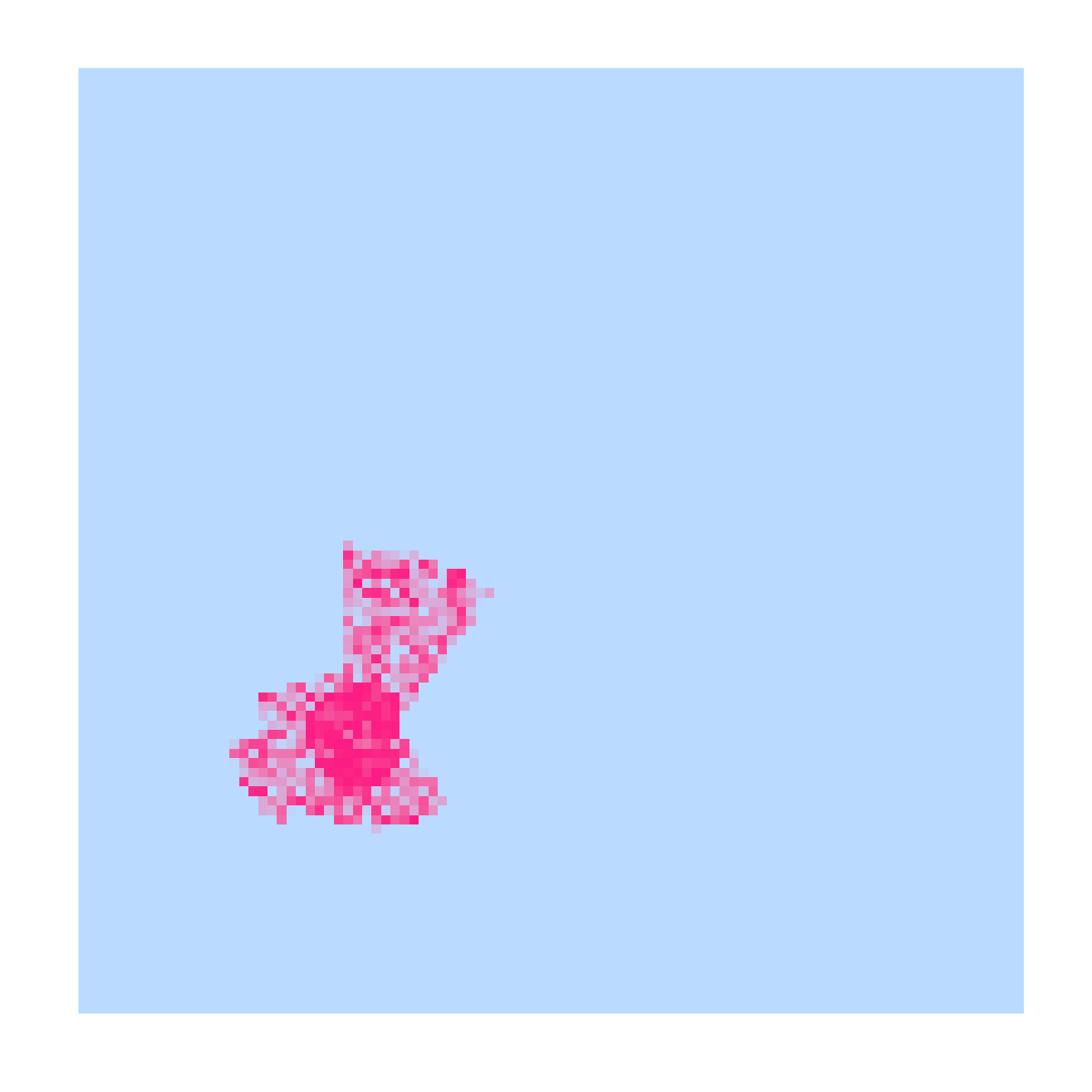}
    \end{subfigure}
    \hfill
    \begin{subfigure}[t]{0.24\textwidth}
        \centering
        \includegraphics[width=1\linewidth]{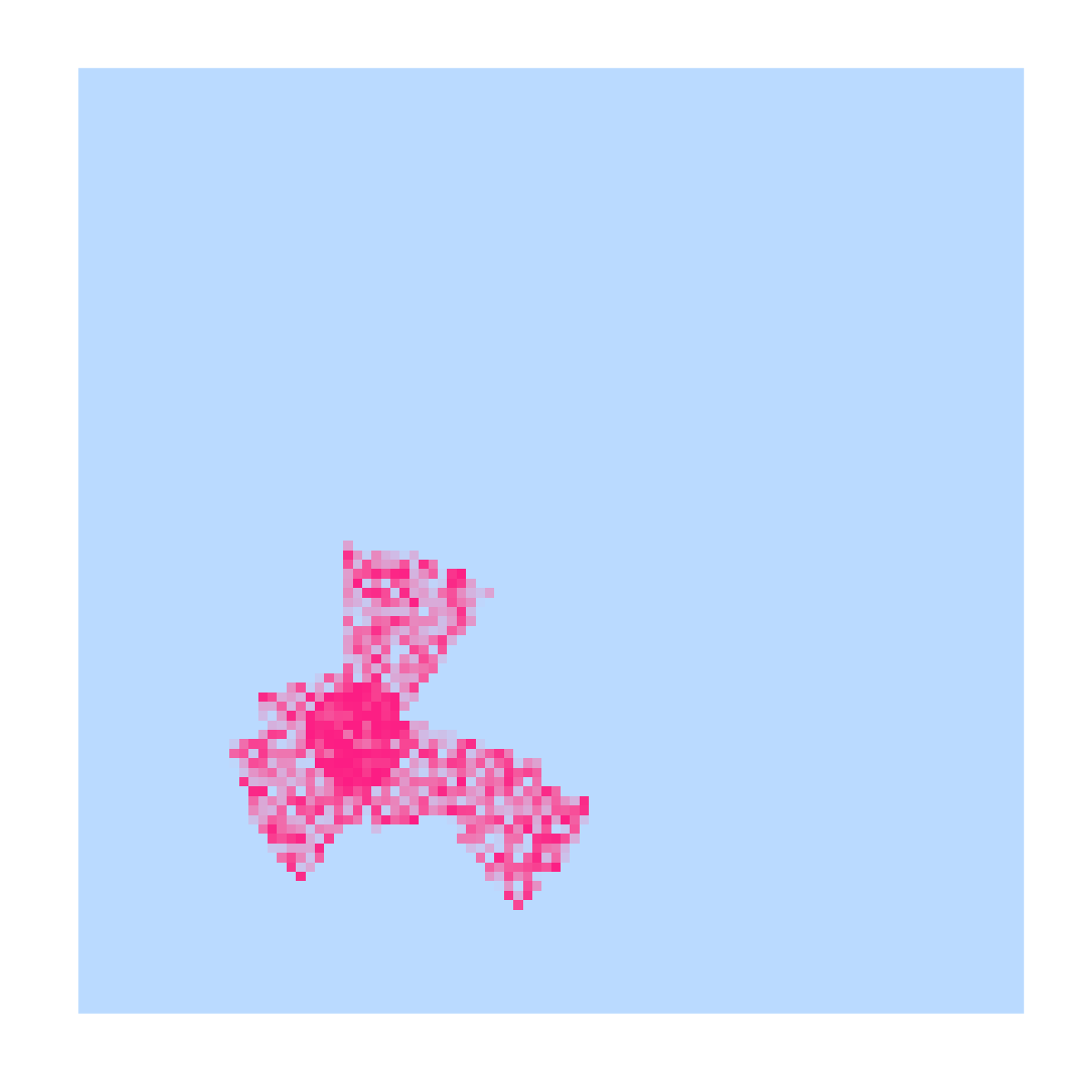}
    \end{subfigure}
    \hfill
    \begin{subfigure}[t]{0.24\textwidth}
        \centering
        \includegraphics[width=1\linewidth]{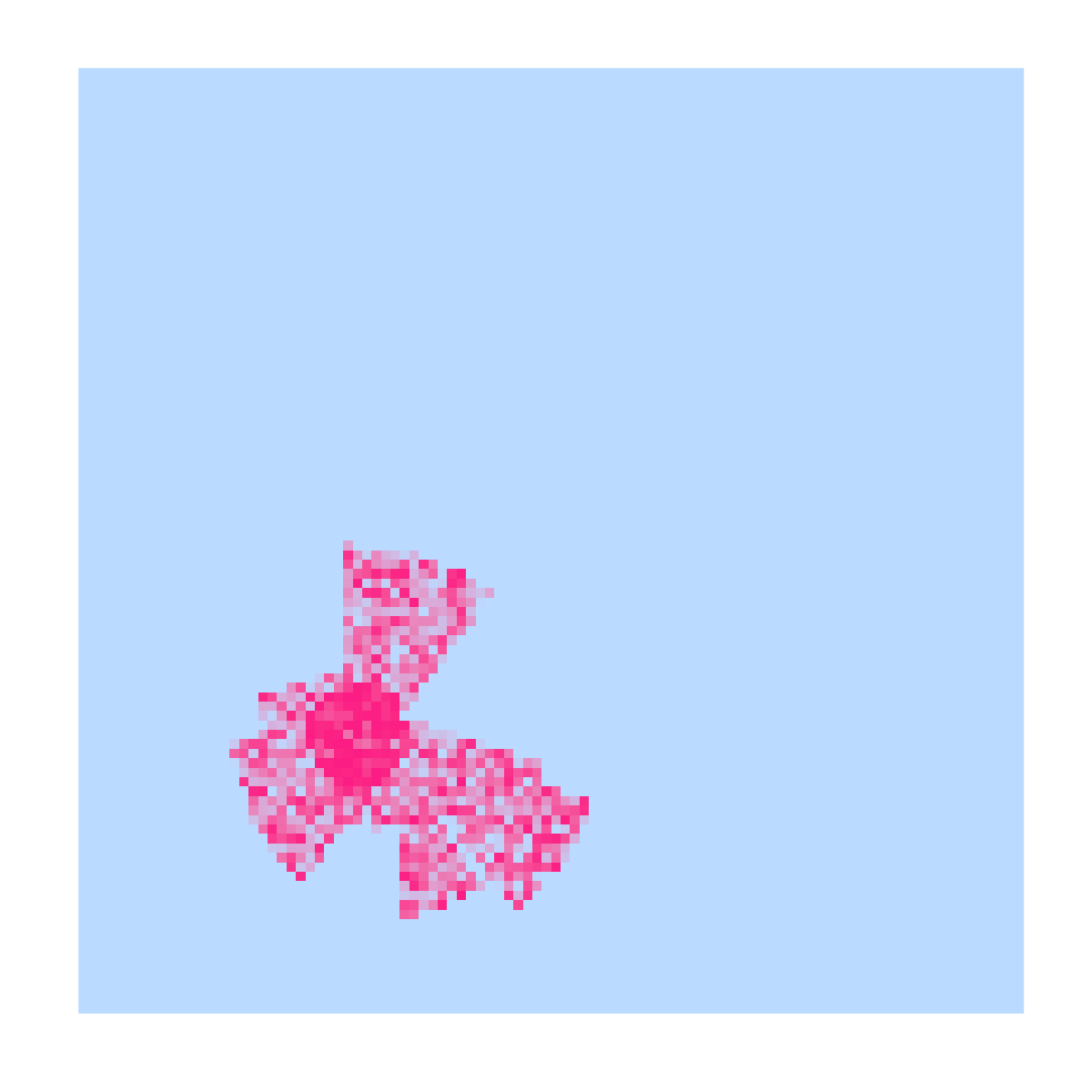}
    \end{subfigure}
    
    \caption{Visualization of the \imt under the block structure: the hypotheses in $M_k$, which interactively expands (darker color indicates a lower $p$-value and possible non-null).}
    \label{fig:include}
\end{figure}

We compare the \imt with the \mst and the batch Stouffer test. We use the martingale Stouffer test (MST) with a preordering that starts at the center of the grid, and the following hypotheses are included into the preordering in randomly chosen (data-independent) directions such that the hypotheses always form a single cluster. Our simulation has $10^4$ hypotheses arranged in a $100 \times 100$ grid with a disc of about 150 non-nulls, placed either at the grid center and or at a corner of the grid. We use Setting~\ref{set:simple} as defined in Section~\ref{sec:power}, where we varied the non-null mean as $(0, 0.3, 0.6, 0.9, 1.2, 1.5, 1.8)$. 

\begin{figure}[t]
    \centering
    \begin{subfigure}[t]{0.48\textwidth}
        \centering
        \includegraphics[width=0.8\linewidth]{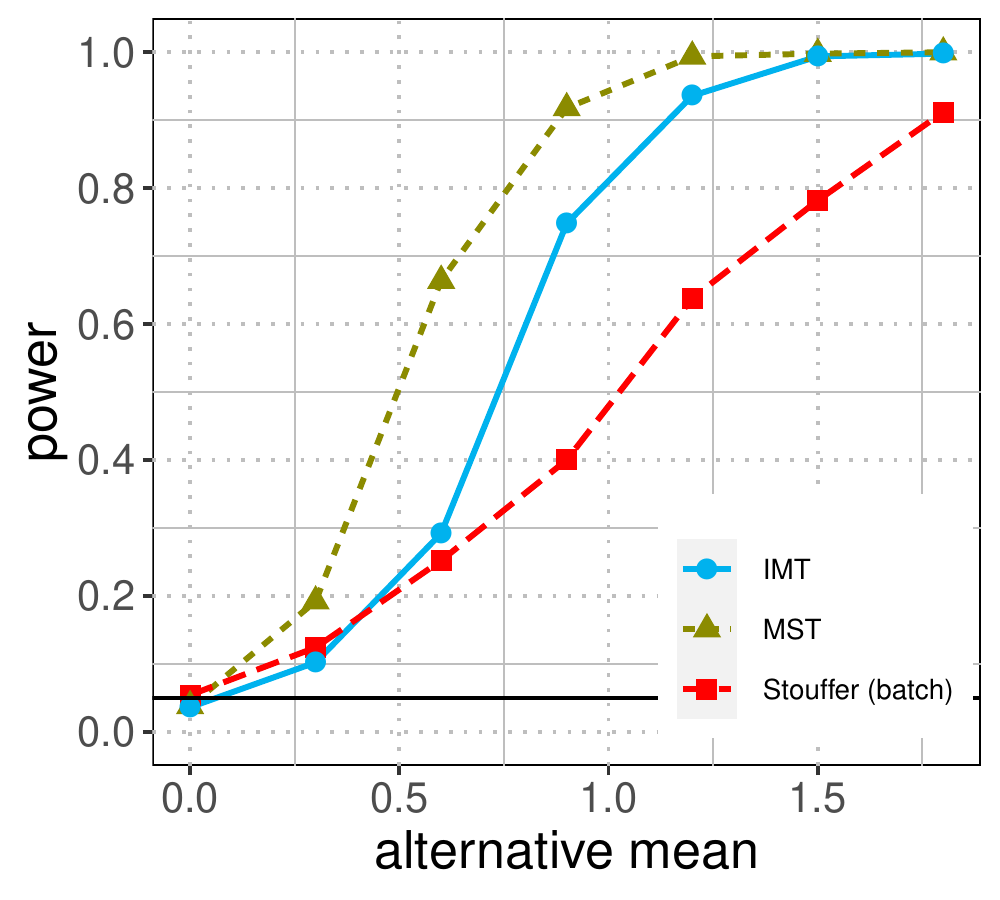}
        \caption{The power against non-null signal. The non-null block is in the grid center. 
        }
        \label{fig:center}
    \end{subfigure}
    \hfill
    \begin{subfigure}[t]{0.48\textwidth}
        \centering
        \includegraphics[width=0.8\linewidth]{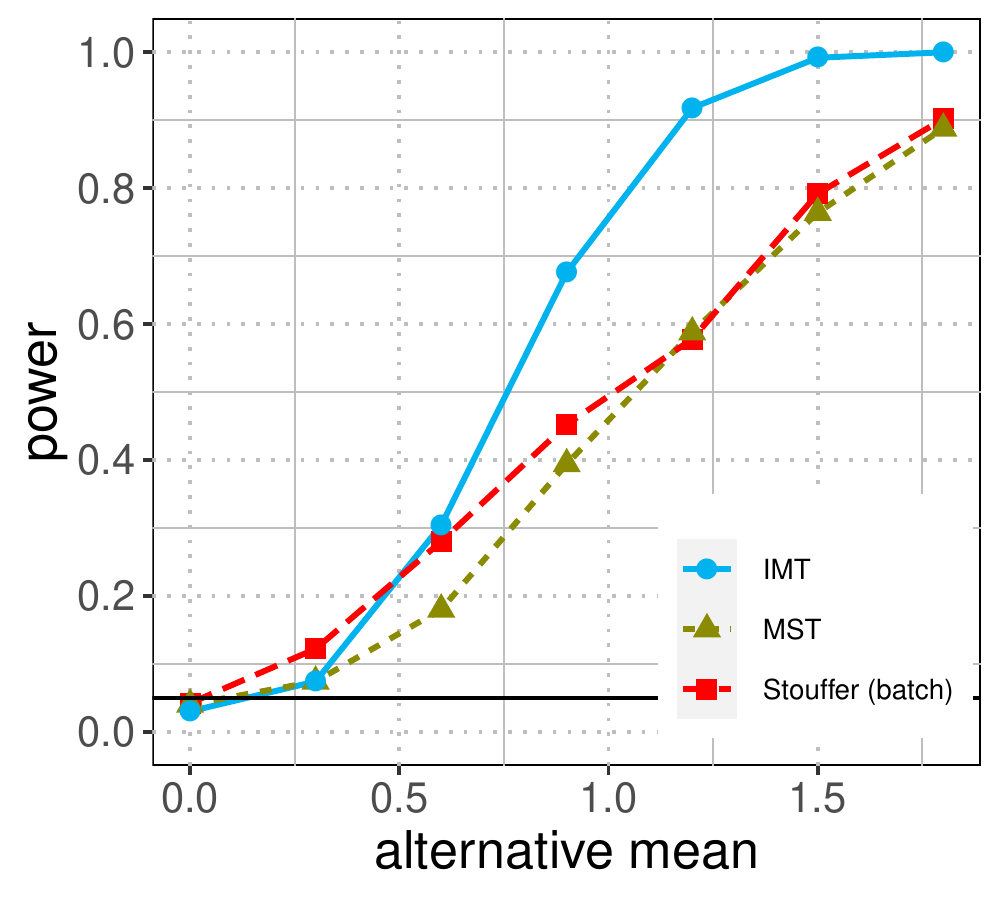}
        \caption{The power against non-null signal. The non-null block is in the grid corner.
        }
        \label{fig:uncenter}
    \end{subfigure}
    
    \caption{Testing the \imt (IMT), the \mst (MST), and the batch Stouffer test with varying alternative mean under a block non-null structure (batch setting). The MST has lower power when the non-null is not in the center, whereas the IMT has high power in both cases. Type-I error corresponds to the power when the alternative mean value is zero. The horizontal line corresponds to the target type-I error level $\alpha = 0.05$.}
    \label{fig:block}
\end{figure}

The \imt has high power for both positions of the non-null block, whereas the power of \mst drops quickly when the block is not at the center (Figure~\ref{fig:block}), which is because the \mst does not have information of the block position (its preordering starts from the center by default), whereas the \imt uses masked \pvalues to learn the block position. It is worth noting that even with a bad preordering, the martingale Stouffer test does not do worse than the batch version, but has much higher power with a good preordering.

\begin{remark}
As mentioned in the introduction, we do not intend to claim that the \imt is in any sense the ``best'' test for this problem setting. It is possible, or even likely, that several other generic tests (Bonferroni, chi-squared, higher criticism, or many others) or specialized tests (scan statistics) might have higher power. We discuss the comparison with two recent methods: the adaptively weighted Fisher test \citep{li2011adaptively,fang2019properties,huo2020p} and the weighted Higher Criticism \citep{zhang2020incorporating} in Appendix~\ref{apd:alter_methods}. Our goal in this section is to demonstrate the tradeoffs between the batch and martingale versions of the same test (Stouffer in this case), and the interactive versus preordered martingale tests.  Also note that the power of our martingale tests depends crucially on the preordering, or on the model and heuristic used to form the ordering interactively, and perhaps better models/algorithms might further improve the power of our own tests. We chose settings that are easy to visualize for intuition, keeping in mind that our tests apply to any general covariates $x_i$, and prior knowledge or structural constraints, any working models, etc.
\end{remark}

\subsection{A sub-tree of non-nulls in a tree of hypotheses} \label{sec:ipm_tree}

In applications such as wavelet decomposition, the hypotheses can have a hierarchical structure, where the child can be a non-null only if its parent is a non-null. 
The hierarchical structure is again encoded in modeling the posterior probability of being non-null~\eqref{eq:p_model} by adding a partial order constraint on $\pi_i$ that
$$
\pi_i \geq \pi_j, \quad \text{if $i$ is the parent of $j$}.
$$
Also, the hierarchical structure is imposed in the strategy of update $M_k$ such that $M_k$ should keep as a sub-tree. Specifically, we compare the posterior probabilities of being non-null for all the leaf nodes of $M_k$ and choose the highest one. 

We compare the \imt with the \mst and Stouffer's test, where the \mst order the hypotheses by level and from left to right within level. We simulate a tree of five levels (the root has twenty children and three children for each parent node after that) with over 800 nodes in total and 7 of them being non-nulls. Each node tests if a Gaussian is zero mean as described in Setting~\ref{set:simple}, where we vary the mean value for the non-nulls as $(0, 0.5, 1, 1.5, 2, 2.5, 3, 3.5, 4)$. The \imt is implemented without modeling the posterior probabilities of being non-null for the sake of computational cost. The \imt has a higher power especially when the signal is strong so that the masked \pvalues provide a better guide on the $M_k$ update (Figure~\ref{fig:batch_tree_power}).

\begin{figure}[h!]
    \centering
        \includegraphics[width=0.45\linewidth]{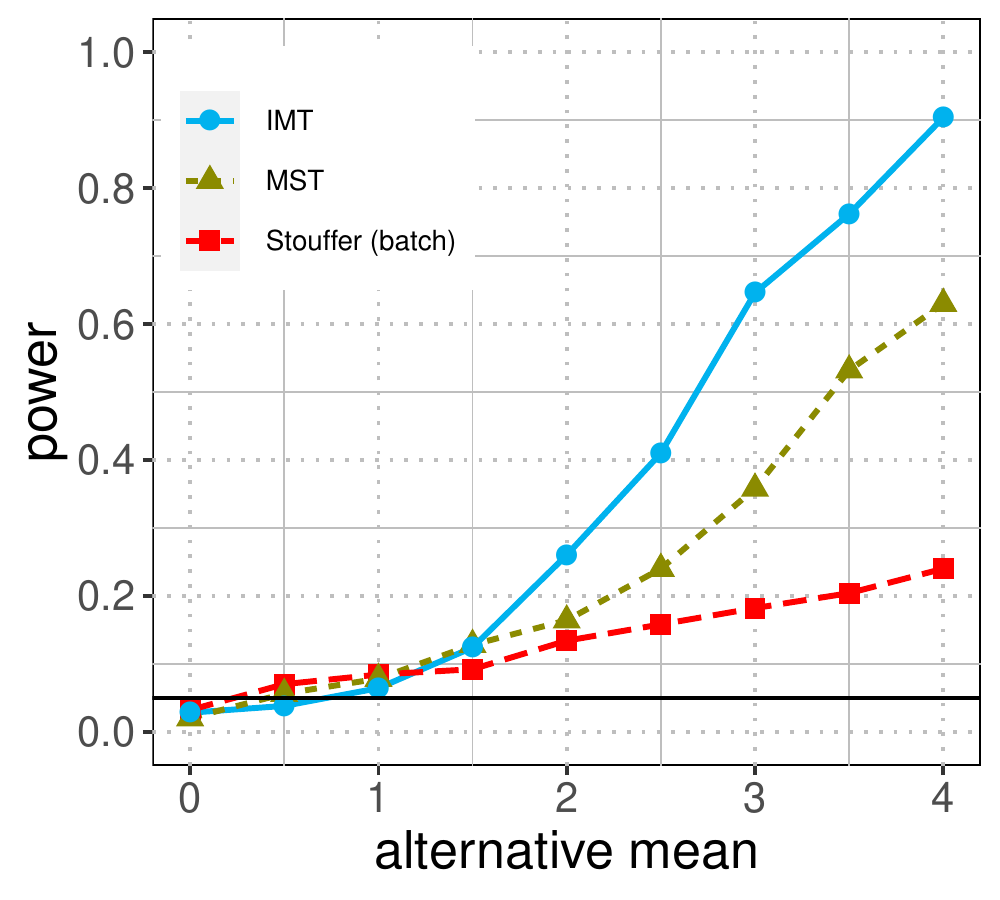}
        \caption{Power of the \imt (IMT), the \mst (MST), and the batch Stouffer test under a hierarchical structure. Hypotheses form a fixed tree (batch setting) with non-nulls only on a sub-tree. When the alternative mean is big, masked \pvalues and the hierarchical non-null structure lead to a good ordering and hence high power for the IMT.}
        \label{fig:batch_tree_power}
\end{figure}

\begin{figure}[h!]
    \centering
     \begin{subfigure}[t]{0.3\textwidth}
        \centering
        \includegraphics[width=1\linewidth]{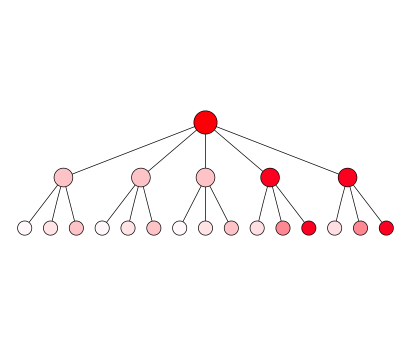}
        \caption{Hypothesis tree with decreasing non-null probability, which is marked by fading red nodes.}
        \label{fig:dec_tree}
    \end{subfigure}
    \hfill
    \begin{subfigure}[t]{0.33\textwidth}
        \centering
        \includegraphics[width=1\linewidth]{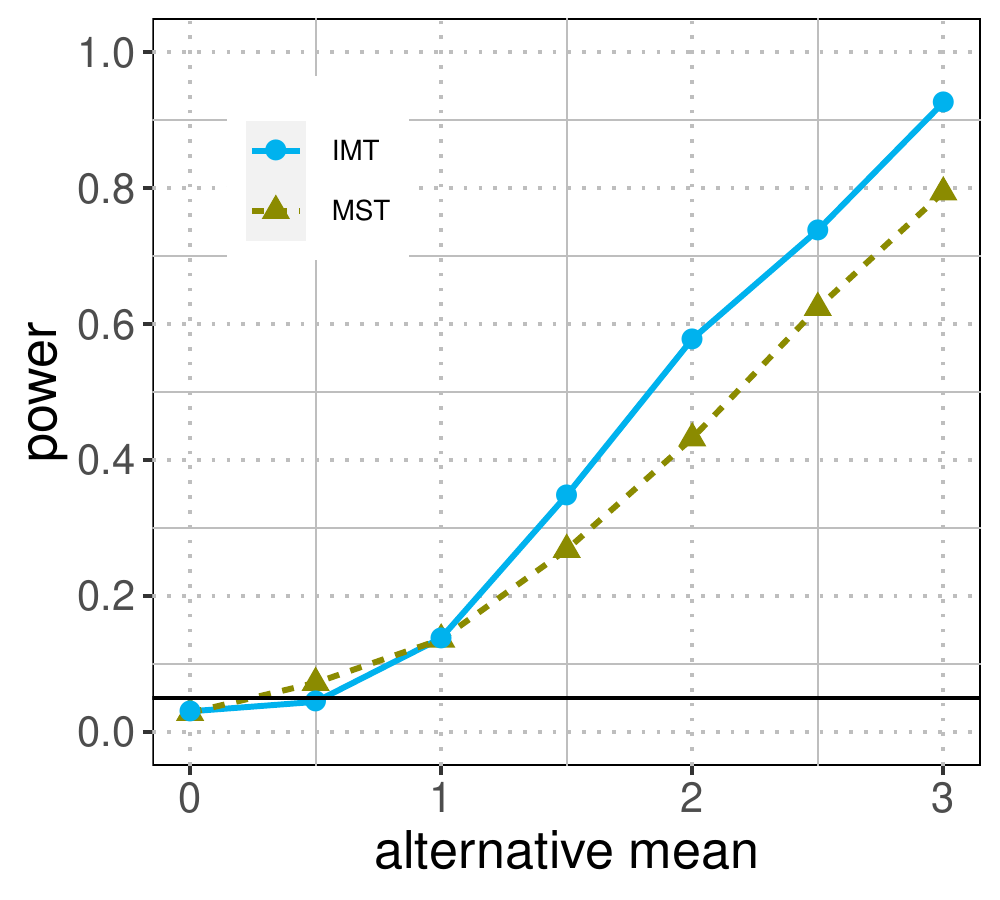}
        \caption{Power against alternative mean in a hypothesis tree with decreasing probability of being non-null.}
        \label{fig:dec_tree_power}
    \end{subfigure}
    \hfill
    \begin{subfigure}[t]{0.33\textwidth}
        \centering
        \includegraphics[width=1\linewidth]{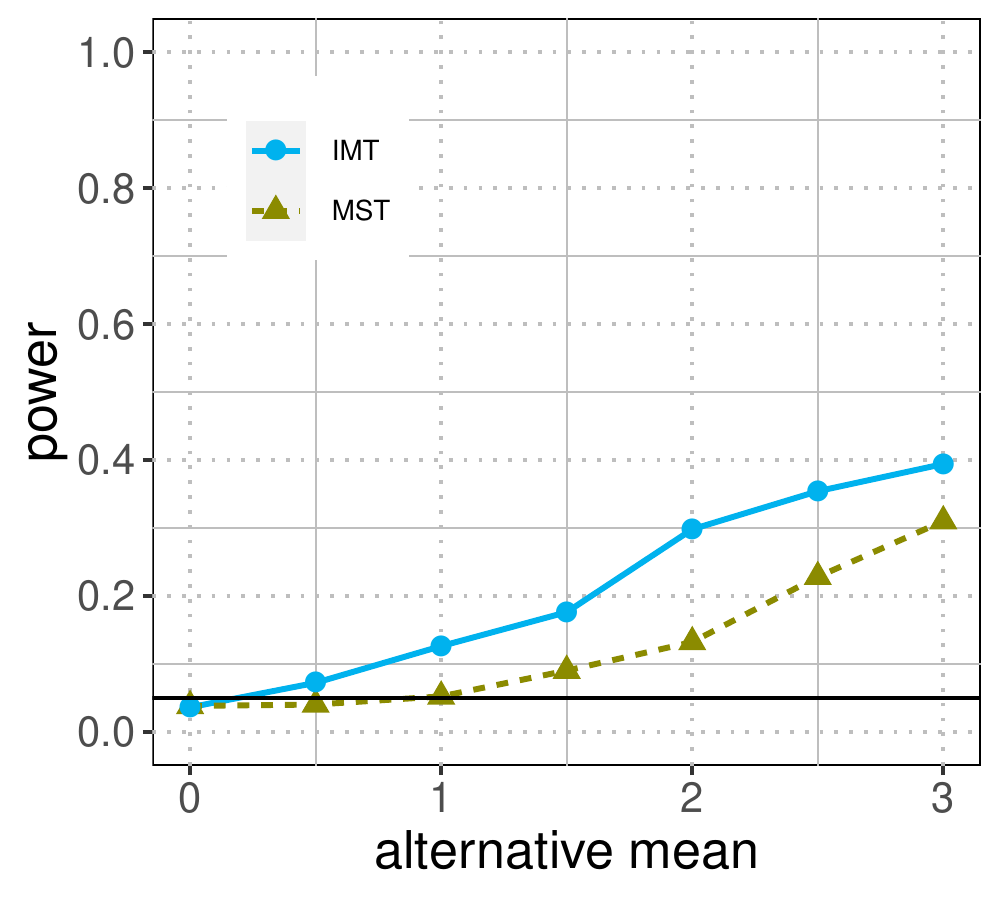}
        \caption{Power against alternative mean in a hypothesis tree with increasing probability of being non-null.}
        \label{fig:inc_tree_power}
    \end{subfigure}
    
    \caption{Hypothesis tree in the batch setting with decreasing/increasing probability of being non-null. Testing the \imt (IMT) with a model for the posterior probability of being non-null, which has higher power than the \mst (MST) in both cases.}
    \label{fig:treeS}
\end{figure}

The \imt with modeling is implemented on a smaller tree with 121 nodes (five levels and three children for each parent node) and 7 of them being non-nulls on a subtree. We consider two types of hierarchical non-null structure: one with the probability of being non-null decreasing down the tree, and one with increasing probability, which means the parent cannot be a non-null unless its children are non-nulls. The result is consistent with the above: the \imt has higher power than the non-interactive \mst (Figure~\ref{fig:treeS}). Compared with decreasing probability of being non-null, both methods have lower power for the tree with an increasing probability of being non-null, because in the latter case, the non-nulls gathered at later generations where there are more nulls and the non-nulls are sparser.

\subsection{Structures in the online setting} 
Recall that in the online setting, a potentially infinite number of hypotheses arrive, and the \postmt and \imt use some discarding rules to only allow promising non-nulls entering $M_k$. This section presents two examples of non-null structures in the online setting, and demonstrates the power of the interactive test as follows.

\paragraph{Blocks of non-nulls in a growing sequence of hypotheses.} \label{sec:interactive_online}
Suppose the non-nulls arrive as blocks. In other words, the next hypothesis is more likely to be a non-null if the last arrived hypothesis is truly non-null; and vise versa. Let the discarding rule in the \imt be $g(p_t) > c_t$, where $c_t = c = 0.05$ by default. The \imt adjusts $c_t$ for $t > 10$ based on previous $p$-values: it alleviates the discarding rule by increasing $c_t$ to $2c$ if the ten $p$-values prior to $t$ ($p_{t-10}, \ldots, p_{t-1}$) are all less than $0.1$; otherwise, it decreases $c_t$ to~$c/4$. For a fair comparison, the discarding threshold in the \postmt is set to $c = 0.05$.

\begin{figure}[ht]
    \centering
    \includegraphics[width=0.4\linewidth]{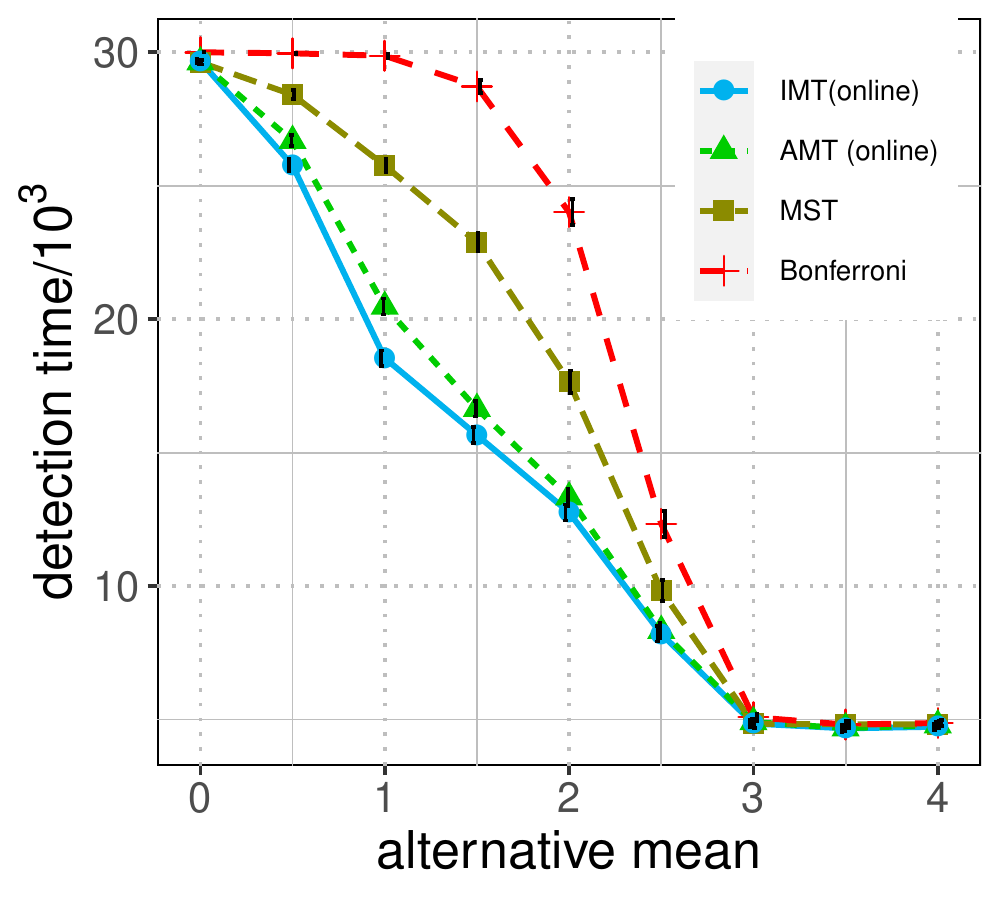}
    \caption{Number of hypotheses needed to reject the global null (detection time) in the online setting of the \imt (IMT), the \postmt (AMT), the \mst (MST), and the Bonferroni test when varying the alternative mean $\mu$. The non-nulls arrive in blocks, and on average, every $10^4$ hypotheses contain a block of $500$ non-nulls. The length of the error bar is two standard error. The \imt is the first to reject the global null because it incorporates the block structure and adjusts the discarding threshold based on past $p$-values.}
    \label{fig:imt_online}
\end{figure}

The \imt is the first to reject the global null since its discarding rule accounts for the block structure (see Figure~\ref{fig:imt_online}). This advantage is more evident when the non-null signal is mild ($\mu < 3$), where the prefixed discarding rule in the \postmt might be too strict or lenient, while the \imt can adjust the rule accordingly. In practice, the adjustment on the discarding threshold can also utilize side information and prior knowledge, if provided.

\paragraph{A sub-tree of non-nulls in a growing tree of hypotheses.}

The online tree grows a new level at every step, with the probabilities of being non-null no bigger than their parents. For an arriving level $k$, the \imt models the posterior probability of being non-null $\pi_j^{(k)}$ for the new hypothesis $H_j$ by equation~\eqref{eq:p_model}, where the prior probability of being non-null is the same as its direct parent $H_i$ from the level $k-1$,
$$
\pi_j^{(0)} = \pi_i^{(k-1)}, \quad \text{if $i$ is the parent of $j$}.
$$
For simplicity, we set the discarding rule in the \imt to be $\pi_i^{(k)} < c$ where $c = 0.6$ as a default. That is, hypothesis with $\pi_i^{(k)} < 0.6$ are omitted. We compare the \imt with the \mst and a classical method, the online Bonferroni method (with the sequence of significance levels $\{\alpha_k\}_{k=1}^\infty$ decreases at the rate of $1/[k(\log k)^2]$). In the online setting, their performances are assessed by the averaged number of hypotheses required to reject the global null (detection time); the smaller the better. 

\begin{figure}[ht]
    \centering
        \centering
        \includegraphics[width=0.45\linewidth]{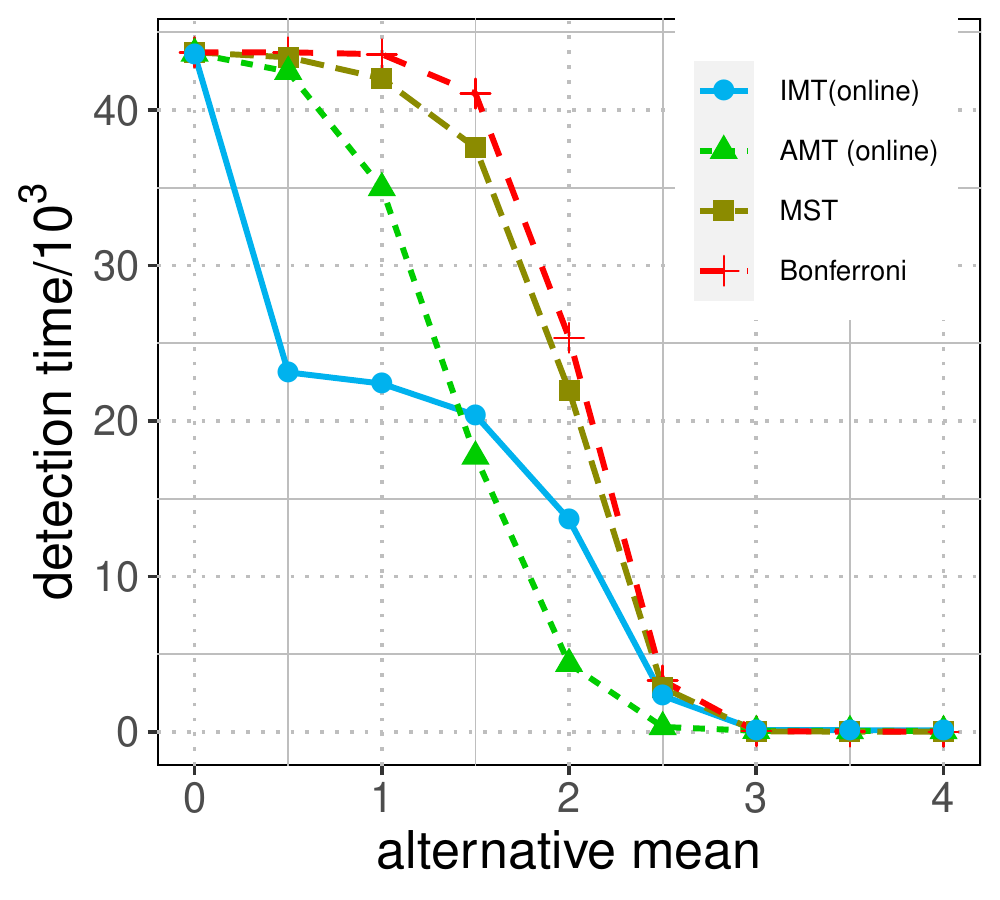}
        \caption{Number of hypotheses needed to reject the global null (detection time) in the online setting of the \imt (IMT), the \postmt (AMT), the \mst (MST), and the Bonferroni test when varying the alternative mean in a growing hypothesis tree (online setting). IMT incorporates the hierarchical structure of non-nulls, so it is the first to reject the global null when the non-null signal is mild ($\mu < 2$).}
    \label{fig:online_tree_power}
\end{figure}

We simulate the online tree with forty children for the root node and three children for each parent node after that. The probability of being non-null for the first generation children is set to $0.1$ for 30 children and $0.9$ for the other 10 children. The ongoing three children of each node reduce the probability of being non-null as by a proportion of $100\%, 20\%, 0\%$. Each node tests if a Gaussian is zero mean as described in Setting~\ref{set:simple}, where we vary the mean value for the non-nulls as $(0, 0.5, 1, 1.5, 2, 2.5, 3, 3.5, 4)$. The \imt needs much shorter time when the non-null signal is not strong ($\mu < 2$) because it incorporates the hierarchical structure and estimates the probability of an arriving hypothesis being non-null with the aid of the data from its ancestors (Figure~\ref{fig:online_tree_power}). When the alternative mean is large, $p$-values themselves provide strong evidence of non-null, while the algorithm using the tree structure would treat all children from a non-null parent as promising non-nulls while at least one of them is null in our simulated example. Thus, the online AMT that uses only the $p$-value information can have better performance when the alternative mean is large.
 
Overall, both in the batch setting and the online setting, the \imt has a higher detection power than the \mst, Stouffer's test, and the online Bonferroni method, provided with structured alternatives. We again remark the advantage of the \imt in practice where prior knowledge often exists in various forms. The \imt is highly flexible in that it allows modifications to the strategy of expanding $M_k$, at any step and with any form as a human analyst (or a program) wants to.
The next section demonstrates one more advantage of the \imt under the \textit{conservative} nulls (see definition in the next section).

\section{Robustness to conservative nulls} \label{sec:robust}

In all the above simulations, the nulls have uniformly distributed $p$-values, but in practice they could be stochastically larger than uniform (condition \eqref{cond:stoch_dom}) or mirror-conservative (condition \eqref{cond:mirror_consv}); both are henceforth referred to as ``conservative nulls''. For simplicity, this section focuses on the conservative null with an increasing density, which satisfies both descriptions in condition~\eqref{cond:stoch_dom} and condition~\eqref{cond:mirror_consv}. Such conservative nulls diminish the detection power of many batch global null tests like Fisher's and Stouffer's methods. For example, each term in Stouffer's test is $\Phi^{-1}(1 - p)$, whose value can be smaller than $-2$ if the $p$-value is bigger than $0.98$; thus as the nulls grow more conservative and their $p$-values closer to one, its power can quickly drop to zero. 

To examine the effect of conservative nulls on the \imt, we first propose an alternative definition of a masked $p$-value as $\widetilde g(p) := \min(p, (p+\tfrac12) \text{mod} 1)$. Recalling that $g(p)=\min(p,1-p)$, we call $g$ and $\widetilde g$ as the tent and railway functions respectively (see Figure~\ref{fig:oldG}, Figure~\ref{fig:newG}). Note that if the $p$-value is exactly uniformly distributed, $\widetilde g(p)$ is still independent of $h(p)$, and $g(p)$ has the same distribution as $\widetilde g(p)$, and so all previous results still hold with the new masking function in place of the old one. (The error control when using the railway masking function can be found in Appendix~\ref{apd:error_control_railway} for uniform and conservative $p$-values.) However, when the $p$-values are conservative, the new masking function has a clear advantage. To see this, consider a $p$-value of 0.99. The original masked $p$-value would be 0.01, thus causing the methods to potentially confuse this with a non-null masked $p$-value, but the new masked $p$-value would be 0.49, which the methods would easily exclude as being a null. 

\begin{figure}[h]
    \centering
    \begin{subfigure}[t]{0.3\textwidth}
        \centering
        \includegraphics[width=1\linewidth]{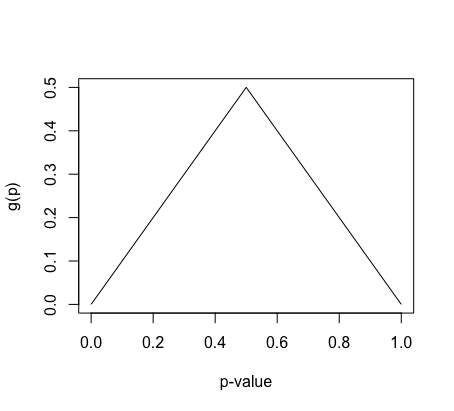}
        \caption{The original masking function (tent). }
        \label{fig:oldG}
    \end{subfigure}
    \hfill
    \begin{subfigure}[t]{0.3\textwidth}
        \centering
        \includegraphics[width=1\linewidth]{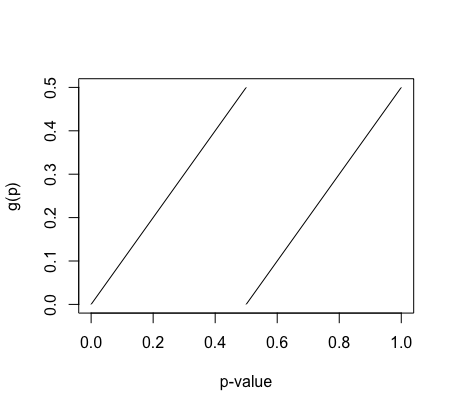}
        \caption{The modified masking function (railway). }
        \label{fig:newG}
    \end{subfigure}
    \hfill
    \begin{subfigure}[t]{0.35\textwidth}
        \centering
        \includegraphics[width=0.8\linewidth]{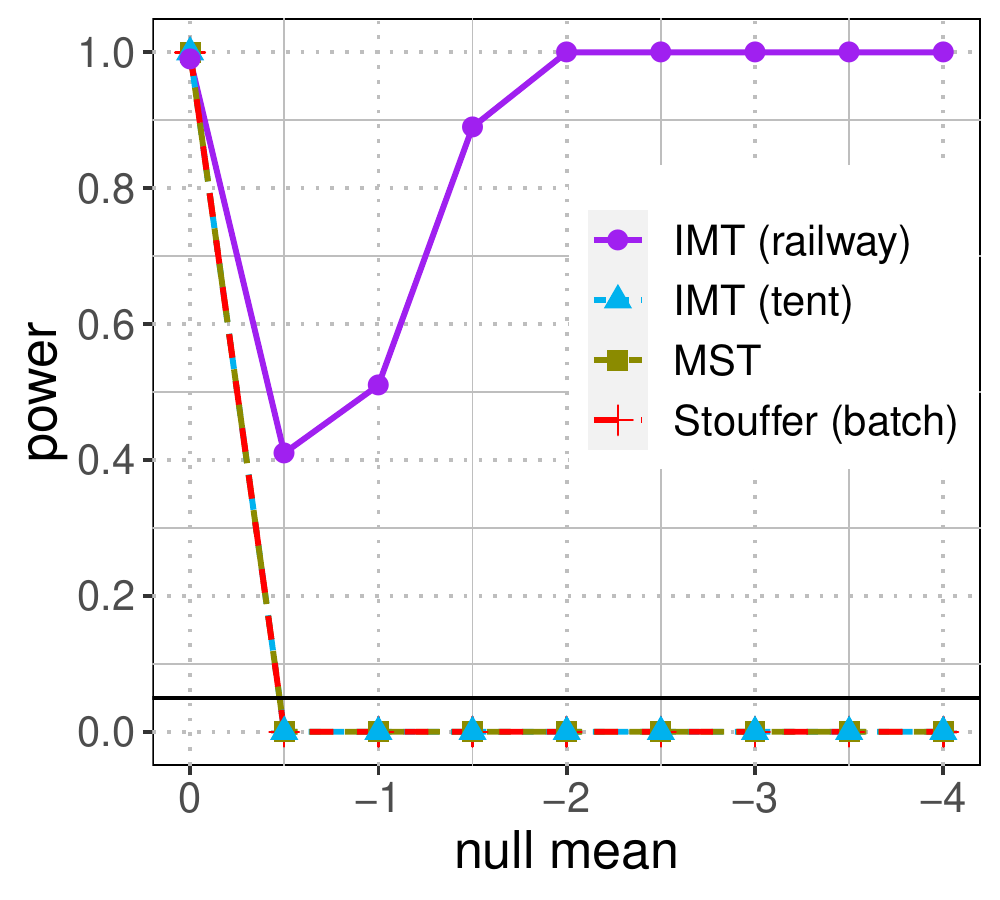}
        \caption{Power when varying the null means (negative values indicate conservative nulls).
        }
        \label{fig:robust}
    \end{subfigure}
    
    \caption{Comparing the \imt (IMT) with tent and railway masking  functions, the \mst (MST), and Stouffer's test for the robustness to conservative nulls. The IMT with railway function is more robust. 
    }
    \label{fig:conservative}
\end{figure}

As an example, we consider the simple case with no prior knowledge and simulate 1000 hypotheses with 100 non-nulls. Each hypothesis is a one sided hypothesis on whether a Gaussian is zero mean as described in Setting~\ref{set:simple}. The alternative mean values are set to 1.5. The mean values for nulls are negative so that the resulting null $p$-values are conservative. We tried nine values from $0$ to $-4$ for the mean of nulls, with a smaller value indicating higher conservativeness. Figure~\ref{fig:robust} compares the power of the interactive martingale test with tent and railway functions, the martingale Stouffer test and Stouffer's test. The power of most tests drops sharply to zero, but the power of \imt with the new railway function initially dips and then improves. The reason for the initial dip is that the increasingly conservative nulls influence the interactive martingale test in two opposite directions: (a) more null $h(p)$ values are now equal to $-1$ (instead of being $\pm1$ with equal probability), and this hurts power because including a null $h(p)$ in the martingale almost always lowers its value (instead of increasing and lowering its value with equal probability), (b) as the $p$-value gets more conservative, $g(p)$ will approach $0.5$ for nulls, allowing the tests to easily distinguish between the non-nulls and the nulls to increase the power. When the $p$-values are only slightly conservative,  effect (a) dominates and hurts power, causing the initial dip in power in Figure~\ref{fig:robust}.

\section{Anytime-valid \pvalues and safe $e$-values} \label{sec:anytime}
In this paper, we defined the problem as testing the global null at a predefined level $\alpha$. Instead, we could ask the test to output a sequential or anytime $p$-value for the global null, which is a sequence of \pvalues $\{\mathfrak{p}_t\}_{t=1}^\infty$ that are valid at any stopping time. We use $\mathfrak{p}_t$ to differentiate it from $p_t$ --- the latter is the input to our global null test, the former is the desired output of our global null test. Specifically, $\mathfrak{p}_t$ is a function of $p_1,\dots,p_t$, such that if $p_1,\dots,p_t$ are all null, then $\mathfrak{p}_t$ will be a valid $p$-value (its distribution will be stochastically larger than uniform), and this fact will be true uniformly over $t$.  

Recall that all of the proposed procedures follow the same form; we reject the global null if
\begin{align*}
    \exists k \in \{1,2,\ldots\} \text{ s.t. } S_k > u_\alpha(k),
\end{align*}
where $S_k$ is a martingale under the global null and $u_\alpha(k)$ is a sequence of upper bounds at level $\alpha$. The anytime \pvalue $\mathfrak{p}_t$ at time $t$ is defined by the smallest level at which our test would have rejected the null at or before time $t$.

\begin{definition}
The \pvalue $\mathfrak{p}_t$ can be defined as the smallest level $\alpha$ at which the test would have rejected at or before time $t$:
\begin{align}
    \mathfrak{p}_t = \inf\{\alpha: \exists k \in \{1,\ldots, t\} \text{ s.t. }  S_k > u_\alpha(k)\}.
\end{align}
\end{definition}

\noindent Viewing  $u_\alpha(k)$ as a function of two variables $k, \alpha$, we define an inverse function at a fixed $k$ with respect to the level $\alpha$ as
\begin{align*}
    u^{-1}(S; k) = \alpha \text{ iff } u_\alpha(k) = S,
\end{align*}
which is unique for a given input $S$ since the bound $u_\alpha(k)$ is continuous and strictly decreasing in $\alpha$. 
Then the \pvalue at time $t$ can be computed as
\begin{align*}
    \mathfrak{p}_t = \min_{1\leq k\leq t}\{ u^{-1}(S_k; k)\}.
\end{align*}

\noindent As one example, if $u_\alpha(k)$ is the linear bound as in test~\eqref{test:mst_lin}, its inverse is
\begin{align*}
    u^{-1}(S; k) = \exp\left\{-2m\frac{S^2}{(k + m)^2}\right\}.
\end{align*}

\noindent The \pvalue sequence $\{\mathfrak{p}_t\}_{t=1}^\infty$ has the following nice properties,
\begin{enumerate}
    \item the anytime \pvalues decrease with time:
    \begin{align*}
        \mathfrak{p}_{t + j} \leq \mathfrak{p}_t \text{ for all } j,t > 0.
    \end{align*} 
    
    \item $\inf_{t \in \mathcal{I}} \mathfrak{p}_t$ is also a valid \pvalue for the global null:
    \begin{align*}
        \mathbb{P}(\inf_{t \in \mathcal{I}} \mathfrak{p}_t \leq x) \leq x ~\equiv~ \mathbb{P}\{\exists t: \mathfrak{p}_t \leq x\} \leq x , \quad \text{ for all } x \in (0,1).
    \end{align*}
    In fact $\inf_{t \in \mathcal{I}} \mathfrak{p}_t$ is the global \pvalue: the smallest level $\alpha$ at which the test would ever reject:
    \begin{align*}
         \inf_{t \in \mathcal{I}} \mathfrak{p}_t = \inf\{\alpha: \exists k \in \{1,2,\ldots\} \text{ s.t. } S_k > u_\alpha(k)\}.
    \end{align*}
    
    \item for any arbitrary stopping time $\tau \in \mathcal{I}$, $\mathfrak{p}_\tau$ is a valid \pvalue:
    \begin{align*}
        \mathbb{P}(\mathfrak{p}_\tau \leq x) \leq x, \quad \text{ for all } x \in (0,1).
    \end{align*}
\end{enumerate}
The second property implies that the \pvalue at any time $t$ is a valid \pvalue. Recalling that fixed-sample $p$-values are dual to fixed-sample confidence intervals, it is also the case that anytime $p$-values are dual to anytime confidence intervals. These ideas are explored and explained in depth by Howard et al.~\cite{howard2020time1}. 
An alternative to anytime $p$-values, called safe $e$-values, was recently proposed by Gr\"{u}nwald et al.~\cite{grunwald2019safe}, and their relationship to confidence sequences, sequential tests and anytime $p$-values was detailed by Ramdas et al.~\cite{ramdas2020admissible}. Specifically, optionally stopped nonnegative supermartingales, which underlie all our bounds, yield safe $e$-values. The main takeaway message for our current paper is that all aforementioned tests can be reformulated as calculating anytime $p$-values or safe $e$-values. To exactly recover our level $\alpha$ tests, we just stop and reject at the first time that $\mathfrak{p}_t \leq \alpha$ (or equivalently, the $e$-value exceeds $1/\alpha$).

\section{Alternative masking functions} \label{sec:other_decompose}

In most of this paper, we have considered one way of decomposing \pvalue as equation~\eqref{default_decompse}, but interactive tests can be developed for other decompositions. Shafer et al.~\cite{shafer2011test} discuss a class of \textit{calibrators} (functions) for the \pvalues~$f:[0,1] \to [0,\infty)$ such that $f$ is non-increasing and $\int_0^1 f(p) dp \leq 1$. They consider a ``product-martingale'' $\prod_{i=1}^k f(p_i)$ and reject the null if
\[
\exists k \in \mathbb{N}: \prod_{i = 1}^k f(p_i) \geq \alpha^{-1},
\]
which uses Ville's inequality (an infinite-horizon uniform extension of Markov's inequality). For each calibrator~$f$, an interactive test can be developed by viewing $f(p)$ as the missing bit for inference and finding the corresponding masked \pvalue $g(p)$ for interactive ordering. Type-I error is controlled if the pair of $f(p)$ and $g(p)$ are \textit{mean independent} under the null: 
\begin{align}
    \mathbb{E}(f(p) \mid g(p)) = \mathbb{E}(f(p)).
\end{align}
Lei et al.~\cite{lei2017star} provide a recipe to construct mean independent $g(p)$ given any calibrator. The interactive test given a pair of $f(p)$ and $g(p)$ follows the same procedure as Algorithm~\ref{alg:imt}, with the rejection rule at each step $k$ changed to
\begin{align}
    \prod_{i = 1}^{M_k} f(p_i) \geq \alpha^{-1}.
\end{align}
or equivalently
\[
\sum_{i = 1}^{M_k} \log f(p_i) \geq \log(\alpha^{-1}).
\]

We explore a class of calibrators $f_c$ parameterized by a constant $c \in (0,1)$:
\begin{align} \label{eq:alter_masking}
    f_c(p) = c p^{c  - 1}.
\end{align}
In an interactive test, $\log f_c(p_i)$ is viewed as playing the role of the missing bit for inference (even though it is technically not one bit, we use the same terminology for simplicity). To calculate the corresponding masked \pvalue, we define function $H_c(x) = x^c - x$ for $x \in [0,p_*]$, where $p_*$ is the solution of $\log f_c(p) = 0$. The masked $p$-value is defined as
\[
g_c(p_i)=
\begin{cases}
p_i, \text{ if } p_i \leq p_*\\
s(p_i), \text{ otherwise},
\end{cases}
\]
where for any $p_i > p_*$, we define $s(p_i)$ as the unique solution of $H_c(x) = H_c(p_i)$ within the range $[0, p_*]$. Both $p_*$ and $s(p_i)$ can be obtained numerically by a simple binary search since $\log f_c(p)$ and $H_c(x)$ are monotonic. To compare different options of missing bits, Figure~\ref{fig:hp_tradeoff} shows the maps for original $h(p_i)$ (one bit) and the log term $\log(f_c(p_i))$, since they play similar roles in the interactive tests as forming cumulative sum statistics. 

Different choices of missing bit and the corresponding masked \pvalue reflect a tradeoff between the information of \pvalues allocated for inference and interactive ordering. Compared with one bit $h$ defined in equation~\eqref{default_decompse}, $f_c$ maps small \pvalues to large value (Figure~\ref{fig:h_varyc}), so that an evident non-null leads to a big increment in the test statistics and higher likelihood of being detected. In other words, $f_c$ takes more information from \pvalues than $h$ for inference. However, the corresponding masked \pvalue is less informative to suggest a good ordering. It's because a wider range of \pvalues that are bigger than 0.5 (from nulls) would have small masked \pvalue (Figure~\ref{fig:g_varyc}), which mixes with the actual small \pvalues and makes it harder to select possible non-nulls. As $c$ approaches zero, more information is allocated to inference and less for interactive ordering.

We also consider a mixture of $f_c$, denoted as $f_m$:
\begin{align}
    f_m(p) = \int_0^1 c p^{c  - 1} dc \equiv \frac{1 - p + p\log p}{p (\log p)^2}.
\end{align}
The corresponding masked \pvalue $g_m(p)$ can be calculated using the same formula as above except for a new definition of~$H_m(x)$ as $\frac{x - 1}{\log x} - x$. As shown in Figure~\ref{fig:hp_tradeoff}, the amount of information that $f_m$ takes for inference is between $f_{0.2}$ and $f_{0.4}$. 

\begin{figure}[h]
    \centering
    \begin{subfigure}[t]{0.4\textwidth}
        \centering
        \includegraphics[width=1\linewidth]{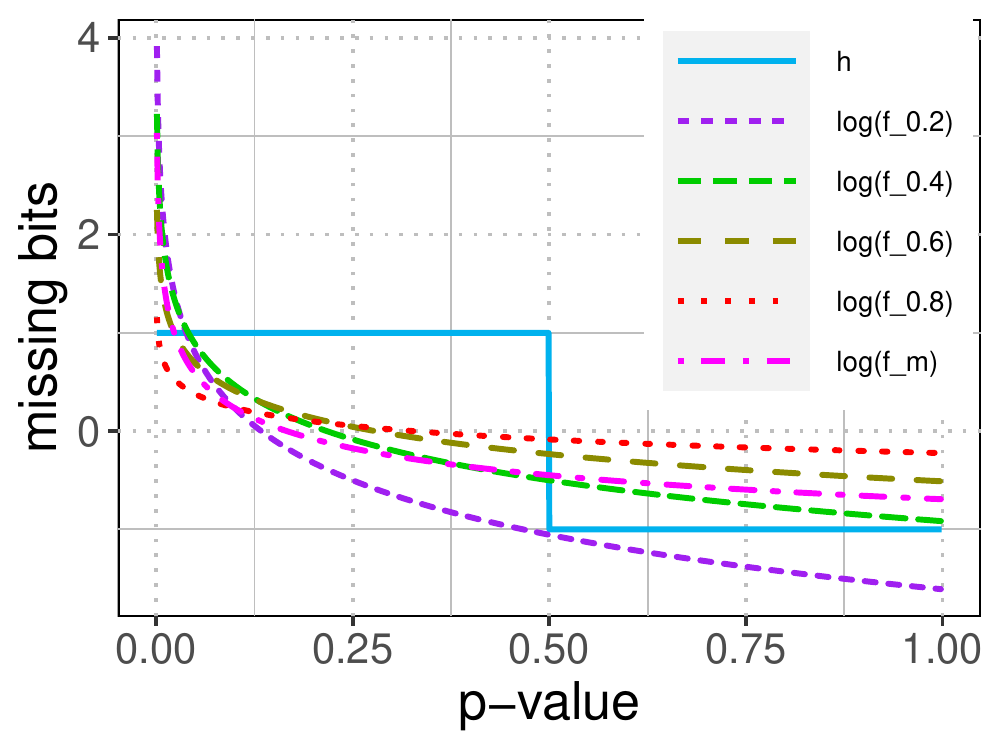}
        \caption{Different maps from \pvalue to the missing bit.}
        \label{fig:h_varyc}
    \end{subfigure}
    \hfill
    \begin{subfigure}[t]{0.4\textwidth}
        \centering
        \includegraphics[width=1\linewidth]{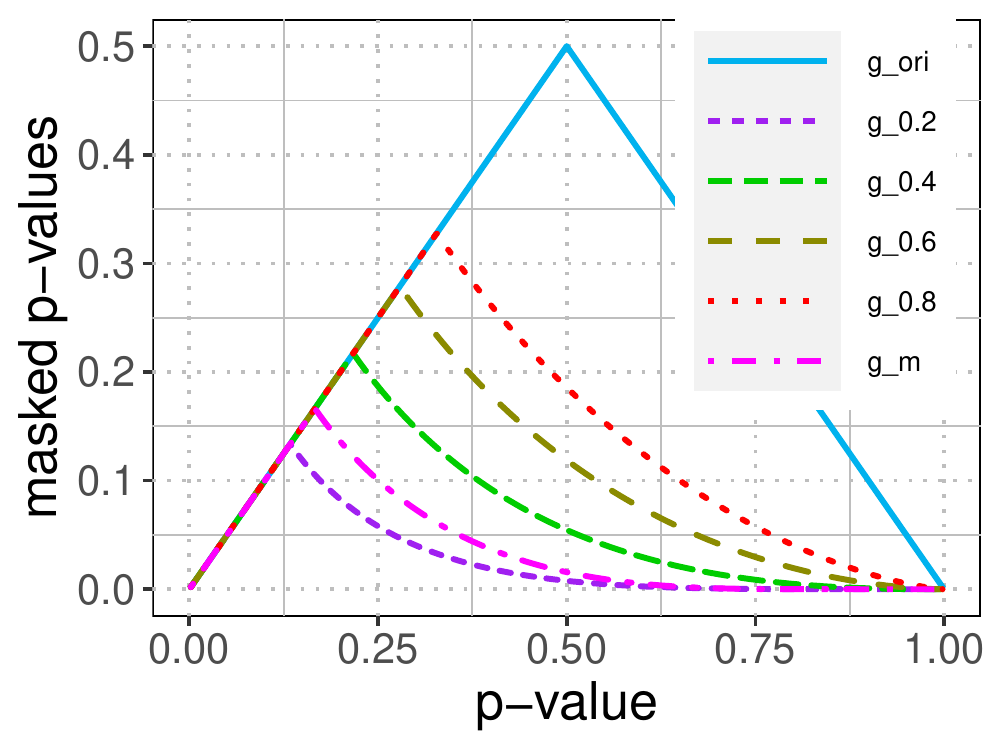}
        \caption{Corresponding maps from \pvalue to the masked \pvalue.}
        \label{fig:g_varyc}
  \end{subfigure}
  
  \caption{Different choices of missing bit and its corresponding masked \pvalue. When small \pvalues (possible non-nulls) are more evident when measured by one choice of the missing bit, they are less distinctive when looking at the corresponding masked \pvalues.
  }
  \label{fig:hp_tradeoff}
\end{figure}

We compare the interactively ordered martingale tests using different missing bits: (a) the original one bit $h(p_i)$ defined in equation~\eqref{default_decompse}; (b) $f_c(p_i)$ where we vary parameter $c$ as $(0.2, 0.4, 0.6, 0.8)$; and (c) the mixed missing bit $f_m(p_i)$. Our simulation uses the structured hypotheses with a cluster of non-nulls (described in Section~\ref{sec:sim_cluster}). The highest power comes from the test with the original definition of the missing bit: $h(p_i) = 2\cdot1\{p_i < 0.5\} - 1$ (Figure~\ref{fig:power_varyc}). 
\vspace{20pt}
\begin{figure}[h!]
        \centering
        \includegraphics[width=0.5\linewidth]{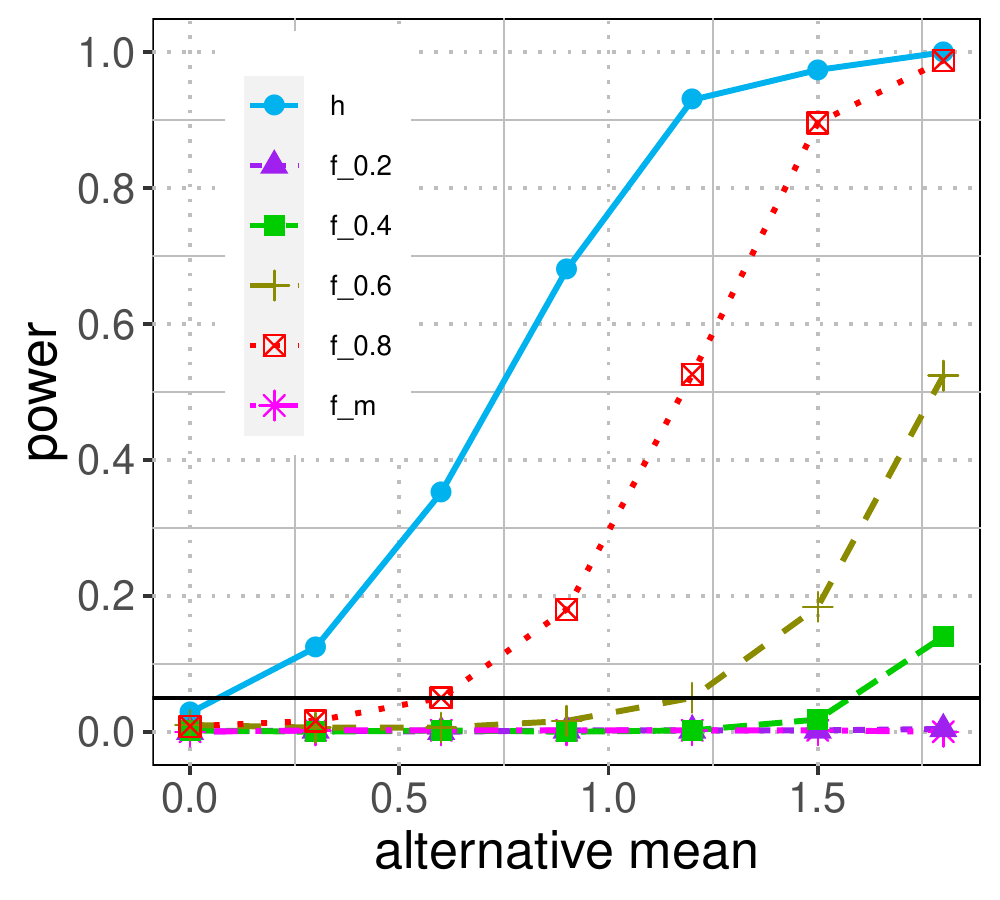}
        \caption{Power of interactive tests using different missing bits. Under the block structure of non-nulls as described in Section~\ref{sec:sim_cluster}, the IMT with the original missing bit defined in equation~(\ref{default_decompse}) has the highest power.}
        \label{fig:power_varyc}
\end{figure}

\vspace{20pt}
However, given that there is a tradeoff between the information contained in the missing bit and the masked \pvalue, and that the masked \pvalue is used together with the prior knowledge for a good ordering, we conjecture that the performance of tests with different missing bits depends on the amount of prior knowledge. When the prior knowledge is informative to order the hypotheses, the test with most of the information in the missing bit has a higher power (an example is the \mst, which has the highest power in Figure~\ref{fig:center}). 
We leave the following as an open question:  under different types of prior knowledge, does there exist and can one determine an ``optimal'' \pvalue decomposition that leads to the highest power?

\section{Summary} \label{sec:dis}

We have introduced martingale analogs of some classical global null tests, and used these to build adaptively ordered martingale tests through the idea of masking. These are further generalized to a protocol for interactively ordered martingale tests that possess the following interesting advantages:
\begin{itemize}
    \item It is a general global null testing framework that can utilize any types of covariates, structural constraints, prior knowledge and repeated user interaction guided by a posited working model, all while provably controlling the type-I error.
    \item It permits the use of Bayesian modeling techniques while retaining frequentist error guarantees.
    \item It applies to both the batch and online settings.
    \item It is robust against conservative nulls.
    \item It has favorable theoretical power guarantees in simple settings, and performs well in simulations.
\end{itemize}

In fact, in most of this paper, we do not need to know the null distribution of the underlying test statistics and be tied to working with $p$-values as inputs. Given test statistics $T_i \in \mathcal{R}_n$ for each hypothesis $H_i$, the framework of the \imt applies as long as there exits two functions $h: \mathcal{R}_n \to \{-1, 1\}$ and $g: \mathcal{R}_n \to \mathcal{R}$ such that 
\begin{align}
    \mathbb{E}\left[h(T_i) \mid g(T_i)\right] \leq 0 \quad \text{ for all } i \in \mathcal{I}.
\end{align}
As an example, if the distribution of the test statistic $T_i$ is symmetric under the null (such as Gaussian with unknown covariance, a t distribution with unknown degrees of freedom, or a centered Cauchy), we can still use $\text{sign}(T_i)$ and $|T_i|$ as $h(T_i)$ and $g(T_i)$ respectively. Indeed, type-I error control (Theorem~\ref{thm:imt_error}) still holds in this setting, since $h(T_i)$ and $g(T_i)$ for the aforementioned decompositions are independent under the null.

We believe that interactive testing protocols are only beginning to be explored in the literature, and constitute both an intellectually fascinating direction for further exploration, as well as a potentially powerful one. Masking (and progressive unmasking) is a promising technique that permits interaction, and it deserves further scrutiny and generalization to other settings.

\section*{Acknowledgements}
We thank the anonymous reviewers for their helpful suggestions. AR acknowledges support from NSF DMS 1916320, and NSF CAREER 1945266. SB acknowledges support from NSF DMS 1713003, and CCF 1763734. LW acknowledges support from NSF DMS 1713003. This work used the Extreme Science and Engineering Discovery Environment
(XSEDE) \citep{towns2014xsede}, which is supported by National Science Foundation grant
number ACI-1548562. Specifically, it used the Bridges system \cite{nystrom2015bridges}, which
is supported by NSF award number ACI-1445606, at the Pittsburgh Supercomputing Center
(PSC).

\bibliographystyle{acm}
\bibliography{rif}

\appendix
\section{Error control}
This section proves the type-I error control for our proposed methods: the \mst and the \imt. 

\subsection{Proof of Theorem~\ref{thm:mst_error}}
\label{sec:proof_mst_error}
\begin{proof}
Under the global null, because $p$-values are independent and stochastically larger than the uniform, the transformed \pvalues $\Phi^{-1}(1 - p_i)$ are independent and stochastically smaller than a standard Gaussian. Thus given the uniform bound for a Gaussian increment martingale $u_\alpha(k)$,
\begin{align*}
    &\mathbb{P}_0\left(\exists k \in \mathbb{N}:  \sum_{i=1}^k \Phi^{-1}(1 - p_i) \geq u_\alpha(k)\right){}\\
    \leq~& \mathbb{P}\left(\exists k \in \mathbb{N}:  \sum_{i=1}^k G_i \geq u_\alpha(k)\right){}\\
    \leq~& \alpha,
\end{align*}
where $G_i$ for $i \in \mathcal{I}$ are i.i.d. standard Gaussians. By definition the above argument proves the type-I error control.
\end{proof}

\subsection{Proof of Theorem~\ref{thm:imt_error}} \label{apd:imt_proof}

This proof also implies Theorem~\ref{thm:adp_error} since the \postmt is a special case of the \imt.

\begin{proof}
\textbf{Batch setting.} We argue that the sum $\{\sum_{i\in M_k}h(p_i)\}_{k\in\mathcal{I}}$ is a supermartingale with respect to the filtration $\{\mathcal{F}_{k-1}\}_{k\in\mathcal{I}}$. First, the sum $\sum_{i\in M_k}h(p_i)$ is measurable with respect to $\mathcal{F}_{k-1}$ because the random set $M_k = M_{k-1}\cup \{i_k^*\}$ has its distribution defined with respect to $\mathcal{F}_{k-1}$. 

Second, we prove that
\begin{align} \label{eq:mtg_cond}
    \mathbb{E}(\sum_{i\in M_k}h(p_i) \mid \mathcal{F}_{k-1}) \leq  \sum_{i\in M_{k - 1}}h(p_i),
\end{align}
Because $\mathbb{E}(\sum_{i\in M_{k}}h(p_i) \mid \mathcal{F}_{k-1}) = \sum_{i\in M_{k-1}}h(p_i) + \mathbb{E}(h(p_{i_k^*}) \mid \mathcal{F}_{k-1})$, condition~\eqref{eq:mtg_cond} boils down to proving
\[
\mathbb{E}(h(p_{i_k^*}) \mid \mathcal{F}_{k-1}) \leq 0.
\]

Since $i_k^*$ and $M_{k-1}$ are $\mathcal{F}_{k-1}$ measurable, and $i_k^* \notin M_{k-1}$, we see that 
\begin{align*}
    \mathbb{E}(h(p_{i_k^*}) \mid \mathcal{F}_{k-1}) \leq \max_{i \notin M_{k-1}} \mathbb{E}(h(p_i) \mid \mathcal{F}_{k-1}) = \max_{i \notin M_{k-1}} \mathbb{E}(h(p_{i}) \mid g(p_{i})),
\end{align*}
where the last equation is because the $p$-values are assumed to be independent of each other and of the covariates~$x_i$ under the global null; and thus, $h(p_i) \mid \mathcal{F}_{k-1}$ has the same distribution as $h(p_i) \mid g(p_i)$.

The proof is completed if
\begin{align} \label{eq:expect_cond}
   \mathbb{E}(h(p_{i}) \mid g(p_{i})) \leq 0,
\end{align}
for any $i \notin M_{k-1}$. In this case, the sum $\{\sum_{i\in M_k}h(p_i)\}_{k\in\mathcal{I}}$ is a martingale. Also, the increment is stochastically smaller than a Rademacher and following the same argument in Section~\ref{sec:proof_mst_error}, so the test using a bound for a Gaussian increment martingale controls the type-I error (because a Rademacher is subGaussian).

We have an intermediate result: the \imt has type-I error control for any $h(p)$ and $g(p)$ such that condition~\eqref{eq:expect_cond} holds. For a mirror-conservative \pvalue, the missing bit $h(p_i)$ conditioned on its corresponding masked \pvalue $g(p_i)$ is stochastically smaller than a fair coin flip:
\begin{align*}
    &\mathbb{P}_0(h(p_i) = -1 \mid g(p_i) = x) = \frac{f_i(1-x)}{f_i(1-x) + f_i(x)} {}\\
    \geq~& \frac{f_i(x)}{f_i(1-x) + f_i(x)}
    =  \mathbb{P}_0(h(p_i) = 1 \mid g(p_i) = x),
\end{align*}
for any $x \in [0,0.5]$ (i.e., the range of $g(p_i)$), which implies condition~\eqref{eq:expect_cond} and thus completes the proof.

\noindent \textbf{Online setting.}
Let the index of the hypothesis that enters the rejection set $M_{k-1}$ be $t_k^*$. Notice that $t_k^*$ is a stopping time with respect to $\mathcal{F}_{t-1}$ (that is, $\{t_k^* = t\}$ is measurable with respect to $\mathcal{F}_{t-1}$ because we decide whether to include $p_t$ based on $\mathcal{F}_{t-1}$). For a clear notation, define a filtration indexed by $k$ as 
\begin{align}
    \mathcal{G}_{k - 1} :=  \mathcal{F}_{t_k^* - 1},
\end{align}
denoting all the information available prior to the $k$-th entered hypothesis. We argue that the sum $\{\sum_{i\in M_k} h(p_i)\}_{k\in\mathcal{I}}$ is a supermartingale with respect to the filtration $\{\mathcal{G}_{k-1}\}_{k\in\mathcal{I}}$. The proof is similar to the above batch setting, where we prove that
\[
\mathbb{E}(h(p_{t_k^*}) \mid \mathcal{G}_{k-1})\leq 0.
\]
Since $t_k^*$ is a stopping time with respect to $\mathcal{F}_{t_k^* - 1}$, we see that 
\begin{align*}
    &\mathbb{E}(h(p_{t_k^*}) \mid \mathcal{G}_{k-1}) = \mathbb{E}(h(p_{t_k^*}) \mid \mathcal{F}_{t_k^* - 1}){}\\
    \leq~& \max_{t} \mathbb{E}(h(p_t) \mid \mathcal{F}_{t-1}) = \max_{t} \mathbb{E}(h(p_{t}) \mid g(p_{t})),
\end{align*}
where the last equation is because the $p$-values are assumed to be independent of each other and of the covariates~$x_i$ under the global null; and thus, $h(p_i) \mid \mathcal{F}_{k-1}$ has the same distribution as $h(p_i) \mid g(p_i)$.

The rest of the proof is the same as the batch setting where we show condition~\eqref{eq:expect_cond} holds: 
\[
\mathbb{E}(h(p_{t}) \mid g(p_{t})) \leq 0,
\]
for mirror-conservative $p$-values. Thus, the sum $\{\sum_{i\in M_k} h(p_i)\}_{k\in\mathcal{I}}$ is a supermartingale with respect to the filtration $\{\mathcal{G}_{k-1}\}_{k\in\mathcal{I}}$. Recall that the increment is stochastically smaller than a Rademacher. Following the same argument in Section~\ref{sec:proof_mst_error}, the \imt in the online setting using bound for a Gaussian increment martingale controls the type-I error. 
\end{proof}

\subsection{Error control of the \imt with railway masking function in Section~\ref{sec:robust}}
\label{apd:error_control_railway}

Let the masked $p$-values defined by the railway function in Section~\ref{sec:robust} be: 
\begin{align*}
    \widetilde g(p) := \min(p, (p+\tfrac12) \text{mod} 1)
\end{align*}
The corresponding \imt has a valid error control when the $p$-values have nondecreasing densities under the global null.
\begin{theorem}
If under $\mathcal{H_G}_0$, the \pvalues have nondecreasing densities and are independent of each other and of the covariates $x_i$, then the \imt using $\widetilde g(p)$ in place of $g(p)$ controls the type-I error at level~$\alpha$.
\end{theorem}

\begin{proof}
Recall that in Appendix~\ref{apd:imt_proof}, we have an intermediate result: the \imt has type-I error control for any $h(p)$ and $g(p)$ such that condition~\eqref{eq:expect_cond} holds. For a \pvalue with a nondecreasing density, the missing bit $h(p_i)$ conditioned on its corresponding masked \pvalue $\widetilde g(p_i)$ is stochastically smaller than a fair coin flip:
\begin{align*}
    &\mathbb{P}_0(h(p_i) = -1 \mid \widetilde g(p_i) = x) = \frac{f_i(x + 0.5)}{f_i(x + 0.5) + f_i(x)} {}\\
    \geq~& \frac{f_i(x)}{f_i(x + 0.5) + f_i(x)}
    =  \mathbb{P}_0(h(p_i) = 1 \mid \widetilde g(p_i) = x),
\end{align*}
for any $x \in [0,0.5]$ (i.e. the range of $\widetilde g(p_i)$), which implies condition~\eqref{eq:expect_cond} and thus completes the proof.
\end{proof}

\begin{remark}
The above proof implies that the error control holds as long as under the global null, the $p$-values satisfy:
\begin{align*}
    f_i(a) \leq f_i(a + 0.5) \text{ for all } 0 \leq a \leq 0.5, i \in \mathcal{I},
\end{align*}
where $f_i$ is the probability mass function of $p_i$ for discrete p-values or the density function otherwise. This condition can be viewed as a third definition of conservativeness in addition to condition~\eqref{cond:stoch_dom} and~\eqref{cond:mirror_consv} in the main paper. 
It is not a consequence of condition~\eqref{cond:stoch_dom} (take $f(a) = \one(a \leq 0.5) + 4(a - 0.5)\one(a > 0.5)$) or condition~\eqref{cond:mirror_consv} (take $f(a) = 4\min(a,1-a)$), and it does not imply condition~\eqref{cond:stoch_dom} and~\eqref{cond:mirror_consv} (take $f(a) = 4(0.5 - a)\one(a < 0.5) + 4(1 - a)\one(0.5 \leq a < 1) + 4\one(a =1)$).
For simplicity, we focus on the $p$-values with increasing densities in Section~\ref{sec:robust}, which are considered as conservative $p$-values in all three definitions. 

\end{remark}

\section{Power guarantees in the batch setting} \label{apd:batch_guarantee}

This section presents the proofs of power guarantees in the batch setting for (1)~the batch Stouffer test, (2)~the \mst and (3)~the \imt.

\subsection{Proof of Theorem~\ref{thm:power_batch}} \label{apd:power_batch_nonadaptive}

We divide the proof into two subsections for the batch Stouffer test and the \mst.

\subsubsection{The batch Stouffer test} \label{apd:proof_batch_power}

\begin{proof}
Define the $Z$-score for each hypothesis $H_i$ as $Z_i = \Phi^{-1}(1 - p_i)$. Under setting~\ref{set:simple} in the main paper of testing Gaussian mean, the $Z$-score is a Gaussian $Z_i \sim N(\mu_i,1)$, or written as $N(r_i\mu_i, 1)$ to separate the true nulls from the true non-nulls. Thus, the sum $S_n = \sum_{i=1}^n Z_i$ is also a Gaussian $S_n \sim N\left(\sum_{i=1}^n r_i \mu_i, n\right)$. The power of the batch Stouffer test is
\begin{align*}
    \mathbb{P}_1\left(\frac{S_n}{\sqrt{n}} \geq \Phi^{-1}(1 - \alpha)\right) =~& \mathbb{P}_1\left(\frac{S_n - \sum_{i=1}^n r_i \mu_i}{\sqrt{n}} \geq  \Phi^{-1}(1 - \alpha) - \frac{\sum_{i=1}^n r_i \mu_i}{\sqrt{n}}\right){}\\
    =~& 1 - \Phi\left( \Phi^{-1}(1 - \alpha) - \frac{\sum_{i=1}^n r_i \mu_i}{\sqrt{n}}\right).
\end{align*}
A power of at least $1-\beta$ is is equivalent to
\begin{align*}
    1 - \Phi\left( \Phi^{-1}(1 - \alpha) - \frac{\sum_{i=1}^n r_i \mu_i}{\sqrt{n}}\right) \geq 1 - \beta,
\end{align*}
which can be rewritten as
\begin{align*}
    \sum_{i=1}^n r_i \mu_i \geq (\Phi^{-1}(1 - \alpha) + \Phi^{-1}(1 - \beta))n^{1/2},
\end{align*}
which is the condition in Theorem~\ref{thm:power_batch}.
\end{proof}

\subsubsection{The \mst} \label{sec:mst_proof}

\begin{proof}
Following the same proof for $S_n \sim N(r_i\mu_i, 1)$ in Section~\ref{apd:proof_batch_power}, for any $k = 1,\ldots, n$, ${S_k \sim N\left(\sum_{i=1}^k r_i \mu_i,k\right)}$.
The power of the \mst is
\begin{align*}
    &\mathbb{P}_1\left(\exists k \in \{1,\ldots,n\}: S_k \geq u_\alpha(k)\right) {}\\
    =~& \mathbb{P}_1\left(\exists k \in \{1,\ldots,n\}: S_k - \sum_{i=1}^k r_i \mu_i \geq u_\alpha(k) - \sum_{i=1}^k r_i \mu_i\right),
\end{align*}
The power of \mst is at least $1-\beta$ if
\begin{align*}
    \exists k^* \in \{1,\ldots,n\}: u_\alpha(k^*) - \sum_{i=1}^{k^*} r_i \mu_i \leq -u_\beta(k^*) \quad \text{ (a sufficient condition)},
\end{align*}
since under such condition,
\begin{align*}
    &\mathbb{P}_1\left(\exists k \in \{1,\ldots,n\}: S_k - \sum_{i=1}^k r_i \mu_i \geq u_\alpha(k) - \sum_{i=1}^k r_i \mu_i\right){}\\
    \geq~& \mathbb{P}_1\left(S_{k^*} - \sum_{i=1}^{k^*} r_i \mu_i \geq u_\alpha(k^*) - \sum_{i=1}^{k^*} r_i \mu_i\right){}\\
    \geq~& \mathbb{P}_1\left(S_{k^*} - \sum_{i=1}^{k^*} r_i \mu_i \geq -u_\beta(k^*)\right){}\\
    \geq~& \mathbb{P}_1\left(\forall k \in \{1,\ldots,n\}: S_k - \sum_{i=1}^k r_i \mu_i \geq -u_\beta(k)\right) \geq 1 - \beta.
\end{align*}
The last step holds because Gaussian increment martingale is symmetric so that $-u_\beta(k)$ is a uniform lower bound.

The power of \mst is less than $1-\beta$ if
\begin{align*}
    \forall k \in \{1,\ldots,n\}: u_\alpha(k) - \sum_{i=1}^{k} r_i \mu_i \geq u_{1 - \beta}(k) \quad \text{ (a necessary condition)},
\end{align*}
since
\begin{align*}
    &\mathbb{P}_1\left(\exists k \in \{1,\ldots,n\}: S_k - \sum_{i=1}^k r_i \mu_i \geq u_\alpha(k) - \sum_{i=1}^k r_i \mu_i\right){}\\
    \leq~& \mathbb{P}_1\left(\exists k \in \{1,\ldots,n\}: S_k - \sum_{i=1}^k r_i \mu_i \geq u_{1 - \beta}(k) \right) \leq 1 - \beta.
\end{align*}
Thus, we find a sufficient condition and a necessary condition for the \mst to have $1-\beta$ power. The proof completes by plugging the curved bound in test~(\ref{test:mst_curve}) in the main paper into the conditions. If without further explanation, $u_\alpha(k)$ in rest of the proofs denotes the curved bound.
\end{proof}

\subsection{Proof of Theorem~\ref{thm:power_batch_imt}} \label{apd:power_imt}

The \postmt uses the missing bits $h(p_i)$ for testing, and under no prior knowledge, uses the masked \pvalues $g(p_i)$ to order the hypotheses. We divide the proof into three steps: (1) derive the power guarantee given a fixed order in Lemma~\ref{lm:power_oriOrder}; (2) quantify the effect of ordering by masked \pvalues in Lemma~\ref{lm:size_shrink}, and (3) derive the power guarantee for the \postmt (Theorem~\ref{thm:power_batch_imt}).

\paragraph{The power of \postmt given a fixed order}
\begin{lemma} \label{lm:power_oriOrder}
Given a fixed sequence of $\{M_k\}_{k=1}^n$ with the size $|M_k| = k$, the \postmt with type-I error control $\alpha$ has power at least $1-\beta$ if
\[
\exists k \in \{1,\ldots,n\}: \sum_{i\in M_{k}} \left(r_i(2S_i(1) - 1) + (1-r_i)(2S_i(0) - 1)\right)  \geq \left(C_k^{\alpha} + C_k^{\beta}\right)k^{\frac{1}{2}}.
\]
where $S_i(1) = \mathbb{P}(h(p_i) = 1 \mid r_i = 1, \{M_k\}_{k=1}^n)$ is a measurement of the ``signal strength" from the non-nulls and $S_i(0) = \mathbb{P}(h(p_i) = 1 \mid r_i = 0, \{M_k\}_{k=1}^n)$ is from the nulls. Meanwhile the power is less than $1 - \beta$ if
\begin{align*}
    &\forall k \in \{1,\ldots,n\}: {}\\
    &\sum_{i\in M_{k}} \left(r_i(2S_i(1) - 1) + (1-r_i)(2S_i(0) - 1)\right)  \leq \left(C_k^{\alpha} - C_k^{1-\beta}\right)k^{\frac{1}{2}}.
\end{align*}
\end{lemma}

\begin{proof}
Consider the re-scaled increment $(h(p_{i_k^*}) + 1)/2 \mid \mathcal{F}_k$, which follows a Bernoulli:
\begin{align*}
    \frac{h(p_{i_k^*}) + 1}{2} \sim r_i\mathrm{Ber}(S_{i_k^*}(1)) + (1 - r_i)\mathrm{Ber}(S_{i_k^*}(0)).
\end{align*}
So the cumulative sum $S_k$ is a martingale with sub-Gaussian increments after centering, with expected value $ \sum_{i\in M_k} \left (r_i(2S_i(1) - 1) + (1-r_i)(2S_i(0) - 1)\right)$. So the power of \postmt is
\begin{align*}
    & \mathbb{P}_1\left(\exists k \in \{1,\ldots,n\}: S_k \geq u_\alpha(k)\right) {}\\
    =~& \mathbb{P}_1\left(\exists k \in \{1,\ldots,n\}: S_k -\sum_{i\in M_k} \left[r_i(2S_i(1) - 1) + (1-r_i)(2S_i(0) - 1)\right]\right.\\
    &\left. \geq u_\alpha(k) - \sum_{i\in M_k} \left[r_i(2S_i(1) - 1) + (1-r_i)(2S_i(0) - 1)\right]\right).
\end{align*}
The proof can be completed by following similar steps in the proof for \mst (Section~\ref{sec:mst_proof}).
\end{proof}

\paragraph{The effect of ordering} 
Define the $Z$-score as $Z_i = \Phi^{-1}(1 - p_i)$ for each hypothesis $H_i$. Under setting~\ref{set:simple} in the main paper, $Z_i$ is a Gaussian with unit variance and mean value $\mu_i$. We consider the simple case where for all the non-nulls $\mu_i = \mu$. The \postmt orders the hypotheses increasingly by $g(p_i)$, which is equivalent to ordering decreasingly by $|Z_i|$. Following definition~(\ref{def:permutation}), the $Z$-scores for non-nulls have the same distribution as $Z(\mu)$, and $Z_{(j)}(\mu)$ is the $Z$-score of $j$-th non-null when they are ordered decreasingly by $|Z_i|$. We describe the effect of ordering by the size of the set $M_k$ right after the $j$-th non-null enters, denoted as $M(j)$. 

\begin{lemma}
\label{lm:size_shrink}
The size of $M(j)$ follows a Binomial distribution (up to a constant):
\begin{align*}
    |M(j)| \sim j + \mathrm{Bin}\left(N_0 , \mathbb{P}(|Z(0)| > |Z_{(j)}(\mu)|)\right).
\end{align*}
The size $|M(j)|$ is uniformly upper bounded:
\begin{align*}
    \mathbb{P}_1\left(\forall j \in 1,\ldots,N_1 : |M(j)| \leq j + t_{\beta/N_1}(N_0, q_j)\right) \geq 1 - \beta,
\end{align*}
where $t_{\beta/N_1}(N_0, q_j)$ is $\beta/N_1$-th upper quantile of $\mathrm{Bin}\left(N_0 , \mathbb{P}(|Z(0)| > |Z_{(j)}(\mu)|)\right)$. 
\end{lemma}

\begin{remark}
\label{rmk:bin}
Denote $P(\mu) = \mathbb{P}(|Z(0)| \geq |Z(\mu)|)$.
The quantile $t_{\beta/N_1}(N_0, q_j)$ is upper bounded by a ratio of $P(\mu)N_0$ (when $P(\mu)N_0 > 1$):
\begin{align*}
    t_{\beta/N_1}(N_0, q_j) \leq \frac{2 + 2\sqrt{2\log(N_1/\beta)}}{N_1\left[\frac{N_1 + 1 -j}{N_1} - P(\mu)\right]^2} \max\{P(\mu) N_0, 1\},
\end{align*}
for $j = 1, \ldots, \lfloor N_1(1 - P(\mu)) + 1\rfloor$.
\end{remark}

\begin{proof}
In $M(j)$, the number of non-nulls is known as $j$ and the number of nulls is random. The nulls in $M(j)$ should have a higher absolute $Z$-score than $|Z_{(j)}(\mu)|$. Note that the $Z$-scores of the nulls are i.i.d. standard Gaussians, so the probability of a null to be in front of the $j$-th non-null is $\mathbb{P}(|Z(0)| > |Z_{(j)}(\mu)|)$ for any nulls. Thus the number of nulls before the $j$-th non-null follows a binomial distribution:
\begin{align*}
    \sum_{i: r_i = 0} 1(|Z_i(0)| > |Z_{(j)}(\mu)|) \sim \mathrm{Bin}\left(N_0 , \mathbb{P}(|Z(0)| > |Z_{(j)}(\mu)|)\right).
\end{align*}
Thus, the size of $M(j)$ is distributed as
\begin{align*}
    |M(j)| \sim j + \mathrm{Bin}\left(N_0 , \mathbb{P}(|Z(0)| > |Z(\mu_{\pi_j})|)\right).
\end{align*}
By the Bonferroni correction, with high probability $|M(j)|$ is upper bounded by
\begin{align*}
    \mathbb{P}_1\left(\forall j \in 1,\ldots, N_1:  |M(j)| \leq j + t_{\beta/N_1}(N_0, q_j)\right) \geq 1-\beta,
\end{align*}
where $t_{\beta/N_1}(N_0, q_j)$ is $\beta/N_1$-th upper quantile of $\mathrm{Bin}\left(N_0 , \mathbb{P}(|Z(0)| > |Z_{(j)}(\mu)|)\right)$. 

We further characterize the Binomial quantile $t_{\beta/N_1}(N_0, q_j)$ (proof of Remark~\ref{rmk:bin}). The quantile is upper bounded (by Chernoff inequality):
\begin{align*}
    t_{\beta/N_1}(N_0, q_j) &\leq \mathbb{P}(|Z(0)| > |Z_{(j)}(\mu)|) N_0 + \sqrt{2\mathbb{P}(|Z(0)| > |Z_{(j)}(\mu)|) N_0 \log(\tfrac{N_1}{\beta})} \nonumber {} \\
    &\leq (1 + \sqrt{2\log(\tfrac{N_1}{\beta})}) \max\{\mathbb{P}(|Z(0)| > |Z_{(j)}(\mu)|) N_0, 1\}.
\end{align*}
The proof completes by showing that the probability term $\mathbb{P}(|Z(0)| > |Z_{(j)}(\mu)|)$ is upper bounded:
\begin{align} \label{eq:bound_quantile}
    \mathbb{P}(|Z(0)| > |Z_{(j)}(\mu)|) \leq \frac{2P(\mu)}{N_1\left[\frac{N_1 + 1 -j}{N_1} - P(\mu)\right]^2}.
\end{align}

The above bound~\eqref{eq:bound_quantile} holds because the event $|Z(0)| > |Z_{(j)}(\mu)|$ can be viewed as comparing the absolute value of $Z(0)$ with $N_1$ Gaussians $\{Z^i(\mu)\}_{i=1}^{N_1}$ with the same distribution as $Z(\mu)$, and $|Z(0)|$ is bigger than $N_1 - j + 1$ of them. The number of $Z^i(\mu)$ that $|Z(0)| > |Z^i(\mu)|$ follows a binomial distribution, with probability $\mathbb{P}\left(|Z(0)| > |Z(\mu)|\right) := P(\mu)$. Let $X$ be $\mathrm{Bin}(N_1, P(\mu))$ and bound~\eqref{eq:bound_quantile} holds because
\begin{align*}
    &\mathbb{P}(|Z(0)| > |Z_{(j)}(\mu)|) = \mathbb{P}(X > N_1 - j + 1){}\\
    \leq~& \exp\left\{-\frac{[N_1(1-P(\mu)) - j + 1]^2}{2N_1P(\mu)(1-P(\mu))}\right\} \leq \exp\left\{-\frac{N_1\left[\frac{N_1 + 1 -j}{N_1} - P(\mu)\right]^2}{2P(\mu)}\right\}{}\\
    \leq~& \frac{2P(\mu)}{N_1\left[\frac{N_1 + 1 -j}{N_1} - P(\mu)\right]^2},
\end{align*}
for $j = 1, \ldots, \lfloor N_1(1 - P(\mu)) + 1\rfloor$. The proof of Remark~\ref{rmk:bin} is completed by plugging bound~\eqref{eq:bound_quantile} in the upper bound for $t_{\beta/N_1}(N_0, q_j)$. 
\end{proof}

\paragraph{Proof of Theorem~\ref{thm:power_batch_imt}} \begin{proof}
Lemma~\ref{lm:power_oriOrder} provides a condition for \postmt to have at least $1-\beta$ power given any choice of $\{M_k\}_{k=1}^n$, thus when $\{M_k\}_{k=1}^n$ is random, the power is at least $1-\beta$ if
\begin{align} 
    &\exists k \in \{1,\ldots,n\}:\nonumber{}\\
    &\sum_{i\in M_k} \left(r_i(2S_i(1) - 1) + (1-r_i)(2S_i(0) - 1)\right) \geq \left(C_{|M_k|}^{\alpha} + C_{|M_k|}^{\beta}\right)(|M_k|)^{1/2}, \label{eq:rand_cond}
\end{align}
where $S_i(0)$ and $S_{i}(1)$ as the probabilities conditioning on $M_k$ are random. Whether the above condition holds is not determinant, and Theorem~\ref{thm:power_batch_imt} provides a sufficient condition such that the above condition holds with high probability.

First, for all the nulls, 
\begin{align*}
    S_i(0) &= \mathbb{P}(h(p_i) > 0| r_i = 0, \{M_k\}_{k=1}^n){}\\
    &\overset{(a)}{=} \mathbb{P}(Z_i > 0| r_i = 0, \{M_k\}_{k=1}^n){}\\
    &\overset{(b)}{=}  \mathbb{P}(Z_i > 0| r_i = 0) = 0.5,
\end{align*}
where $(a)$ is because by the definition of the $Z$-score, $h(p_i) > 0$ is equivalent to $Z_i > 0$; and $(b)$ is because $\{M_k\}_{k=1}^n$ is determined by $|Z_i|$ which is independent of $\one(Z_i > 0)$ when $r_i = 0$. Thus, $(2S_i(0) - 1)(1-r_i) = 0$ and in the above condition the sum on the left-hand side only increases when a non-null enters~$M_k$. Therefore, the above condition is satisfied if and only if it is satisfied when a non-null enters~$M_k$:
\begin{align*}
    \exists j \in \{1,\ldots,N_1\} : \sum_{i\in M(j)} r_i(2S_i(1) - 1) \geq \left(C_{|M(j)|}^{\alpha} + C_{|M(j)|}^{\beta}\right)(|M(j)|)^{1/2}.
\end{align*} 

Second, the non-nulls in $M(j)$ are the ones with $j$ highest absolute $Z$-scores, whose $Z$-scores are $Z_{(1)}(\mu),\ldots, Z_{(j)}(\mu)$. Thus, $\sum_{i\in M(j)} r_iS_i(1)$ can be  expressed as $\sum_{s=1}^j \mathbb{P}(Z_{(s)}(\mu) > 0)$, and the above condition can be rewritten as 
\begin{align*}
    \exists j \in \{1,\ldots,N_1\}: \sum_{s=1}^j \left(2\mathbb{P}(Z_{(s)}(\mu) > 0) - 1\right) \geq\left(C_{|M(j)|}^{\alpha} + C_{|M(j)|}^{\beta}\right)(|M(j)|)^{1/2}.
\end{align*}
The above condition holds with probability at least $1-\beta$ if
\begin{align} \label{eq:cond_fix}
    \exists j \in \{1,\ldots,N_1\}:  \sum_{s=1}^j \left(2\mathbb{P}(Z_{(s)}(\mu) > 0) - 1\right) \geq\left(C_n^{\alpha} + C_n^{\beta}\right)(j + t_{\beta/N_1}(N_0, q_j))^{\frac{1}{2}},
\end{align}
where $C_n^{\alpha} + C_n^{\beta} \geq C_{|M(j)|}^{\alpha} + C_{|M(j)|}^{\beta}$ and $j + t_{\beta/N_1}(N_0, q_j)$ is the uniform upper bound of $|M(j)|$ by Lemma~\ref{lm:size_shrink}.

Overall when condition~\eqref{eq:cond_fix} as above holds, the probability of failing to reject is less than the sum of (a) the probability that $|M(j)|$ exceeds its upper bound, which is less than $\beta$; and (b) the probability of not rejecting when condition~\eqref{eq:rand_cond} is satisfied, which is also less than $\beta$; thus the power is at least $1 - 2\beta$. The proof of theorem~\ref{thm:power_batch_imt} completes after replacing all $\beta$ in condition~\eqref{eq:cond_fix} with $\beta/2$.
\end{proof}

\subsection{Proof of condition~\eqref{eq:cond_imt_mu} in the main paper} \label{apd:cond_imt_mu}

\begin{proof}
Let $j = N_1/2$ in Theorem~\ref{thm:power_batch_imt}, the power of \postmt is at least $1-\beta$ if
\begin{align}
    \label{eq:singleJ}
    \sum_{s=1}^{N_1/2} \left(2\mathbb{P}(Z_{(s)}(\mu) > 0) - 1\right) \geq\left(C_n^{\alpha} + C_n^{\beta/2}\right)\left(N_1/2 + t_{\beta/(2N_1)}\left(N_0, q_{N_1/2}\right)\right)^{1/2}.
\end{align}
First, the left-hand side can be lower bounded by
\[
\sum_{s=1}^{N_1/2} \left(2\mathbb{P}(Z_{(s)}(\mu) > 0) - 1\right) \geq N_1/2\cdot (2\Phi(\mu) - 1) = N_1\Phi(\mu) - N_1/2,
\]
since the term $\frac{1}{j}\sum_{s=1}^{j} \left(2\mathbb{P}(Z_{(s)}(\mu) > 0) - 1\right)$ decreases in $j$ and is minimum at $j = N_1$, whose value is 
\begin{align*}
    \frac{1}{N_1}\sum_{s=1}^{N_1} \left(2\mathbb{P}(Z_{(s)}(\mu) > 0) - 1\right) &=\frac{1}{N_1}\sum_{s=1}^{N_1} \left(2\mathbb{E}(\one(Z_{(s)}(\mu) > 0)) - 1\right){}\\
    &=  \frac{1}{N_1} \left(2\mathbb{E}\left(\sum_{s=1}^{N_1}\one(Z_{(s)}(\mu) > 0)\right) - N_1\right) {}\\
    &=  \frac{1}{N_1} \left(2N_1\mathbb{E}\left(\one(Z(\mu) > 0)\right) - N_1\right) = 2\Phi(\mu) - 1.
\end{align*}
Second on the right-hand side, $t_{\beta/(2N_1)}\left(N_0, q_{N_1/2}\right)$ can be upper bounded (by Chernoff inequality):
\begin{align*}
    t_{\beta/(2N_1)}\left(N_0, q_{N_1/2}\right) \leq~& \mathbb{P}(|Z(0)| > |Z_{(N_1/2)}(\mu)|) N_0 {}\\
    &+ \sqrt{2\mathbb{P}(|Z(0)| > |Z_{(N_1/2)}(\mu)|) N_0\log(2N_1/\beta)},
\end{align*}
in which the probability term $\mathbb{P}(|Z(0)| > |Z_{(N_1/2)}(\mu)|)$ can be further upper bounded by
\[
\mathbb{P}(|Z(0)| > |Z(\mu_{\pi_{N_1/2}})|) \leq 2 - 2\Phi(\mu),
\]
since 
\begin{align*}
    \mathbb{P}(|Z(0)| > |Z(\mu_{\pi_{N_1/2}})|) 
    &\overset{(a)}{\leq}  \frac{2P(\mu)}{N_1\left(1-P(\mu) - \frac{N_1/2 - 1}{N_1}\right)^2}{}\\
    &\overset{(b)}{\leq} P(\mu) \overset{(c)}{\leq} 2 - 2\Phi(\mu),
\end{align*}
where $(a)$ is in the proof of Remark~\ref{rmk:bin} in Section~\ref{apd:power_imt}; $(b)$ holds because of the condition ${N_1 \geq 6 \left(C_n^{\alpha} + C_n^{\beta/2}\right)^2}$ and $\mu > 2$ (an assumption we visit later); and $(c)$ is because
$P(\mu) = \mathbb{P}(|Z(0)| \geq |Z(\mu)|) = {2\mathbb{P}(Z(0) \geq |Z(\mu)|)}$, which is less than $2\mathbb{P}(Z(0) \geq Z(\mu))$.

Plugging the lower bound of the left-hand side and the upper bound of the right-hand side, condition~\eqref{eq:singleJ} is implied by
\begin{align*}
    (\Phi(\mu) - \tfrac{1}{2})^2 \geq~& \left(C_n^{\alpha} + C_n^{\beta/2}\right)^2\frac{4 \max\{( 1 - \Phi(\mu)) N_0, \sqrt{( 1 - \Phi(\mu)) N_0\log(\tfrac{2N_1}{\beta})}\}}{N_1^2}{}\\
    &+ \left(C_n^{\alpha} + C_n^{\beta/2}\right)^2 \frac{N_1/2}{N_1^2}.
\end{align*}
Given $\mu > 2$ and $N_1 \geq 6 \left(C_n^{\alpha} + C_n^{\beta/2}\right)^2$, the above condition holds if

\begin{align*}
    &\frac{1}{( 1 - \Phi(\mu))}{}\\
\geq~& \left(C_n^{\alpha} + C_n^{\beta/2}\right)^2\left(\frac{28N_0}{N_1^2}\right) \max\left(1 , \left(C_n^{\alpha} + C_n^{\beta/2}\right)^2\left(\frac{28\log(\tfrac{2N_1}{\beta})}{N_1^2}\right)\right).
\end{align*}
Given $\mu > 2$ and ${N_1 \geq 6 \left(C_n^{\alpha} + C_n^{\beta/2}\right)^2}$, indicating $1 - \Phi(\mu) \leq  e^{-\mu^2/2}/2 $ and $\log(2N_1/\beta) < \frac{N_1}{5}$, we have a sufficient condition of the above condition:
\[
2 e^{\mu^2/2}  \geq \frac{28}{\sqrt{2\pi}}\left(C_n^{\alpha} + C_n^{\beta/2}\right)^2\left(\frac{N_0}{N_1^2}\right),
\]
which can be written as a condition on $\mu$:
\[
\mu  \geq \sqrt{2\log\left( \frac{N_0}{N_1^2}\right) + 4\log\left(C_n^{\alpha} + C_n^{\beta/2}\right) + 3.45}.
\]
Finally we complete the proof by noting that the above condition implies the assumption $\mu \geq 2$ when $N_0 > 0.1 N_1^2$.
\end{proof}

\begin{remark}
Condition~\eqref{eq:cond_imt_mu} in the main paper falls within the ``detectable region" derived in the work of Donoho and Jin~\cite{donoho2015special}: for any test for the problem of detecting sparse Guassian mean ($N_1 \leq n^{1/2}$), type-I error $\alpha$ and type-II error~$\beta$ would be big such that $\alpha + \beta \to 1$ when $n\to \infty$ unless
\begin{align} 
    \mu &\geq \sqrt{\log\left(\frac{n}{N_1^2}\right)},  \quad &\text{when } n^{1/4} \leq N_1 \leq n^{1/2}, \label{eq:detect_reg1}{}\\
    \mu &\geq \sqrt{2}(\sqrt{\log n } -\sqrt{\log N_1 }), \quad &\text{when } 1 < N_1 < n^{1/4}. \label{eq:detect_reg2}
\end{align}
\end{remark}

\begin{proof}
First note that condition~\eqref{eq:cond_imt_mu} in the main paper indicates
\[
\mu \geq \sqrt{2\log\left( \frac{n}{N_1^2}\right)},
\]
for any $N_1 \leq n^{1/2}$, since
\begin{align*}
    &\sqrt{2\log\left( \frac{N_0}{N_1^2}\right) + 4\log\left(C_n^{\alpha} + C_n^{\beta/2}\right) + 3.45} {}\\
    \geq~& \sqrt{2\log\left( \frac{N_0}{N_1^2}\right) + 4\log\left(C_1^{1} + C_1^{1}\right) + 3.45} = \sqrt{2\log\left( \frac{n}{N_1^2} - \frac{1}{N_1}\right) + 8.6} {}\\
    \geq~& \sqrt{2\log\left( \frac{n}{2N_1^2}\right) + 8.6} \geq \sqrt{2\log\left( \frac{n}{N_1^2}\right)},
\end{align*}
when $2\leq N_1 \leq n^{1/2}$ and it is obvious when $N_1 = 1$. So when ${n^{1/4} \leq N_1 \leq n^{1/2}}$, condition~\eqref{eq:cond_imt_mu} is a subset in the detectable region~\eqref{eq:detect_reg1}. 

When $1 < N_1 < n^{1/4}$, denote $N_1 = n^a$ where $0<a<1/4$. The detectable region~\eqref{eq:detect_reg2} can be written as
\[
\mu \geq (1 - \sqrt{a}) \sqrt{2\log n}, 
\]
which is implied by condition~\eqref{eq:cond_imt_mu}, since
\[
\sqrt{2\log\left( \frac{n}{N_1^2}\right)} = \sqrt{1-2a}\sqrt{ 2\log n} \geq (1 - \sqrt{a}) \sqrt{2\log n},
\]
when $a < 1/4$. Hence condition~\eqref{eq:cond_imt_mu} is a subset of the detectable region~\eqref{eq:detect_reg1} and~\eqref{eq:detect_reg2}.
\end{proof}

\section{Power guarantees in the online setting} \label{apd:power_online}

This section proves the power guarantees in the online setting for three methods: the \mst, the \postmt, and a benchmark, the online Bonferroni method.

\subsection{Proof of Theorem~\ref{thm:power_online}} 
\label{apd:power_online_Bonf}

The power guarantee for the \mst in the online setting follows the same steps as that in the batch setting (Section~\ref{sec:mst_proof}), except that the range of $k$ is changed from $\{1,\ldots,n\}$ to $\{1,2,\ldots\}$. We present the proof of the power guarantee for the online Bonferroni method as follows.

First, we derive an upper bound on the power of the online Bonferroni test. Recall the Z-score $Z_k = \Phi^{-1}(1 - p_k)$, which follows a Gaussian distribution $Z_k \sim N(r_k\mu_k, 1)$. The power of rejecting the $k$-th hypothesis at $\alpha_k$ is
\begin{align*}
    \mathbb{P}(p_k < \alpha_k) = \mathbb{P}(Z_k > \Phi^{-1}(1 - \alpha_k)) = 1 - \Phi[\Phi^{-1}(1 - \alpha_k) - r_k \mu_k],
\end{align*}
and the overall power of the online Bonferroni is upper bounded by a union of rejecting individual hypotheses:
\begin{align}
    \mathbb{P}(\exists k \in \mathbb{N}: p_k < \alpha_k) \leq \sum_{k=1}^\infty \mathbb{P}(p_k < \alpha_k) = \sum_{k=1}^\infty 1 - \Phi[\Phi^{-1}(1 - \alpha_k) - r_k \mu_k].
\end{align}

To upper bound the overall power, we claim the following upper bound on individual power of any hypothesis $k$, which is in the ratio of the individual significance level $\alpha_k$.
\begin{lemma} \label{lm:individual_power}
Given any constant $C \in (e^{1/4}, 1)$, if the alternative mean is upper bounded:
\begin{align} \label{eq:cond_online_bon_mu}
    r_k\mu_k \leq \frac{1}{4\Phi^{-1}(1 - \alpha_k)},
\end{align}
the power of rejecting individual hypothesis $k$ is upper bounded:
\begin{align*}
    1 - \Phi[\Phi^{-1}(1 - \alpha_k) - r_k \mu_k] \leq C \cdot \alpha_k,
\end{align*}
for large $k$ such that $\alpha_k < a(C)$, where the threshold $a(C)$ increases in $C$. For example, $a(2) > 0.3$.
\end{lemma}
\begin{proof}
Consider the ratio of individual power over $\alpha_k$: 
\begin{align*} 
    \frac{1 - \Phi\left[\Phi^{-1}(1 - \alpha_k) - \frac{1}{4\Phi^{-1}(1 - \alpha_k)}\right]}{\alpha_k},
\end{align*}
which converges to $e^{1/4}$ as $\alpha_k \to 0$ by L'Hospital's rule:
\begin{align*}
    &\lim_{\alpha_k \to 0} \frac{1 - \Phi\left[\Phi^{-1}(1 - \alpha_k) - \frac{1}{4\Phi^{-1}(1 - \alpha_k)}\right]}{\alpha_k}{}\\
    =~& \lim_{\alpha_k \to 0} \frac{ \phi\left[\Phi^{-1}(1 - \alpha_k) - \frac{1}{4\Phi^{-1}(1 - \alpha_k)}\right]}{\phi\left[\Phi^{-1}(1 - \alpha_k)\right]}\left(1 + \frac{1}{4\left(\Phi^{-1}(1 - \alpha_k)\right)^2}\right) = e^{1/4}.
\end{align*}
We observe through simulations that the threshold $a(C) \geq 0.3$ when $C \geq 2$. 
\end{proof}

In the following, we derive sufficient conditions for the power of the online Bonferroni to be less than $1 - \beta$ (i.e., the complement of necessary conditions to have at least $1 - \beta$ power), separately under the case of dense non-nulls and sparse non-nulls.
\begin{proof}[Proof of Theorem~\ref{thm:power_online}]
\textbf{Dense non-nulls.} First, consider the dense case where the number of non-nulls are infinite, $\sum_{k=1}^\infty r_k= \infty$. The power of the online Bonferroni is less than $1 - \beta$ when 
\[
\sum_{k=1}^\infty 1 - \Phi[\Phi^{-1}(1 - \alpha_k) - r_k \mu_k] \leq 1 - \beta,
\]
which holds if for each individual hypothesis $k$ with a positive error budget (i.e., $\alpha_k > 0$), the power of rejection is bounded
\begin{align} \label{eq:cond_online_bon_ind}
    1 - \Phi\left[\Phi^{-1}(1 - \alpha_k) - r_k \mu_k\right] \leq \frac{1-\beta}{\alpha}\alpha_k,
\end{align}
where the upper bound $\frac{1-\beta}{\alpha}\alpha_k$ is chosen to satisfy two conditions: (a)~the overall power is less than $1-\beta$: $\sum_{k=1}^\infty \frac{1-\beta}{\alpha}\alpha_k \leq 1-\beta$ and (b)~individual power bound is larger than the corresponding error control level, $\frac{1-\beta}{\alpha}\alpha_k > \alpha_k$, so that the above condition is not trivially satisfied in the case of a null: $r_k\mu_k = 0$. By Lemma~\ref{lm:individual_power}, the above bound on individual power holds when $r_k\mu_k$ satisfy condition~\eqref{eq:cond_online_bon_mu} and $\alpha_k < 0.3$ (Notice that here the constant in the lemma is $C = \frac{1 - \beta}{\alpha} \geq 4$, so threshold $a\left(C\right) > 0.3$).

To further characterize condition~\eqref{eq:cond_online_bon_mu} on $r_k\mu_k$, we consider a baseline sequence where $\alpha_k^* = (6/\pi^2)\alpha/k^2$, which sums to $\alpha$.
For an arbitrary sequence $\{\alpha_k\}_{k=1}^\infty$ that sums to~$\alpha$, apply the condition for the baseline sequence, $r_k\mu_k \leq \frac{1}{4\Phi^{-1}(1 - \alpha_k^*)}$, and the power for each hypothesis $k$ is still upper bounded. Particularly, this upper bound differs by whether $\alpha_k \leq \alpha_k^*$ or $\alpha_k > \alpha_k^*$:
\begin{align*}
    &1 - \Phi\left[\Phi^{-1}(1 - \alpha_k) - \frac{1}{2\Phi^{-1}(1 - \alpha_k^*)}\right]{}\\ \leq~& 1 - \Phi\left[\Phi^{-1}(1 - \alpha_k^*) - \frac{1}{2\Phi^{-1}(1 - \alpha_k^*)}\right] \leq C \alpha_k^*, \quad \text{ if } \alpha_k \leq \alpha_k^*;{}\\
    &1 - \Phi\left[\Phi^{-1}(1 - \alpha_k) - \frac{1}{2\Phi^{-1}(1 - \alpha_k^*)}\right]{}\\
    \leq~& 1 - \Phi\left[\Phi^{-1}(1 - \alpha_k) - \frac{1}{2\Phi^{-1}(1 - \alpha_k)}\right] \leq C \alpha_k, \quad \text{ if } \alpha_k > \alpha_k^*,
\end{align*}
for $k$ such that  $\max\{\alpha_k, \alpha_k^*\} \leq a(C)$, and hence,
\begin{align*}
    1 - \Phi\left[\Phi^{-1}(1 - \alpha_k) - \frac{1}{2\Phi^{-1}(1 - \alpha_k^*)}\right] \leq C \max\{\alpha_k^*,\alpha_k\} \leq C (\alpha_k^* + \alpha_k).
\end{align*}
Choose the constant $C = \frac{1 - \beta}{2 \alpha}$ (with $a(C) > 0.3$), and the overall power is upper bounded by $1 - \beta$:
\begin{align*}
    \sum_{k=1}^\infty 1 - \Phi\left[\Phi^{-1}(1 - \alpha_k) - \frac{1}{2\Phi^{-1}(1 - \alpha_k^*)}\right] \leq \frac{1 - \beta}{2 \alpha} (2 \alpha) = 1 - \beta,
\end{align*}
if (a)~the significance levels are small: $\max\{\alpha_k, \alpha_k^*\} \leq 0.3$ for all $k = 1,2, \ldots$, which holds since $\alpha \leq (1 - \beta)/4 \leq 0.25$; and (b)~the alternative mean $r_k \mu_k$ satisfies condition~\eqref{eq:cond_online_bon_mu} for the baseline sequence, which holds when
\[
r_k\mu_k \leq 0.25\left(\sqrt{2\log \left(\frac{k^2}{\alpha}\right)}\right)^{-1},
\]
where the bound decreases at the rate of $\left(\sqrt{\log k }\right)^{-1}$. \\

\noindent \textbf{Sparse non-nulls.} Suppose the sequence $\{\alpha_k\}_{k=1}^\infty$ is nonincreasing. A stronger necessary condition can be derived if the non-nulls are sparse in the sense that there exists an upper bound $M$ such that $\sum_{k = 1}^\infty r_k \leq M < \infty$. We separately discuss the set of nulls $\{k: r_k = 0\}$, and the set of small and large $\alpha_k$. Let $k^* = M^2/\alpha$, and define the sets of large and small $\alpha_k$ as $L(k^*) := \{k \leq k^*: r_k = 1\}$ and $S(k^*) := \{k > k^*: r_k = 1\}$. The power would be less than $1 - \beta$ if 
\begin{align}
    \sum_{r_k = 0} 1 - \Phi[\Phi^{-1}(1 - \alpha_k) - r_k \mu_k]& \leq \alpha, \text{ and } \label{eq:cond_null}{}\\
    \sum_{k \in L(k^*)} 1 - \Phi[\Phi^{-1}(1 - \alpha_k) - r_k \mu_k]& \leq 2\alpha, \text{ and } \label{eq:cond_large} {}\\
    \sum_{k \in S(k^*)} 1 - \Phi[\Phi^{-1}(1 - \alpha_k) - r_k \mu_k]& \leq 1 - \beta - 3\alpha. \label{eq:cond_small}
\end{align}
Power bound~\eqref{eq:cond_null} for the nulls ($r_k = 0$) holds because individual power equals~$\alpha_k$ and $\sum_{r_k = 0} \alpha_k \leq \alpha$. Power bound~\eqref{eq:cond_large} for large $\alpha_k$ holds if we bound the power of each individual hypothesis $k \in L(k^*)$:
\begin{align*}
    1 - \Phi[\Phi^{-1}(1 - \alpha_k) - r_k \mu_k] \leq 2\alpha_k,
\end{align*}
which can be rewritten as
\begin{align*}
    r_k \mu_k \leq \Phi^{-1}(1 - \alpha_k) - \Phi^{-1}(1 - 2\alpha_k).
\end{align*}
Note that the above bound on $r_k \mu_k$ decreases in $\alpha_k$ and that the set of $\alpha_k$ for $k \in L(k^*)$ is lower bounded because $L(k^*)$ has finite number of hypotheses. 
Thus, the above condition holds if for $k \in L(k^*)$, all $r_k\mu_k$ are smaller than the bound corresponding to the smallest significance level in $L(k^*)$, which is $\alpha_{k^*}$:
\begin{align*}
    r_k \mu_k \leq  \Phi^{-1}\left(1 - \alpha_{k^*}\right) - \Phi^{-1}\left(1 - 2\alpha_{k^*}\right),
\end{align*}
where $k^* = M^2/\alpha$. Notice that $\Phi^{-1}\left(1 - x\right)$ is a convex function and its derivative is $-\left(\phi(\Phi^{-1}\left(1 - x\right)\right)^{-1}$, so we have
\begin{align*}
    \Phi^{-1}\left(1 - \alpha_{k^*}\right) - \Phi^{-1}\left(1 - 2\alpha_{k^*}\right) \geq \left(\phi(\Phi^{-1}\left(1 - 2\alpha_{k^*}\right)\right)^{-1} \alpha_{k^*} \geq 0.4\sqrt{\alpha_{k^*}},
\end{align*}
and power bound~\eqref{eq:cond_large} for large $\alpha_k$ holds when $r_k\mu_k \leq 0.4\sqrt{\alpha_{k^*}}$.

For small $\alpha_k$, a sufficient condition for the power bound~\eqref{eq:cond_small} is 
\begin{align*}
    1 - \Phi[\Phi^{-1}(1 - \alpha_k) - r_k \mu_k] \leq \frac{1 - \beta - 3\alpha}{M},
\end{align*}
for all $k \in S(k^*)$ using the fact that the number of hypotheses in $S(k^*)$ is smaller than $M$. The above condition can be rewritten as 
\begin{align*}
    r_k \mu_k \leq \Phi^{-1}(1 - \alpha_k) - \Phi^{-1}\left(1 - \frac{1 - \beta - 3\alpha}{M}\right).
\end{align*}
To characterize the rate of the above bound, recall that the sequence $\{\alpha_k\}_{k=1}^\infty$ decreases and sums to~$\alpha$, so $\alpha_k \leq \alpha/k$ for any $k = 1,2, \ldots$. Thus, the above condition on $r_k\mu_k$ holds when 
\begin{align*}
    r_k \mu_k \leq \sqrt{\log \left(\frac{k}{4\alpha}\right)} - \sqrt{2\log\left(\frac{M}{2(1 - \beta - 3\alpha)}\right)},
\end{align*}
where the threshold increases at the rate of $\sqrt{\log k}$. We note that the above threshold is positive for $k \in S(k^*)$, since $k > k^*$ and $\frac{k}{4\alpha} > \frac{M^2}{4\alpha^2} \geq \frac{M^2}{4(1 - \beta - 3\alpha)^2}$, so that the condition on $r_k\mu_k$ is nontrivial. 


\end{proof}

We also demonstrate that the necessary condition for dense non-nulls is fairly tight when all the hypotheses are non-null.
\begin{lemma} \label{lm:Bonf_example}
 Suppose the sequence $\{\alpha_k\}_{k=1}^\infty$ decreases at a slow rate,
 \begin{align*}
     \alpha_1 = 0 \text{ and } \alpha_k = A/[k(\log k)^2] \text{ for } k > 1,
 \end{align*}
 with constant $A = \alpha/\left(\sum_{k=2}^\infty 1/[k(\log k)^2]\right)$ such that $\sum_{k=1}^\infty \alpha_k = \alpha$. The power of the online Bonferroni test is one if all hypotheses are non-null for $k > 1$ and the mean value decreases: $\mu_k = \left(\log k\right)^{-1/c}$ for any $c > 2$.
\end{lemma}
\begin{proof}
Let $Z_k = \Phi^{-1}(1 - p_k) \sim N(\mu_k,1)$ and $X_k = Z_k - \mu_k \sim N(0,1)$. The power of the online Bonferroni test is 
\begin{align} \label{eq:exact_power_online_bon}
    \mathbb{P}(\exists k \in \mathbb{N}: Z_k \geq \Phi^{-1}(1 - \alpha_k)) =~& \mathbb{P}(\exists k \in \mathbb{N}: X_k \geq \Phi^{-1}(1 - \alpha_k) - \mu_k) \nonumber{}\\
    =~& 1 - \prod_{k=1}^\infty \Phi\left[\Phi^{-1}\left(1 - \alpha_k\right) - \mu_k\right].
\end{align}
Intuitively, the power would not converge to one when $\Phi\left[\Phi^{-1}\left(1 - \alpha_k\right) - \mu_k\right] \gtrapprox (1 - \alpha_k)$ (the case with $\mu_k = 0$) since $1 - \prod_{k=1}^\infty (1 - \alpha_k) \leq \sum_{k=1}^\infty \alpha_k \leq \alpha$, but could be one when $\Phi\left[\Phi^{-1}\left(1 - \alpha_k\right) - \mu_k\right] \ll 1 - \alpha_k$. To quantify this comparison, we consider the following ratio:
\[
b_k := \frac{1 - \Phi\left[\Phi^{-1}\left(1 - \alpha_k\right) - \mu_k\right]}{\alpha_k},
\]
and the power could be one when $b_k$ is large. Indeed, we claim that $b_k$ increases at a rate faster than $\log k$, or equivalently, $(\log k)/b_k \to 0$. 
It can be verified by L'Hospital's rule:
\begin{align*}
    \lim_{k \to \infty} (\log k)/b_k =~& \lim_{k \to \infty} \frac{\alpha_k \log k}{1 - \Phi\left[\Phi^{-1}\left(1 - \alpha_k\right) - \mu_k\right]}{}\\ 
    =~& \lim_{k \to \infty} \frac{\phi \left[\Phi^{-1}\left(1 - \alpha_k\right)\right]}{\phi\left[\Phi^{-1} \left(1 - \alpha_k\right) - \mu_k\right]}
    \frac{\log k + \frac{\alpha_k}{k} \big/ \frac{\partial \alpha_k}{\partial k}}{1 + \phi \left[\Phi^{-1}\left(1 - \alpha_k\right)\right] \frac{\partial \mu_k}{\partial k} \big/ \frac{\partial \alpha_k}{\partial k} },{}\\
\end{align*}
where for large $k$, we have $\Phi^{-1}\left(1 - \alpha_k\right) \geq \sqrt{\log k}$ and 
\begin{align*}
    \frac{\phi \left[\Phi^{-1}\left(1 - \alpha_k\right)\right]}{\phi\left[\Phi^{-1} \left(1 - \alpha_k\right) - \mu_k\right]} 
    \leq~& 2\exp\{-(\log k)^{1/2-1/c}\};
    {}\\
    \log k + \frac{\alpha_k}{k} \big/ \frac{\partial \alpha_k}{\partial k} \leq~& 2\log k;{}\\
    1 + \phi \left[\Phi^{-1}\left(1 - \alpha_k\right)\right] \frac{\partial \mu_k}{\partial k} \bigg/ \frac{\partial \alpha_k}{\partial k}  \geq~& 1.
\end{align*}
Thus, $\lim_{k \to \infty} (\log k)/b_k \leq \lim_{k \to \infty} \frac{4\log k }{\exp\{(\log k)^{1/2-1/c}\}} = 0$ for any $c > 2$. In other words, we have proved that $b_k/\log k \to \infty$.

The power~\eqref{eq:exact_power_online_bon} is one if $\prod_{k=1}^\infty \Phi\left[\Phi^{-1}\left(1 - \alpha_k\right) - \mu_k\right] = 0$, or equivalently,
\begin{align} \label{eq:cond_log_online_bon}
    \sum_{k=1}^\infty \log \Phi\left[\Phi^{-1}\left(1 - \alpha_k\right) - \mu_k\right] = -\infty,
\end{align}
where for large $k$, we have
\begin{align*}
    &\log \Phi\left[\Phi^{-1}\left(1 - \alpha_k\right) - \mu_k\right]{}\\
    =~& \log(1 - b_k \alpha_k) \leq - b_k\alpha_k{}\\
    \leq~& - A\log k /[k(\log k)^2] = - A/(k \log k).
\end{align*}
Condition~\eqref{eq:cond_log_online_bon} holds because $\sum_{k=1}^\infty - A/(k \log k) = -\infty$; and thus, we prove that the power of the online Bonferroni test is one.

\end{proof}

\subsection{Proof of Theorem~\ref{thm:power_online_imt}}
Theorem~\ref{thm:power_online_imt} is a simplified version of the following Theorem~\ref{thm:power_online_imtfull} (by Claim~\ref{clm:trans7-8}). Before stating Theorem~\ref{thm:power_online_imtfull}, we first define the distinction measure $D(c)$
as
\[
D(c) = \frac{\mathbb{P}( |Z(\mu)| > c)}{\mathbb{P}( |Z(0)| > c)},
\]
where $c$ is the screening parameter in the online \postmt. Bigger $D(c)$ indicates bigger distinction. Further denote $N_1(k) = \sum_{i=1}^k r_i$ as the number of non-nulls after $k$ hypotheses arrive and $N_0(k)= \sum_{i=1}^k 1 - r_i$ as for the nulls.

\begin{theorem}
\label{thm:power_online_imtfull}
The \postmt with type-I error $\alpha$ and threshold $c$ guarantees $1-\beta$ power if
\begin{align*}
   \exists k \in \mathbb{N}: 
   &(2S(\mu, c) - 1) \left( N_1(k) - \frac{C_k^{\beta/3} \sqrt{N_1(k)}}{2\mathbb{P}( |Z(\mu)| > c)}\right){}\\
   & \geq \frac{C_k^{\alpha} + C_k^{\beta/3}}{\mathbb{P}^{1/2}( |Z(\mu)| > c)}\left[ N_1(k) + D^{-1}(c) N_0(k) + \frac{C_k^{\beta/3} k^{1/2}}{2\mathbb{P}( |Z(\mu)| > c)} \right]^{1/2},
\end{align*}
where $S(\mu; c) = \mathbb{P}(Z(\mu) > 0 \mid |Z(\mu)| > c)$.
\end{theorem}

\begin{proof}
Denote $M_k$ as the set of hypotheses that pass screening ($|Z_i| > c$) after $k$ hypotheses arrive. By extending Lemma~\ref{lm:power_oriOrder} from $k = 1,\ldots,n$ to $k = 1,2,\ldots$, the power of \postmt is at least $1-\beta$ if
\begin{align} \label{eq:rand_cond_online}
    \exists k \in \mathbb{N}:~&  \sum_{i\in M_k} \left(r_i(2S_i(1) - 1) + (1-r_i)(2S_i(0) - 1)\right) \nonumber{}\\
    &\geq \left(C_{|M_k|}^{\alpha} + C_{|M_k|}^{\beta}\right)(|M_k|)^{1/2},
\end{align}
where for the passed non-nulls, $S_i(1) = \mathbb{P}(h(p_i) = 1 \mid r_i = 1,  i \in M_i)$, which can be written in terms of $Z_i$ as $\mathbb{P}(Z_i > 0\mid r_i = 1, |Z_i| > c) = S(\mu, c)$, and for passed the nulls, $S_i(0) = \mathbb{P}(Z_i > 0 \mid r_i = 0, |Z_i| > c) = \mathbb{P}(Z(0) > 0 \mid |Z(0)| > c) = 0.5$. By the lemmas presented below, the right-hand side is upper bounded by
\begin{align*}
    |M_k| \leq& \mathbb{P}( |Z(\mu)| > c) \left(N_1(k) + D^{-1}(c)N_0(k)\right) + \frac{C_k^{\beta} }{2}k^{1/2},
\end{align*}
with probability $1 - \beta$ (Lemma~\ref{lm:online_shrink}). The left-hand side is lower bounded by
\begin{align*}
    \sum_{i \in M_k} (2S_i(1) - 1)r_i 
    =~& (2S(\mu, c) - 1) \sum_{i \in M_k} r_i {}\\
    \geq~& (2S(\mu, c) - 1) \left(\mathbb{P}(|Z(\mu)| > c) N_1(k) - \frac{C_k^{\beta} }{2}\sqrt{N_1(k)}\right),
\end{align*}
with probability $1 - \beta$ (Lemma~\ref{lm:shrink_nn}). The condition in Theorem~\ref{thm:power_online_imtfull} results from plugging the bounds of both sides into condition~\eqref{eq:rand_cond_online}. 

Overall, when the condition in Theorem~\ref{thm:power_online_imtfull} holds, the probability of failing to reject is less than the sum of (a) the probability that the upper bound for the right-hand side is violated, which is less than $\beta/3$; (b) the probability that the lower bound for the left-hand side is violated, which is less than $\beta/3$; and (c) the probability of not rejecting when condition~\eqref{eq:rand_cond_online} is satisfied, which is less than~$\beta/3$; thus the power is at least $1-\beta$.
\end{proof}

\begin{lemma} \label{lm:online_shrink}
The size of $M_k$ in the online setting is uniformly upper bounded:
\begin{align*}
    \mathbb{P}_1\left(\forall k \in \mathbb{N}: |M_k| - \mathbb{E}(|M_k|) \leq \frac{C_k^{\beta}}{2} k^{1/2} \right) \geq 1 - \beta,
\end{align*}
where
\begin{align*}
    \mathbb{E}(|M_k|) = \mathbb{P}( |Z(\mu)| > c) \left(N_1(k) + D^{-1}(c)N_0(k)\right).
\end{align*}
\end{lemma}

\begin{proof}
The probability of a hypothesis $H_i$ passing screening is $\mathbb{P}( |Z(\mu)| > c)$ when $H_i$ is a non-null, and $\mathbb{P}( |Z(0)| > c)$ when $H_i$ is a null. Denote $X_i$ as the indicator of whether $H_i$ passes the screening, then $|M_k| = \sum_{i=1}^k X_i$. Because $X_i$ are independent and each $X_i$ is a mixture of two Bernoullis (of value $\{0,1\}$), the size $|M_k|$ is a martingale with $\frac{1}{4}$-subGaussian increment. Therefore, 
\begin{align*}
   \mathbb{P}_1\left(\forall k \in \mathbb{N}: |M_k| - \mathbb{E}(|M_k|) \leq \frac{u_\beta(k)}{2} \right) \geq 1 - \beta,
\end{align*}
where $u_\beta(k)$ is the upper bound for Gaussian increment martingale as test~\eqref{test:mst_curve} in the main paper. The expected value is 
\begin{align*}
    \mathbb{E}(|M_k|) =~& \sum_{i=1}^k r_i \mathbb{P}( Z(\mu)| > c) + (1-r_i) \mathbb{P}( |Z(0)| > c){}\\
    =~& \mathbb{P}( |Z(\mu)| > c) \left(N_1(k) + D^{-1}(c)N_0(k)\right),
\end{align*}
which completes the proof.
\end{proof}

\begin{lemma}
\label{lm:shrink_nn}
The number of non-nulls in $M_k$ is uniformly lower bounded:
\begin{align*}
    \mathbb{P}_1\left(\forall k \in \mathbb{N}, \quad \sum_{i \in M_k} r_i - \mathbb{E}(\sum_{i \in M_k} r_i) \geq - \frac{C_k^{\beta}}{2}(N_1(k))^{1/2} \right) \geq 1 - \beta,
\end{align*}
where
\begin{align*}
    \mathbb{E}(\sum_{i \in M_k} r_i) = \mathbb{P}( |Z(\mu)| > c) N_1(k).
\end{align*}
\end{lemma}
The proof follows the same steps as in Lemma~\ref{lm:online_shrink}, by considering only the non-nulls, or equivalently assuming $r_i = 1$ for all $i$.  

\begin{claim}
\label{clm:trans7-8}
The condition of \postmt to have $1-\beta$ power in Theorem~\ref{thm:power_online_imt} implies that in Theorem~\ref{thm:power_online_imtfull}.
\end{claim}

\begin{proof}
First, the condition in Theorem~\ref{thm:power_online_imtfull} can be written as a quadratic inequality on $N_1(k)$,
\begin{align*}
    \exists k \in \mathbb{N}:& 
    (2S(\mu, c) - 1)^2[0.9 N_1(k)]^2 {}\\
    \geq~& \frac{\left(C_k^{\alpha} + C_k^{\beta/3}\right)^2}{\mathbb{P}( |Z(\mu)| > c)}\left((1 - D^{-1}(c)) N_1(k) + D^{-1}(c)k + \frac{C_k^{\beta/3} k^{1/2}}{2 \mathbb{P}( |Z(\mu)| > c)}\right),
\end{align*}
by noting that $ N_1(k) - \frac{C_k^{\beta/3} \sqrt{N_1(k)}}{2\mathbb{P}( |Z(\mu)| > c)} \geq 0.9 N_1(k)$ since the condition in Theorem~\ref{thm:power_online_imt} guarantees ${N_1(k) \geq \left(\frac{C_k^{\beta/3}}{0.2 \mathbb{P}( |Z(\mu)| > c)}\right)^2}$ (a claim we visit later). 

Solve the quadratic inequality for $N_1(k)$ to get a sufficient condition of the above one:
\begin{align*}
    2N_1(k) \geq~& \frac{\left(C_k^{\alpha} + C_k^{\beta/3}\right)^2}{\widetilde{S}(\mu, c)}(1 - D^{-1}(c)) {}\\
    &+ \left\{\frac{\left(C_k^{\alpha} + C_k^{\beta/3}\right)^4}{\widetilde{S}^2(\mu, c)}\left(1 - D^{-1}(c)\right)^2 + 4 \frac{\left(C_k^{\alpha} + C_k^{\beta/3}\right)^2}{\widetilde{S}(\mu, c)} D^{-1}(c) k \right.{}\\
    &+ \left. \frac{\left(C_k^{\alpha} + C_k^{\beta/3}\right)^2}{\widetilde{S}(\mu, c)}\frac{C_k^{\beta/3}}{2 \mathbb{P}( |Z(\mu)| > c)}k^{1/2}\right\}^{1/2},
\end{align*}
where $\widetilde{S}(\mu, c) = [0.9(2S(\mu, c) - 1)]^2\mathbb{P}( |Z(\mu)| > c)$ and $D^{-1}(c) = \frac{2\Phi(-c)}{\Phi(\mu-c) + \Phi(-\mu-c)}$. Note that under the square root, the last two terms involving $k$ is upper bounded by
\begin{align*}
    &4 \frac{\left(C_k^{\alpha} + C_k^{\beta/3}\right)^2}{\widetilde{S}(\mu, c)} D^{-1}(c) k + \frac{\left(C_k^{\alpha} + C_k^{\beta/3}\right)^2}{\widetilde{S}(\mu, c)}\frac{ C_k^{\beta/3}}{2 \mathbb{P}( |Z(\mu)| > c)}k^{1/2} {}\\
    =~& \frac{\left(C_k^{\alpha} + C_k^{\beta/3}\right)^2}{\widetilde{S}(\mu, c)(\Phi(\mu-c) + \Phi(-\mu-c))} \left(8\Phi(-c) k + \frac{C_k^{\beta/3}}{2}k^{1/2}\right) {}\\
    \leq~& \frac{\left(C_k^{\alpha} + C_k^{\beta/3}\right)^2}{\widetilde{S}(\mu, c)(\Phi(\mu-c) + \Phi(-\mu-c))} 9\Phi(-c) k = \frac{9\left(C_k^{\alpha} + C_k^{\beta/3}\right)^2D^{-1}(c)}{2\widetilde{S}(\mu, c)}  k,
\end{align*}
when $k \geq \left(\frac{C_k^{\beta/3}}{2\Phi(-c)}\right)^2$. By the fact that $\sqrt{a + b} \leq \sqrt{a} + \sqrt{b}$ for $a,b > 0$, an upper bound on the right-hand side is
\[
2\frac{1 - D^{-1}(c)}{\widetilde{S}(\mu, c)}\left(C_k^{\alpha} + C_k^{\beta/3}\right)^2 + 3(C_k^{\alpha} + C_k^{\beta/3})\frac{\sqrt{D^{-1}(c)/2}}{\widetilde{S}^{1/2}(\mu, c)}k^{1/2}.
\]
Thus, the above condition on $N_1(k)$ is implied by
\begin{align*}
    \exists k \geq \left(\frac{C_k^{\beta/3}}{2\Phi(-c)}\right)^2: N_1(k) \geq \widetilde B(\mu; c)\left(C_k^{\alpha} + C_k^{\beta/3}\right)^2 + A(\mu;c)(C_k^{\alpha} + C_k^{\beta/3})k^{1/2},
\end{align*}
where $A(\mu;c) = 3/2\frac{\sqrt{D^{-1}(c)/2}}{\widetilde{S}^{1/2}(\mu, c)}$ and $\widetilde B(\mu; c) = \frac{1 - D^{-1}(c)}{\widetilde{S}(\mu, c)}$.

Finally we review the assumptions made throughout the proof: (a) we assume $N_1(k) \geq \left(\frac{C_k^{\beta/3}}{0.2 \mathbb{P}( |Z(\mu)| > c)}\right)^2$, which is implied if $\widetilde B(\mu, c)$ is adjusted to $B(\mu, c)$ as defined in Theorem~\ref{thm:power_online_imt}; and (b) we assume $k \geq \left(\frac{C_k^{\beta/3}}{2\Phi(-c)}\right)^2$, which holds when ${k \geq T(\beta; c)}$; adjusting for these assumptions results in the condition in Theorem~\ref{thm:power_online_imt}.
\end{proof}

\section{Choices for the uniform bounds in the \mst} \label{apd:stou_bounds}

The martingale Stouffer test has the general form:
\begin{align*}
    \exists k \in \mathbb{N}: \sum_{i=1}^k \Phi^{-1}(1-p_i) \geq u_\alpha(k),
\end{align*}
where $u_\alpha(k)$ is the uniform bound for a martingale with standard Gaussian increment. We present four bounds from the work of Howard et al.~\cite{howard2020time, howard2020time1},
\begin{enumerate}
    \item a linear bound
    \begin{align}
    \label{eq:bound_lin}
        u_\alpha(k) = \sqrt{\frac{-\log\alpha}{2m}}k + \sqrt{\frac{-m\log\alpha}{2} },
    \end{align}
    where $m \in \mathbb{R}_+$ is a tuning parameter that determines the time at which the bound is tightest: a larger $m$ results in a lower slope but a larger offset, making the bound loose early on. 
    
    \item a curved bound from polynomial stitching method
    \begin{align}
        u_\alpha(k) = 1.7\sqrt{k\left(\log\log(2 k) + 0.72 \log \frac{5.2}{\alpha}\right)}.
    \end{align}
    
    \item a curved bound from discrete mixture method
    \begin{align}
        u_\alpha(k) = \inf\left\{ s \in \mathcal{R}: \sum_{i = 0}^\infty \omega_i\exp\{\lambda_is - \psi(\lambda_i)k\}\geq 1/\alpha \right\},
    \end{align}
    where $\lambda_i = 1.1^{-(i + 1/2)}\lambda_{\max}$ and $ \omega_i = 1.1^{-(i + 1)} \lambda_{\max} f(1.05\lambda_i)/10$, in which $\lambda_{\max} = \sqrt{2\log \alpha^{-1}}$ and ${f(x) = 0.4\frac{\textbf{1}_{0\leq x \leq \lambda_{\max}}}{x \log^{1.4}(e\lambda_{\max}/x)}}$.
    
    \item a curved bound from inverted stitching method (for finite time)
    \begin{align}
    \label{eq:bound_finite}
        u_\alpha(k) = 2.42\sqrt{k \log\log (ek) + 4.7}, \quad k =  1, 2, \ldots, 10^4,
    \end{align}
   where the time limit $10^4$ is chosen as the number of hypotheses in the following simulation.
\end{enumerate}
We use simulations to explore two choices in the \mst: (1)~the choice of parameter~$m$ in the linear bound~\eqref{eq:bound_lin}; and (2)~the choice among the above four types of bound.

\paragraph{Choice of the parameter $m$ in the linear bound}
A good choice of parameter $m$ should make the bound tight at where most non-nulls appear; thus, it depends on how the non-nulls distribute. A smaller $m$ results in a faster slope but a tighter bound at front, so it is desired when the non-nulls are gathered at front; and vice versa. 

We seek for a robust value of $m$ such that the resulting test has relatively high power under different non-null sparsity. The following constructed simulation is used for exploring bounds in both the \mst and the martingale Fisher test (introduced in Appendix~\ref{apd:fisher}). 

\begin{setting}
\label{set:sim_bound}
Consider the hypothesis of testing if a Gaussian has zero mean as in Setting~\ref{set:simple} in the main paper. In total $n = 10^4$ samples are simulated, with $100$ from the non-null distribution $N(1.5, 1)$ and the rest from the null $N(0,1)$. The non-null sparsity varies by restricting the range where the non-nulls randomly distribute. The non-null range is set as $H_1$ to $H_l$ and we test values $l = 100, 10^3, 2\times10^3, \ldots, 10^4$. We define the non-null sparisty as $\frac{l}{n}$ and a bigger value indicates a more sparse non-null distribution. 
\end{setting}

\begin{figure}[h] 
    \centering
    \begin{subfigure}[t]{0.45\textwidth}
        \centering
        \includegraphics[width=0.8\linewidth]{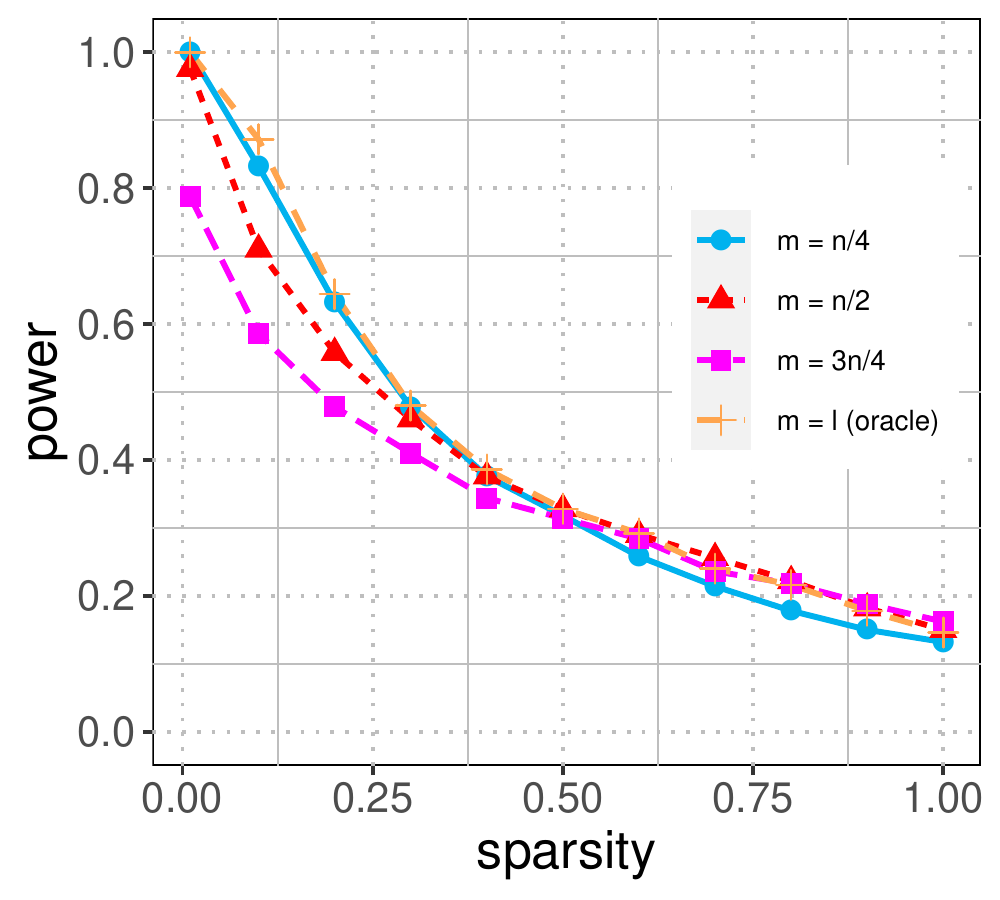}
        \caption{Power of the \mst using the linear bound with different choices of parameter~$m$.}
        \label{fig:power_st_fix}
    \end{subfigure}
    \hfill
    \begin{subfigure}[t]{0.45\textwidth}
        \centering
        \includegraphics[width=0.8\linewidth]{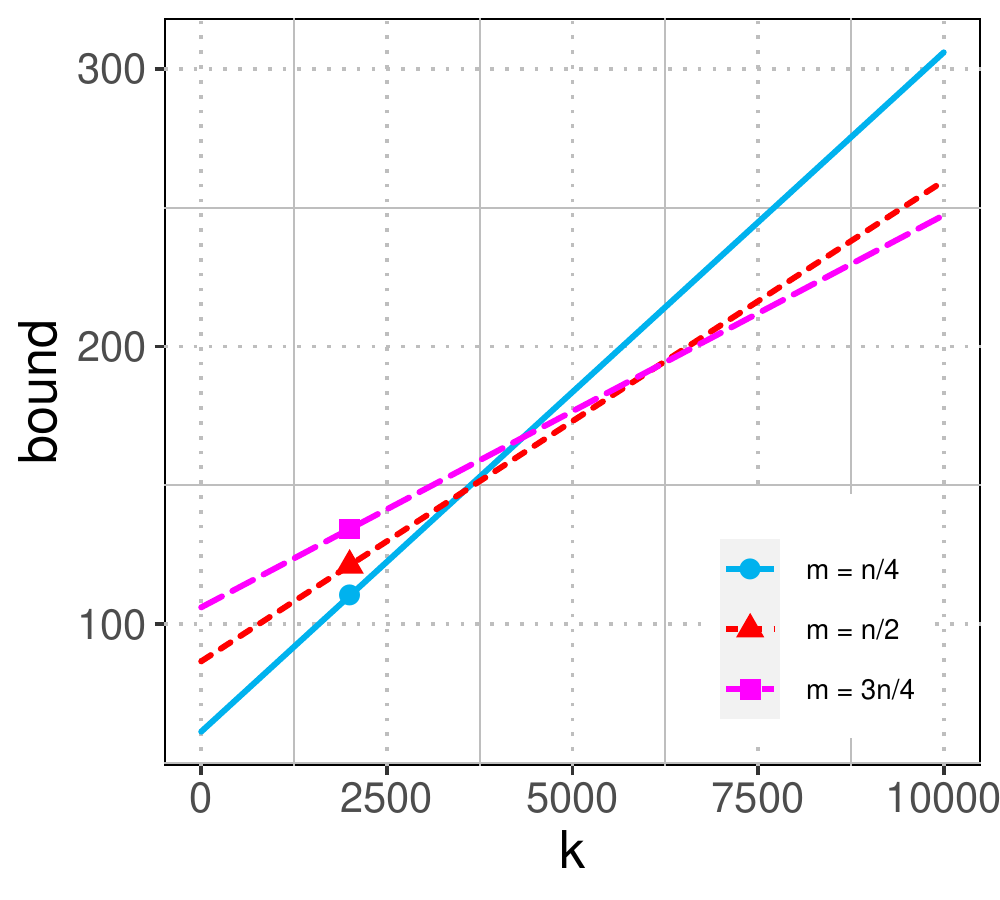}
        \caption{Plot of the linear bound with different choices of parameter~$m$.}
        \label{fig:curve_st_fix}
    \end{subfigure}
    \caption{Testing \mst using linear bound~\eqref{eq:bound_lin} with different choices of parameter $m$ across varying non-null sparsity. The choice $m=n/4$ leads to the highest power.}
    \label{fig:mst_linear}
\end{figure}

We compare three choices of $m = n/4, n/2, 3n/4$, with an oracle benchmark of $m = l$ (whose corresponding bound is the tightest right after all the non-nulls appear). The choice of $m=n/4$ leads to the highest power, which is also close to the oracle benchmark (see Figure~\ref{fig:power_st_fix}).


\paragraph{Choice of the uniform bound}
The four bounds presented above can be generally classified as two types: linear and curved. Curved bounds have a slower increasing rate $O(\sqrt{k\log\log(k)})$ than the linear bound, indicating a tighter bound for large enough $k$, but they are usually looser for small $k$ (Figure~\ref{fig:bd_gua}). 

\begin{figure}[h] 
    \centering
    \begin{subfigure}[t]{0.45\textwidth}
        \centering
        \includegraphics[width=0.8\linewidth]{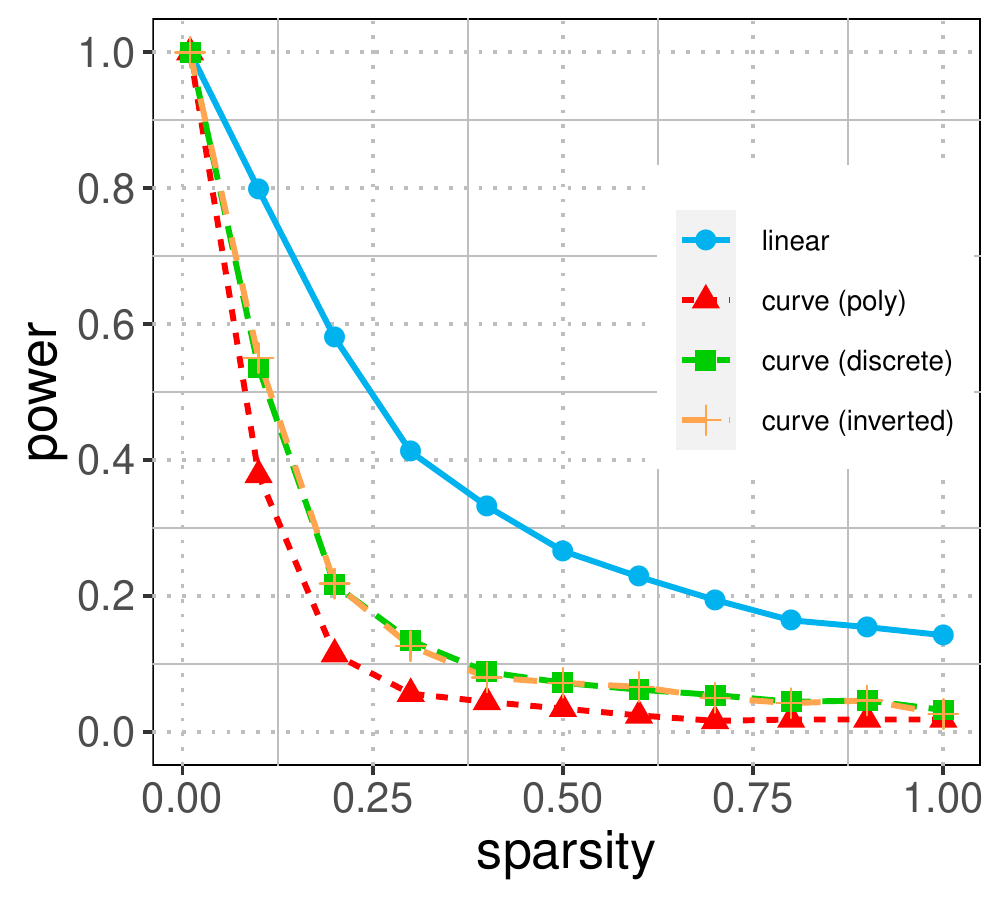}
        \caption{Power of the \mst with varying non-null sparsity.}
        \label{fig:power_gua_s}
    \end{subfigure}
    \hfill
    \begin{subfigure}[t]{0.45\textwidth}
        \centering
        \includegraphics[width=0.8\linewidth]{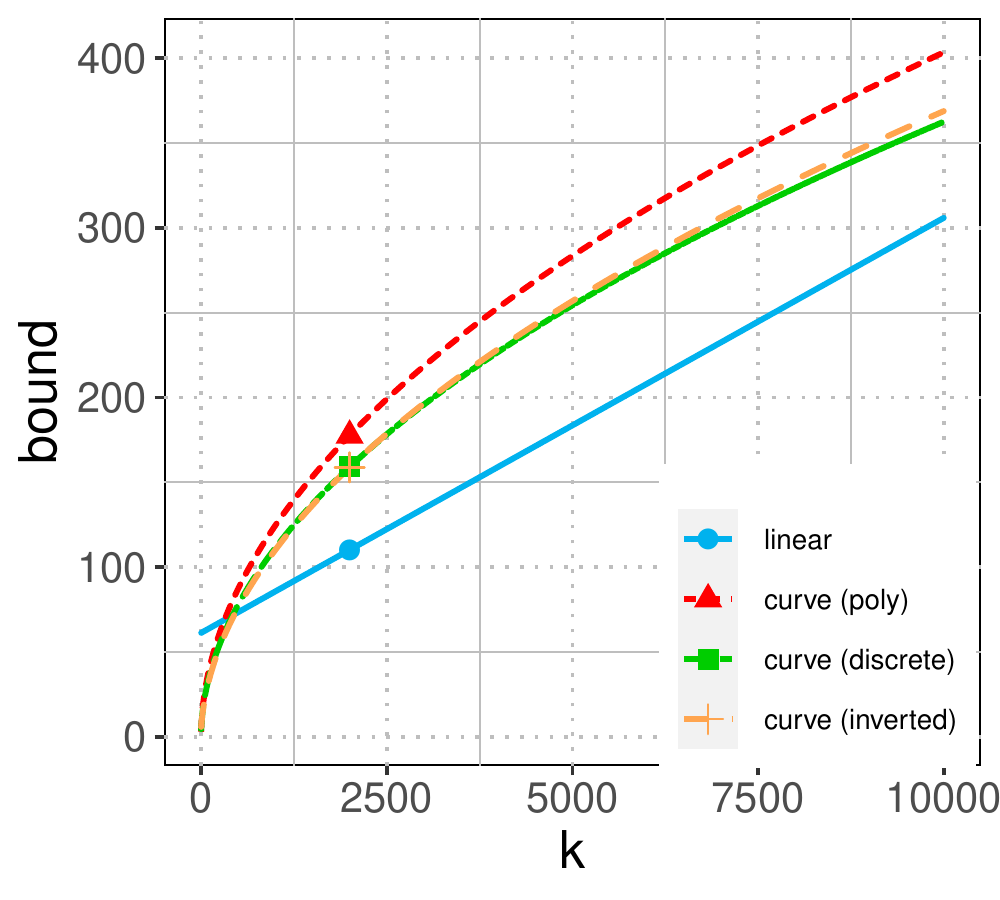}
        \caption{Plot of four bounds. The linear bound is much tighter than the curved bounds for most $k \leq 10^4$.}
        \label{fig:bd_gua}
    \end{subfigure}
    
    \caption{Comparison of the aforementioned four bounds~\eqref{eq:bound_lin}-\eqref{eq:bound_finite} for the \mst.}
\end{figure}

Under the batch setting where the number of hypotheses $n$ is finite, we use the simulation setting~\ref{set:sim_bound}, and the linear bound~\eqref{eq:bound_lin} (with $m = n/4$) results in the highest power (Figure~\ref{fig:power_gua_s}). Similar to tuning the parameter $m$ in the linear bound, we explored to tune the implicit parameters in the curved bound, and yet the linear bound still has the highest power. However, under the online setting where new hypotheses keep arriving, the tests with curved bounds are expected to need less time (number of hypotheses) on average to reach rejection.

\section{Martingale Fisher test} \label{apd:fisher}
The batch test by Fisher~\cite{fisher1992statistical} calculates $S_n = -2\sum_{i=1}^n \log p_i$. Since the distribution of $S_n$ under the global null is $\chi^2_{2n}$ (chi-square with $2n$ degree of freedom), the batch test rejects when $S_n$ is bigger than the $1 - \alpha$ quantile for $\chi^2_{2n}$. To design the martingale test, simply observe that $\{S_k\}_{k\in \mathcal{I}}$ is a martingale whose increments $f(pi) = -2 \log p_i$ are $\chi^2_2$ under the global null (after centering as $S_k - 2k$). Similar to the \mst, there are several types of uniform boundaries $u_\alpha(k)$ for chi-square increment martingales from the work of Howard et al.~\cite{howard2020time, howard2020time1}. We present two types: a sub-exponential (linear) boundary, and a sub-Gamma (curved) boundary. The general form of the martingale Fisher test rejects the global null if
\begin{align}
    \exists k \in \mathbb{N}: -2\sum_{i=1}^k \log p_i - 2k \geq u_\alpha(k),
\end{align}
where examples of $u_\alpha(k)$ include
\begin{enumerate}
    \item a sub-exponential linear boundary
    \begin{align} \label{eq:bound_exp}
        u_\alpha(k) = \left( \left(\frac{1.41m}{x_{m,\alpha}} + 2\right)\log\left(1+\frac{1.41x_{m,\alpha}}{m}\right) - 2\right) (k - m) + 2.82x_{m,\alpha},
    \end{align}
    where $x_{m,\alpha} = \min\left\{x:\ \exp\left\{-0.71x + \frac{m}{2} \log (1 + \frac{1.41x}{m})\right\} \leq \alpha\right\}$; and
    
    \item a sub-Gamma curved boundary
    \begin{align}
    \label{eq:bound_gam}
        u_\alpha(k) =~& 4.07\sqrt{k\left(\log\log(2 k) + 0.72 \log \frac{5.2}{\alpha}\right)}{}\\
        &+ 9.66\left(\log\log(2 k) + 0.72\log \frac{5.2}{\alpha}\right). \nonumber
    \end{align}
\end{enumerate}
The linear bound contains a parameter $m$ with the same interpretation as $m$ in the linear bound~\eqref{test:mst_lin} for \mst (in the main paper): it determines the time at which the bound is tightest --- a larger m results in a lower slope but a larger offset, making the bound loose early on. Based on the simulation results in Figure~\ref{fig:power_fish_fix}, we suggest a default value of $m = n/4$ if the number of hypotheses $n$ is finite, but it should be chosen based on the time by which we expect to have encountered most non-nulls (if any).

\begin{figure}[h!] 
    \centering
    \begin{subfigure}[t]{0.45\textwidth}
        \centering
        \includegraphics[width=0.8\linewidth]{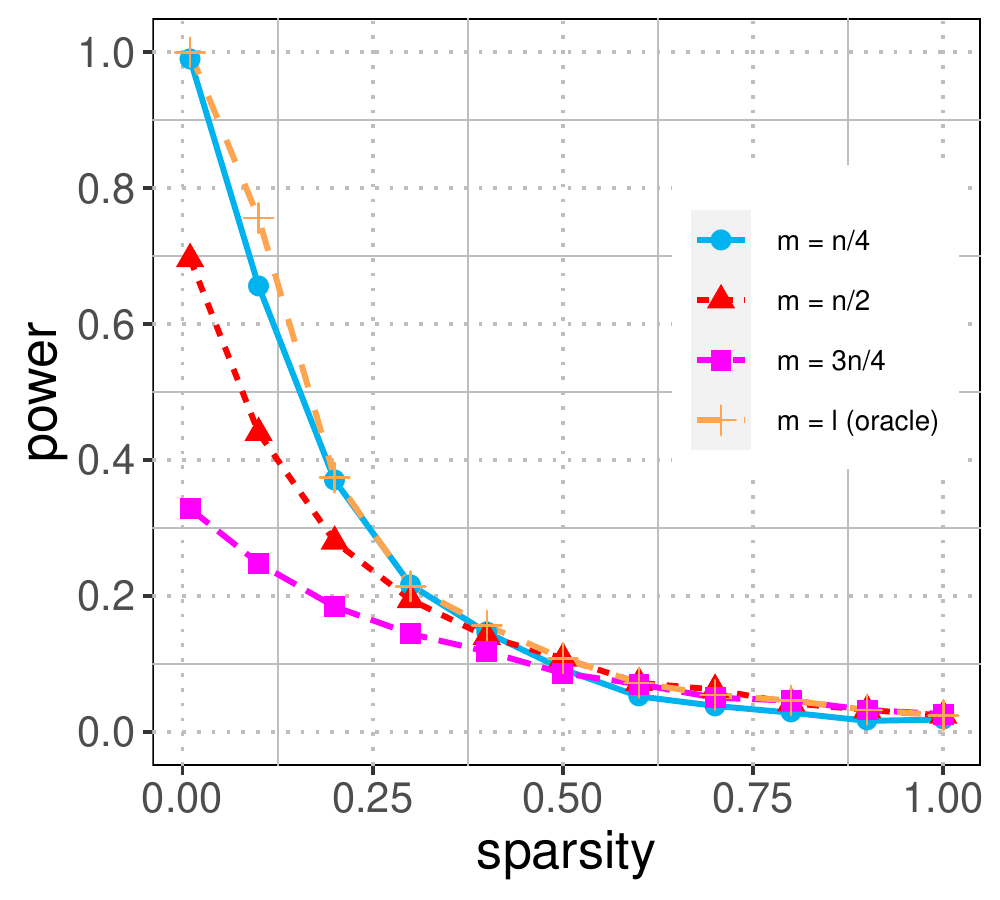}
        \caption{Power of martingale Fisher test using the linear bound with different choices of parameter~$m$.}
        \label{fig:power_fish_fix}
    \end{subfigure}
    \hfill
    \begin{subfigure}[t]{0.45\textwidth}
        \centering
        \includegraphics[width=0.8\linewidth]{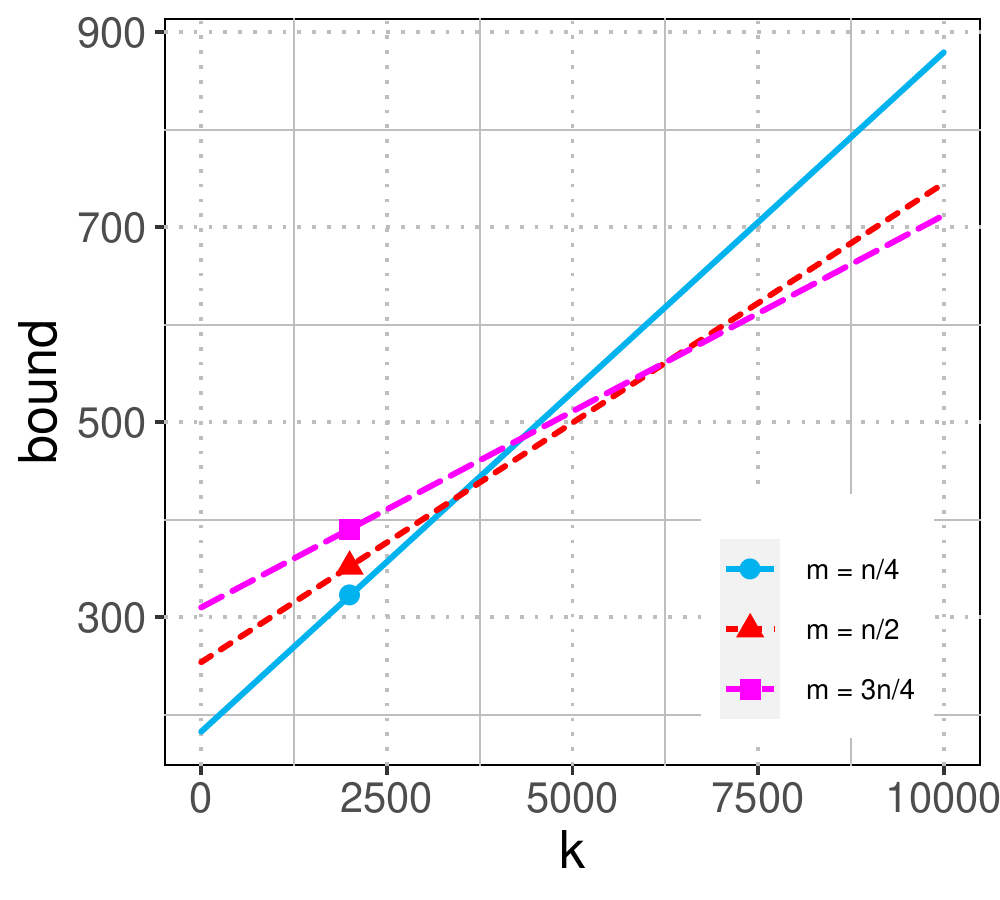}
        \caption{Plot of the linear bound with different choices of parameter~$m$.}
        \label{fig:curve_fish_fix}
    \end{subfigure}
    
    \caption{Testing the martingale Fisher test using the linear bound~\eqref{eq:bound_exp} with different choices of parameter $m$ across varying non-null sparsity. The choice $m=n/4$ leads to the highest power.}
    \label{fig:mft_linear}
\end{figure}

\begin{figure}[h!]
    \centering
    \begin{subfigure}[t]{0.45\textwidth}
        \centering
        \includegraphics[width=0.8\linewidth]{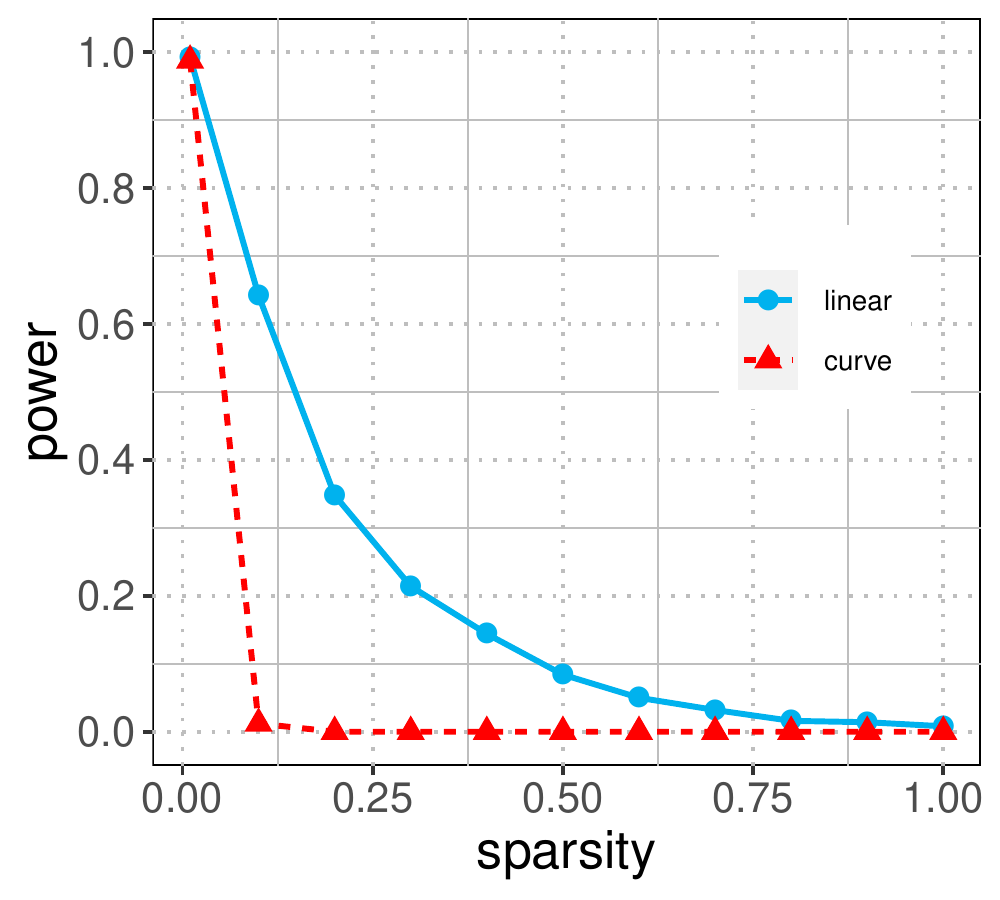}
        \caption{Power of the martingale Fisher test with varying sparsity score.}
        \label{fig:power_mft}
    \end{subfigure}
    \hfill
    \begin{subfigure}[t]{0.4\textwidth}
        \centering
        \includegraphics[width=1\linewidth]{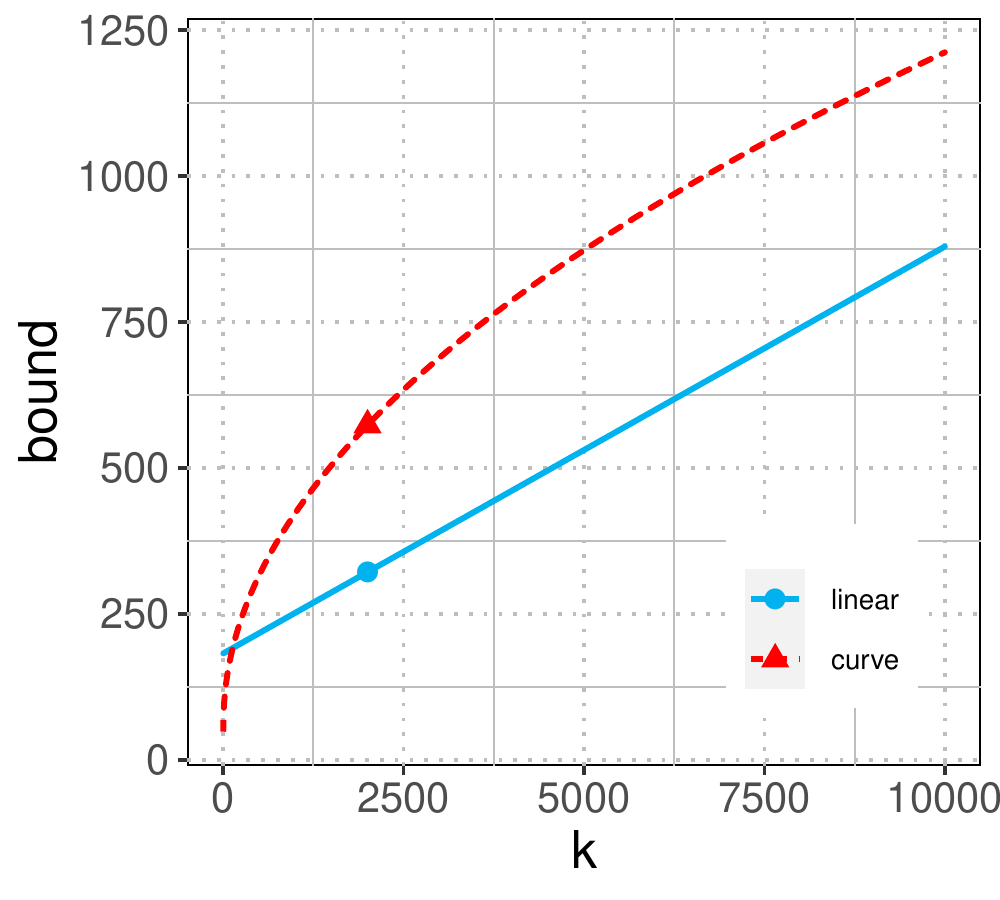}
        \caption{Plot of two bounds. The linear bound ($m = n/4)$ is tighter for most $k \leq 10^4$.}
        \label{fig:curve_mft}
    \end{subfigure}
    \caption{
    Comparison of the aforementioned two bounds~\eqref{eq:bound_exp} and~\eqref{eq:bound_gam} for the martingale Fisher test.
    }
    \label{fig:mft_curve}
\end{figure}

The power of the martingale Fisher test using linear and curved bounds are compared under different non-null sparsity (using simulation setting~\ref{set:sim_bound}). 
The curve bound loses power quickly when non-null is rather sparse (see Figure~\ref{fig:power_mft}), consistent with the comparison between linear and curved bounds for the \mst in Appendix~\ref{apd:stou_bounds}.

\section{Martingale chi-squared test} \label{apd:chi-square}
The chi-squared test calculates $S_n = \sum_{i=1}^n \left[\Phi^{-1}(1 - p_i)\right]^2$. Since the distribution of $S_n$ under the global null is $\chi^2_{n}$ (a chi-square with $n$ degrees of freedom), the batch test rejects when $S_n$ is bigger than the $1 - \alpha$ quantile for $\chi^2_{n}$. To design the martingale test, simply observe that $\{S_k-k\}_{k\in \mathcal{I}}$ is a martingale, whose increment $\left[\Phi^{-1}(1 - p_i)\right]^2-1$ is distributed as~$\chi^2_1$ (minus one) under the global null. Similar to the \mst and martingale Fisher test  (in Appendix~\ref{apd:stou_bounds} and~\ref{apd:fisher}), there are several linear and curved boundaries $u_\alpha(k)$ for chi-square increment martingales from the work of Howard et al.~\cite{howard2020time, howard2020time1}. We present two types: a sub-exponential (linear) boundary, and a sub-Gamma (curved) boundary. The general form of the martingale chi-square test rejects the global null if
\begin{align}
    \exists k \in \mathbb{N}: \sum_{i=1}^k \left[\Phi^{-1}(1 - p_i)\right]^2 - k \geq u_\alpha(k),
\end{align}
where examples of $u_\alpha(k)$ include
\begin{enumerate}
    \item a sub-exponential linear boundary
    \begin{align} \label{eq:bound_exp_chi}
        u_\alpha(k) = \left( \left(\frac{m}{2x_{m,\alpha}} + 1\right)\log\left(1+\frac{2x_{m,\alpha}}{m}\right) - 1\right) (k - m) + 2x_{m,\alpha},
    \end{align}
    where $x_{m,\alpha} = \min\left\{x:\ \exp\left\{-\frac{x}{2} + \frac{m}{4} \log (1 + \frac{2x}{m})\right\} \leq \alpha\right\}$; and
    
    \item a sub-Gamma curved boundary
    \begin{align}
    \label{eq:bound_gam_chi}
        u_\alpha(k) =~& 3.42\sqrt{k\left(\log\log(2 k) + 0.72 \log \frac{5.2}{\alpha}\right)}{}\\
        &+ 9.66\left(\log\log(2 k) + 0.72\log \frac{5.2}{\alpha}\right). \nonumber
    \end{align}
\end{enumerate}
We expect the discussions on parameter~$m$ in the linear bound and on the comparison between the linear and curved bounds to be similar to that in the martingale Stouffer test (Appendix~\ref{apd:stou_bounds}) and the martingale Fisher test (Appendix~\ref{apd:fisher}). If testing the martingale chi-squared test by the same numerical experiment in Setting~\ref{set:sim_bound}, $m = n/4$ should lead to high power for various degrees of sparsity; and the linear bound should be tighter than the curved bound for most time $k \leq 10^4$, and hence lead to higher power when non-null is rather sparse.
\color{black}

\section{Bayesian modeling for the posterior probability of being non-null} \label{apd:em}

\paragraph{Modeling the posterior probabilities of being non-null} 
Define the $Z$-score for hypothesis $H_i$ be $Z_i = \Phi^{-1}(1 - p_i)$. Instead of modeling the \pvalues, we choose to model the $Z$-scores since under setting~\ref{set:simple} in the main paper they are distributed as a Gaussian either under the null or the alternative:
\[
H_0: Z_i \sim N(0,1) \text{ versus } H_1: Z_i \sim N(\mu,1),
\]
where $\mu$ is the mean value for all the non-nulls. We model $Z_i$ by a mixture of Guassians:
\[
Z_i \sim (1 - q_i) N(0,1) + q_i N(\mu,1), \text{ with } q_i \sim \mathrm{Bernoulli}(\pi_i),
\]
where $q_i$ is the indicator of whether the hypothesis $H_i$ is a true non-null.

The non-null structures are imposed by the constraints on non-null probability $\pi_i$. In our examples, the blocked non-null structure is encoded by fitting non-null probabilities $\pi_i$ as a smooth function of the hypothesis position (covariates) $x_i$, specifically as a logistic regression model on a spline basis:
\begin{align} \label{eq:p_model_block}
    \pi_i = \pi_\beta(x_i) = \frac{1}{1 + \exp(- \beta \phi(x_i))},
\end{align}
where $\phi(x_i)$ is a spline basis. The hierarchical structure is imposed by a partial ordering constraint on $\pi_i$:
\begin{align} \label{eq:p_model_tree}
    \pi_i \geq \pi_j, \quad \text{if $i$ is the parent of $j$},
\end{align}
when the probability of being non-null decreases down the tree ($\pi_i \geq \pi_j$ if the probability increases).

\paragraph{An EM framework for the posterior probabilities of being non-null} 
An EM algorithm is used to train the model because masked \pvalues are modeled. Specifically, we treat \pvalues as the hidden variables, and the masked \pvalues $g(p)$ as observed. In terms of the Z-score $Z_i$, $Z_i$ is a hidden variable and the observed variable $\widetilde{Z_i}$ is its absolute value $|Z_i|$ (if $p_i$ is masked).

Define a sequence of hypothetical labels $w_i = \one(Z_i = \widetilde{Z_i})$, and the likelihood of data $(\widetilde{Z_i}, w_i, q_i)$ is
\begin{align*}
    l(\widetilde{Z_i}, w_i, q_i) =~& w_iq_i\log(\pi_i\phi(\widetilde{Z_i} - \mu)) + w_i(1-q_i)\log((1 - \pi_i)\phi(\widetilde{Z_i})) {}\\
    +~& (1 - w_i)q_i\log(\pi_i\phi(-\widetilde{Z_i} - \mu)){}\\
    +~& (1 - w_i)(1-q_i)\log((1 - \pi_i)\phi(-\widetilde{Z_i})),
\end{align*}
where $\phi(\cdot)$ is the PDF of a standard Gaussian. 

The E-step updates $w_i,q_i$. Notice that $w_i$ and $q_i$ are not independent, so we update the joint distribution of $(w_i, q_i)$, namely parameters
\[
w_iq_i =: a_i, \quad w_i(1 - q_i) =: b_i, \quad (1 - w_i)q_i =: c_i, \quad (1 - w_i)(1 - q_i) =: d_i,
\]
where $a_i + b_i + c_i + d_i = 1$. For a simple expression of the updates, we define
\begin{align*}
    L\left(\widetilde{Z_i}, \mu, \pi_i\right) :=~& \pi_i\phi(\widetilde{Z_i} - \mu)) + (1 - \pi_i)\phi(\widetilde{Z_i}) {}\\
    &~+ \pi_i\phi(-\widetilde{Z_i} - \mu) + (1 - \pi_i)\phi(-\widetilde{Z_i}).
\end{align*}
For hypothesis $i$ whose $p$-value is masked, the updates are 
\begin{align*}
    a_{i,\text{new}} =~& \mathbb{E}[w_iq_i \mid \widetilde{Z_i}] = \frac{\pi_i\phi(\widetilde{Z_i} - \mu)}{L\left(\widetilde{Z_i}, \mu, \pi_i\right)};{}\\
    b_{i,\text{new}} =~& \mathbb{E}[w_i(1 - q_i) \mid \widetilde{Z_i}] = \frac{(1 - \pi_i)\phi(\widetilde{Z_i})}{L\left(\widetilde{Z_i}, \mu, \pi_i\right)};{}\\
    c_{i,\text{new}} =~& \mathbb{E}[(1 - w_i)q_i \mid \widetilde{Z_i}] = \frac{\pi_i\phi(-\widetilde{Z_i} - \mu)}{L\left(\widetilde{Z_i}, \mu, \pi_i\right)};{}\\
    d_{i,\text{new}} =~& \mathbb{E}[(1 - w_i)(1 - q_i) \mid \widetilde{Z_i}] = \frac{(1 - \pi_i)\phi(-\widetilde{Z_i})}{L\left(\widetilde{Z_i}, \mu, \pi_i\right)}.
\end{align*}
If the $p$-value is unmasked for hypothesis $i$, we have $w_i = 1$ and the updates are
\begin{align*}
    a_{i, \text{new}} =~& \left(1 + \frac{(1-\pi_i)\phi\left(\widetilde{Z_i}\right)}{ \pi_i\phi\left(\widetilde{Z_i}-\mu\right)}\right)^{-1};{}\\
    b_{i, \text{new}} =~& 1 - a_{i, \text{new}}; \quad c_{i, \text{new}} = 0; \quad d_{i, \text{new}} = 0.
\end{align*}

In the M-step, parameters $\mu$ and $\pi_i$ are updated. The update for $\mu$ is
\begin{align*}
    \mu_{\text{new}} = \argmin_\mu  \sum_i l(\widetilde{Z_i}) = \frac{\sum (a_i - c_i) \widetilde{Z_i}}{\sum (a_i + c_i)}.
\end{align*}
The update for $\pi_i$ depends on the non-null structure, which encodes different constraints on $\pi_i$. Under the block non-null structure, updating $\pi_i$ corresponds to updating $\beta$ in model~\eqref{eq:p_model_block} for $\pi_\beta(x_i)$. The update is equivalent to fitting $a_i + c_i$ by a logistic regression:
\begin{align*}
    (\beta_{\text{new}}) =~& \argmax_\beta \sum_i (a_i + c_i)\log \pi_\beta(x_i) + (b_i + d_i) \log(1 - \pi_\beta(x_i)){}\\
    =~& \argmax_\beta \sum_i (a_i + c_i)\log \pi_\beta(x_i) + (1 - a_i - c_i) \log(1 - \pi_\beta(x_i)),
\end{align*}
and $\pi_{i,\text{new}} = \pi_{\beta_{\text{new}}}(x_i)$. Under the hierarchical structure, updating $\pi_i$ is equivalent to fitting a partial isotonic regression on $a_i + c_i$ (Barlow~\cite{barlow1972isotonic}, Theorem~3.1 and Robertson~\cite{robertson1988order}, Theorem 1.5.1):
\begin{align*}
    (\pi_{i,\text{new}}) =~& \argmax_{\text{partial ordered} \{\pi_i\}} \sum_i (a_i + c_i)\log \pi_i + (1 - a_i - c_i) \log(1 - \pi_i){}\\
    =~& \argmin_{\text{partial ordered} \{\pi_i\}} \sum_i (a_i + c_i - \pi_i)^2,
\end{align*}
where the partial ordering is defined in statement~\eqref{eq:p_model_tree}.

Suppose we wish to model the alternative mean $\mu$ differently for individual hypotheses. In that case, we can think of the alternative mean as a parametric function of the covariates: $\mu_i = \mu_\gamma(x_i)$ where the vector $\gamma$ denotes the parameters. A simple example is a linear function: $\mu_\gamma(x_i) = \gamma^T x_i$. The updates in the E-step is the same as above with $\mu$ replaced by $\mu_\gamma(x_i)$. In the M-step, the update for~$\mu_i$ corresponds to the update for $\gamma$: 
\begin{align*}
    (\gamma_{\text{new}}) = \argmax_\gamma \sum_i a_i \left(\widetilde{Z_i} - \mu_\gamma(x_i)\right)^2 + c_i \left(-\widetilde{Z_i} - \mu_\gamma(x_i)\right)^2,
\end{align*}
which is equivalent to the solution of a least square regression to a set of pseudo responses $\{\widetilde{Z_1}, \ldots, \widetilde{Z_n},  -\widetilde{Z_1}, \ldots -\widetilde{Z_n}\}$ with weights $\{a_1, \ldots, a_n, c_1, \ldots, c_n\}$. We use the EM algorithm with constant $\mu$ for the experiments in our paper, because it tends to be robust to heterogeneous alternative mean values in simulations.

\section{Comparison with alternative methods} \label{apd:alter_methods}
We compared the interactive test with the adaptive weighted Fisher test (AW-Fisher) and weighted Higher Criticism (weighted-HC) in the example of a grid of hypotheses. Our simulation considers a small grid ($10 \times 10$) because the AW-Fisher test has a very high computational cost. We used the R package \texttt{AWFisher} by Huo et al. (2020) \cite{huo2020p}, which refers to a base library of null distributions for cases with less than 100 hypotheses; it took 6373.5 CPU hours using AMD Opteron(tm) Processor (1.4GHz) to complete the base library. Without such a base library, the computational complexity of the AW-Fisher test is $\mathcal{O}(2^n)$, and roughly $\mathcal{O}(n \log(n))$ for our interactive test.

\begin{figure}[h]
    \centering
    \includegraphics[width=0.5\linewidth]{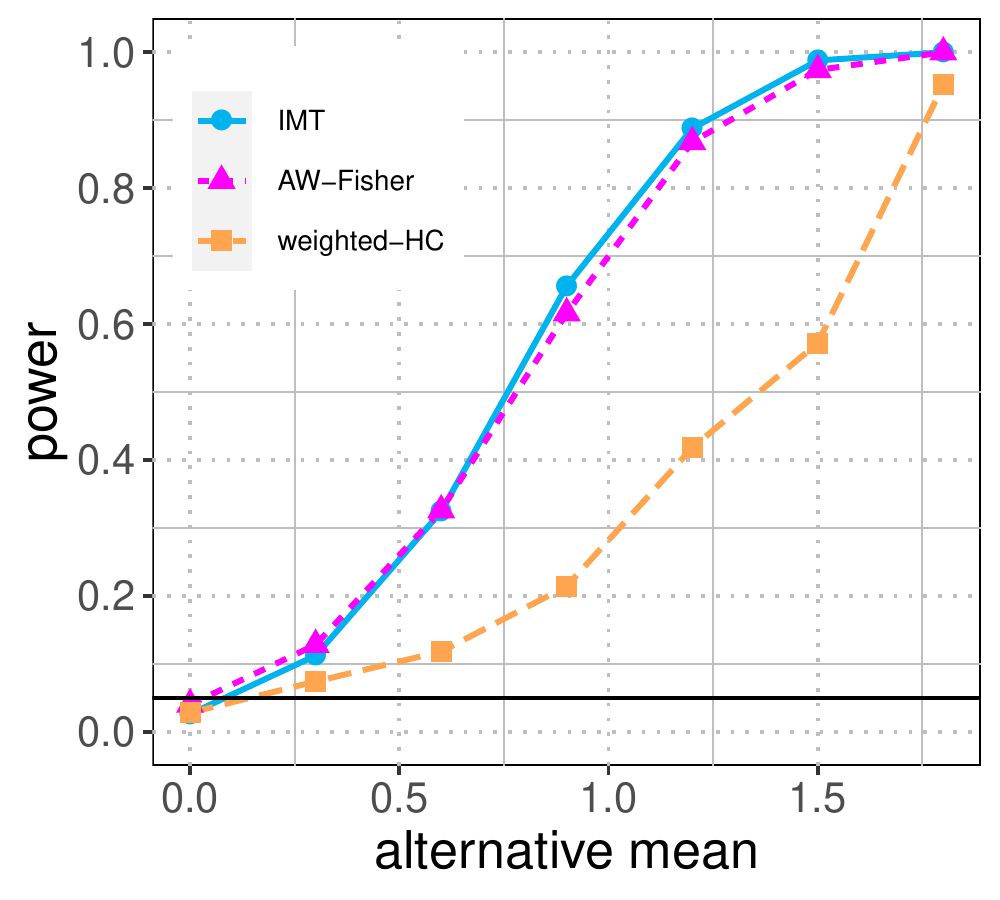}
    \caption{Power of the \imt (IMT), AW-Fisher, and weighted-HC when the non-null cluster is in the center of a $10 \times 10$ grid. IMT and AW-Fisher both have high power, but the AW-Fisher has a high computational cost.}
    \label{fig:alternative_methods}
\end{figure}

As described in Section~\ref{sec:sim_cluster}, we simulated a non-null cluster is in the center of the hypothesis grid. The weights in HC use the oracle information of the non-null position and is set to $1$ for the non-nulls and $0.5$ for others. Since we have included several simulations to compare the \imt with \mst and Stouffer's test in Section~\ref{sec:ipm}, above in Figure~\ref{fig:alternative_methods}, we only focus on the comparison among the interactive test, AW-Fisher and weighted-HC. Although the AW-Fisher test achieves similar power as the \imt, it has very high computational cost as described above. Also, we remark that one main advantage of the interactive test we propose is that it can incorporate various types of prior knowledge and covariates in a data-dependent way. Meanwhile, most existing methods require the analyst to commit to one structure or prior knowledge before observing the $p$-values. For example, the weighted-HC might achieve higher power with a different set of weights, but the weights need to be specified ahead of time.

\end{document}